\documentclass[11pt]{article}
\usepackage{preamble}
\usepackage[margin=1in]{geometry}
\title{The Complexity of NISQ}
\mathtoolsset{showonlyrefs}

\usepackage{appendix}
\newcommand{\Ibal}{I_{\mathsf{bal}}}
\newcommand{\Iimbal}{I_{\mathsf{imbal}}}
\newcommand{\Igood}{I_{\mathsf{good}}}
\newcommand{\Ibad}{I_{\mathsf{bad}}}
\newcommand{\Iconc}{I_{\mathsf{conc}}}
\newcommand{\Idiv}{I_{\mathsf{anticonc}}}

\newcommand{\tr}{\mathsf{tr}}
\newcommand{\ketbra}[2]{\lvert #1 \rangle \! \langle #2 \rvert}
\renewcommand{\epsilon}{\varepsilon}

\author{Sitan Chen\thanks{Email: \texttt{sitanc@berkeley.edu} This work is supported by NSF Award 2103300.} \\
UC Berkeley
 \and 
Jordan Cotler\thanks{Email: \texttt{jcotler@fas.harvard.edu}. This work is supported by
a Junior Fellowship from the Harvard Society of Fellows, as well as in part by the Department of
Energy under grant DE-SC0007870.} \\
Harvard University
\and
Hsin-Yuan Huang\thanks{Email: \texttt{hsinyuan@caltech.edu}. This work is supported by a Google Ph.D. fellowship.} \\
Caltech
\and
Jerry Li\thanks{Email: \texttt{jerrl@microsoft.com}.} \\
Microsoft Research
}

\newcommand{\NISQ}{\mathsf{NISQ}}
\newcommand{\BQP}{\mathsf{BQP}}
\newcommand{\BPP}{\mathsf{BPP}}
\newcommand{\QNC}{\mathsf{QNC}}
\newcommand{\SSP}[1]{#1\text{-}\mathsf{SSP}}
\newcommand{\id}{\mathrm{id}}
\renewcommand{\Id}{I}

\usepackage[titles]{tocloft}

\begin{document}
\pagestyle{empty}
{
  \renewcommand{\thispagestyle}[1]{}
  \maketitle

    \begin{abstract}
    The recent proliferation of NISQ devices has made it imperative to understand their computational power.
    In this work, we define and study the complexity class $\NISQ$, which is intended to encapsulate problems that can be efficiently solved by a classical computer with access to a NISQ device.
    To model existing devices, we assume the device can (1) noisily initialize all qubits, (2) apply many noisy quantum gates, and (3) perform a noisy measurement on all qubits.
    We first give evidence that $\BPP \subsetneq \NISQ \subsetneq \BQP$, by demonstrating super-polynomial oracle separations among the three classes, based on modifications of Simon's problem.
    We then consider the power of $\NISQ$ for three well-studied problems.
    For unstructured search, we prove that $\NISQ$ cannot achieve a Grover-like quadratic speedup over $\BPP$.
    For the Bernstein-Vazirani problem, we show that $\NISQ$ only needs a number of queries logarithmic in what is required for $\BPP$.
    Finally, for a quantum state learning problem, we prove that $\NISQ$ is exponentially weaker than classical computation with access to noiseless constant-depth quantum circuits.
    
    \end{abstract}
    
    \newpage

    \tableofcontents
}

\clearpage
\pagestyle{plain}
\pagenumbering{arabic}

\section{Introduction}

Fault-tolerant quantum computing promises to offer speed-ups to various computational problems, including simulating quantum systems \cite{lloyd1996universal,georgescu2014quantum,o2016scalable,babbush2018low,low2019hamiltonian}, factoring large numbers \cite{shor1994algorithms,shor1999polynomial}, performing optimization \cite{crosson2016simulated,farhi2016quantum,brandao2017quantum1,brandao2017quantum2}, and solving linear systems of equations \cite{harrow2009quantum}.
While a fault-tolerant quantum computer has not yet been built, recent technological advancements have resulted in the development of new quantum devices that can outperform the best classical supercomputers for certain artificial computational tasks \cite{arute2019quantum}.
The opportunities offered by these devices and the challenges they engender have invigorated theorists and experimentalists alike, ushering in a new age of research into quantum computation and information often referred to as the NISQ (noisy intermediate-scale quantum) era \cite{preskill2018quantum}.

Computation in the NISQ era is modeled by hybrid computation consisting of a classical computer and a noisy quantum device.
Most currently available noisy quantum devices, such as those built from superconducting qubits \cite{arute2019quantum, clarke2008superconducting, reagor2018demonstration, havlivcek2019supervised}, trapped ions \cite{debnath2016demonstration,zhang2017observation,friis2018observation,wright2019benchmarking}, nuclear spins in silicon \cite{Noiri2022Silicon, Xue2022Silicon, Madzik2022Silicon}, or other solid-state systems \cite{weber2010quantum,doherty2013nitrogen,choi2017observation}, are restricted to preparing a noisy initial state, performing noisy quantum gates, and executing a noisy measurement on all qubits to generate a random bitstring.
In most circumstances these noisy quantum devices are too weak to perform useful computation on their own.
Generally, useful computation is facilitated by a classical computer that repeatedly runs the noisy quantum device with various gate sequences to obtain different classical output bit strings, and then performs classical post-processing on those strings.
Algorithms within this computational model are referred to as hybrid quantum-classical algorithms \cite{mcclean2016theory,endo2021hybrid} or NISQ algorithms \cite{preskill2018quantum, RevModPhys.94.015004}.

The excitement around NISQ algorithms has led to a plethora of new near-term algorithms targeting different applications, including quantum chemistry \cite{peruzzo2014variational,o2016scalable,mcclean2016theory,kandala2017hardware,huggins2022unbiasing}, machine learning \cite{havlivcek2019supervised,farhi2018classification,schuld2019quantum,huang2021power,huang2021information,caro2022generalization,cerezo2022QML}, combinatorial optimization \cite{farhi2014quantum, hadfield2019quantum, zhou2020quantum, basso2022performance}, linear system solvers \cite{bravo2019variational,xu2021variational,huang2019near}, and experimental data analysis \cite{aharonov2022quantum,chen2021exponential,huang2022foundations,huang2022quantum, schuster2022learning, cotler2022information}.
However, to the best of our knowledge, no existing works have rigorously examined the class of \emph{all possible} NISQ algorithms and studied their inherent computational power.
As a result, many important and fundamental questions remain unanswered.
In particular, how powerful are NISQ algorithms compared to classical algorithms?
Are NISQ algorithms inherently weaker than fault-tolerant quantum algorithms?

\paragraph{Defining the $\NISQ$ complexity class.}
In this work, we formalize and study these basic questions through the lens of computational complexity theory.
To do so rigorously, we define a complexity class which we believe encapsulates what is possible on the vast majority of existing quantum devices.
Our definition is intended to capture the following capabilities of noisy quantum devices, which we alluded to above:
\begin{enumerate}[leftmargin=*,itemsep=0pt]
    \item {\bf Noisy quantum gates.} The device can execute noisy two-qubit logic gates. Using quantum logic gates (as opposed to, say, more general non-unitary CPTP maps) is standard in existing quantum devices, and it is well-understood that in real-world settings, they will be subject to noise. For concreteness, we consider the standard model of local depolarizing noise per qubit.
    However, our results extend to more general noise models; see Remarks~\ref{remark:noise}, \ref{remark:superpoly} and~\ref{remark:BV}.
    \item {\bf Noisy state preparation at the start.}
    The quantum devices have a fixed number of qubits and as such cannot bring in fresh qubits during the computation.  This means that the device must prepare all qubits at the start.
    Notably, since we assume all quantum gates are subject to noise, this means all qubits will accrue entropy throughout the computation.
    \item {\bf Noisy measurement at the end.} 
    The quantum devices are limited to perform noisy measurements only at the end of the computation, which means the measurement is performed on all qubits simultaneously.
    From a physical perspective, this constraint arises due to the difficulty of isolating subsets of qubits and measuring them without decohering the residual qubits.
\end{enumerate}
Finally, we consider a classical computer that can repeatedly run the noisy quantum device and analyze the output from the noisy quantum device.

These constraints are chosen to encapsulate the gap between the physical limitations of what we can achieve with existing quantum computers, and general quantum computation.
We note these considerations preclude the implementation of all known general fault-tolerant quantum computation schemes \cite{aharonov1997fault,preskill1998fault,bacon2000universal,raussendorf2007fault,ben2013quantum,fawzi2020constant,divincenzo2007effective}, but that removing any one of these constraints would already allow for some form of nontrivial quantum fault tolerance \cite{aharonov1997fault, ben2013quantum, fawzi2020constant}.
The obstruction to fault tolerance can be understood intuitively.
The noisy quantum gates cause all qubits to accrue entropy, which cannot be pumped out until the measurement at the end.
Since too much entropy would destroy all useful quantum correlations, it is not possible for the noisy quantum devices under the above constraints to perform an arbitrarily long quantum computation.

Motivated by the above considerations, in Section~\ref{sec:def-NISQ-intro} we formally define the $\NISQ$ complexity class to be the set of all problems that can be efficiently solved by a classical computer with access to a noisy quantum device that can (i) prepare a noisy $\mathrm{poly}(n)$-qubit all-zero state, (ii) execute noisy quantum gates, and (iii) perform a noisy measurement on all of the $\mathrm{poly}(n)$ qubits.

\begin{figure*}[t]
\centering
\includegraphics[width=1.0\textwidth]{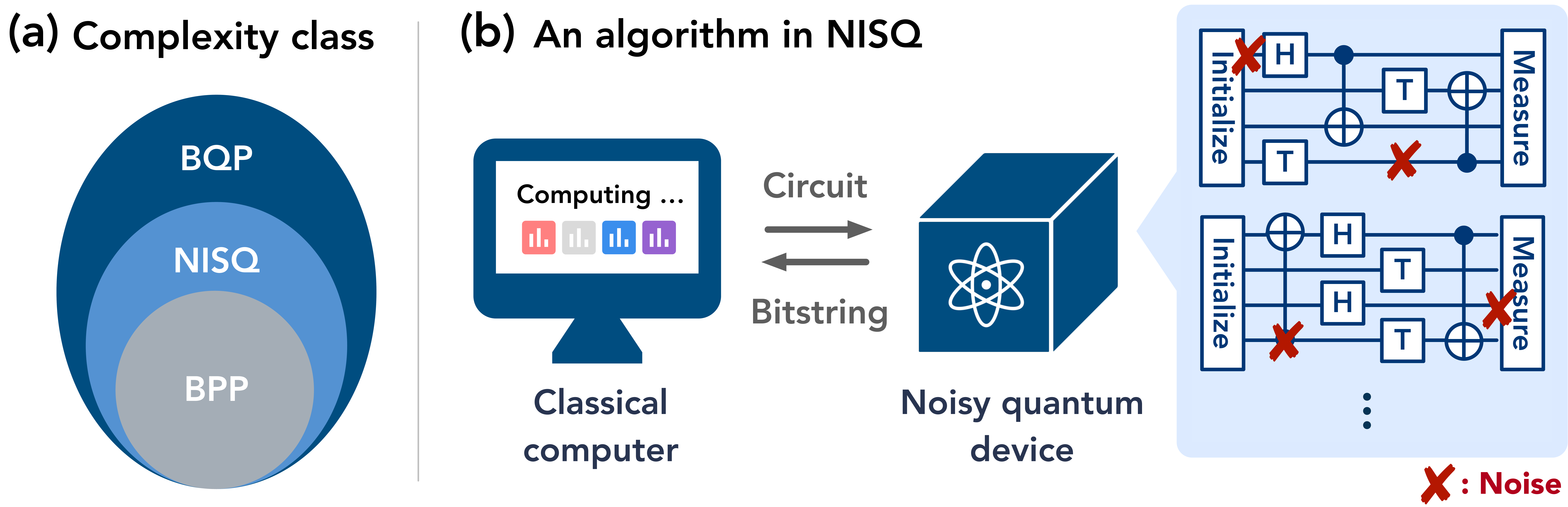}
    \caption{
    \emph{Illustration of the $\NISQ$ complexity class:} (a) Complexity classes: $\NISQ$ contains problems that can be solved by classical computation ($\BPP$), and some problems that can be solved by quantum computation ($\BQP$). (b) $\NISQ$ algorithm: An algorithm in the complexity class $\NISQ$ is modeled by a hybrid quantum-classical algorithm, where a classical computer can specify the circuit to run on a noisy quantum device and the device would run a noisy version of the circuit and return a random classical bitstring obtained from noisy measurement.
    \label{fig:NISQ}}
\end{figure*}

\section{Main Results}

In Section~\ref{sec:def-NISQ-intro} we give an overview of the definition of $\NISQ$.
Then, in Section~\ref{sec:super-poly-intro}, we give two modifications of Simon's problem which respectively yield a super-polynomial separation between $\BPP$ and $\NISQ$, and an exponential separation between $\NISQ$ and $\BQP$.
In Section~\ref{sec:three-intro}, we study the NISQ complexity of three well-known problems: unstructured search, Bernstein-Vazirani problem, and shadow tomography.
We defer all technical details to the appendix.

\subsection{Definition of \texorpdfstring{$\NISQ$}{NISQ}}
\label{sec:def-NISQ-intro}
We begin with the definition of the $\NISQ$ complexity class.
The formal mathematical definition is given in Appendix~\ref{subsec:nisq}.
In Appendix~\ref{subsec:oracle}, we recall the standard definition of oracle access in classical and quantum computation \--- when the oracle is classical, we extend the definition to $\NISQ$ by giving oracle access to both the classical computer and the noisy quantum computer.

\begin{definition}[$\NISQ$ complexity class, informal]
$\NISQ$ contains all problems that can be solved by a polynomial-time probabilistic classical algorithm with access to a noisy quantum device.
To solve a problem of size $n$, the classical algorithm can access a noisy quantum device that can:
\begin{enumerate}
    \setlength\itemsep{0em}
    \item Prepare a noisy $\mathrm{poly}(n)$-qubit all-zero state;
    \item Apply arbitrarily many layers of noisy two-qubit gates;
    \item Perform a noisy computational basis measurements on all the qubits simultaneously.
\end{enumerate}
All quantum operations are subject to a constant amount of depolarizing noise per qubit.
\end{definition}

\noindent The definition of a noisy quantum device given above forbids the implementation of all known fault-tolerant quantum computation schemes \cite{aharonov1997fault,preskill1998fault,bacon2000universal,raussendorf2007fault,divincenzo2007effective}.
Hence, it is plausible that there are problems that could be solved efficiently by a fault-tolerant quantum algorithm but not by a $\NISQ$ algorithm, i.e.~that $\NISQ \subsetneq \BQP$.

The definition of $\NISQ$ immediately gives us that
\begin{equation}
\BPP \subseteq \NISQ \subseteq \BQP\,.
\end{equation}
The inclusion $\BPP \subseteq \NISQ$ follows from the fact that a $\NISQ$ algorithm is a hybrid quantum-classical algorithm that can also run any classical computation.
The latter inclusion $\NISQ \subseteq \BQP$ holds because a quantum computer can simulate any noisy quantum device.
However, it remains an open question whether $\BPP \subsetneq \NISQ$ and also if $\NISQ \subsetneq \BQP$.
The strict inclusion $\BPP \subsetneq \NISQ$ would imply that $\NISQ$ algorithms have a super-polynomial speedup over classical algorithms.
The second strict inclusion $\NISQ \subsetneq \BQP$ would imply that $\NISQ$ algorithms are not as powerful as fault-tolerant quantum algorithms.
Establishing either one of the strict inclusions would imply that $\BPP \subsetneq \BQP$, which is a long-standing open problem in computational complexity theory.



\subsection{Super-polynomial separations}
\label{sec:super-poly-intro}

We begin with a set of results involving two modified versions of Simon's problem \cite{simon1997power}.
Simon's problem is one of the earliest examples of a computational problem demonstrating exponential quantum advantage in query complexity: given a function $f: \{0, 1\}^n \mapsto \{0, 1\}^n$, the goal is to  decide if $f$ is $2$-to-$1$ (with a promise that $f(x) = f(x \oplus s)$ for a secret string $s \in \{0, 1\}^n$) or $1$-to-$1$.
In the language of computational complexity theory, it exhibits a classical oracle $O$ for which there is a relativized complexity separation between $\BPP$ and $\BQP$, denoted by~$\BPP^O \subsetneq \BQP^O$.

In this work, we give two modifications of Simon's problem that yield relativized separations among $\BPP$, $\NISQ$, and $\BQP$.
We first construct a modification of Simon's problem that requires at least a super-polynomial number of oracle queries for $\BPP$ and only a linear number for $\NISQ$.
We then give another modification of Simon's problem that requires at least an exponential number of queries for $\NISQ$ and only a linear number of queries for $\BQP$.
These two results can be summarized as follows:

\begin{theorem}
\label{thm:superpoly1}
There is a classical oracle $O_1$ such that $\BPP^{O_1} \subsetneq \NISQ^{O_1}$.
\end{theorem}
\begin{theorem}
\label{thm:superpoly2}
There is a classical oracle $O_2$ such that $\NISQ^{O_2} \subsetneq \BQP^{O_2}$.
\end{theorem}

\noindent Theorem~\ref{thm:superpoly1} is established by constructing a robust version of the classical function $f: \{0, 1\}^n \mapsto \{0, 1\}^n$ in Simon's problem.
We denote the robust version as $\wt{f}: \{0, 1\}^{n'} \mapsto \{0, 1\}^n$ with more input bits: $n' \gg n$.
Each input $x \in \{0, 1\}^n$ corresponds to a large set of inputs $A_x \subseteq \{0, 1\}^{n'}$, such that every $z \in A_x$ produces the same output $\wt{f}(z) = f(x)$.
The new function $\wt{f}$ is robust to noise, which allows a $\NISQ$ algorithm to achieve a super-polynomial speed-up.

The proof of Theorem~\ref{thm:superpoly2} is essentially the opposite of that of Theorem~\ref{thm:superpoly1}. We construct a highly non-robust version $\wt{f}$ of the classical function $f$ in Simon's problem.
Any noisy access to $\wt{f}$ provides exponentially little information, hence any $\NISQ$ algorithm would require exponentially more queries than a noiseless quantum algorithm. In fact, the separation we show is even stronger: not only is $\NISQ$ exponentially weaker than $\BQP$ relative to this oracle, but it is in fact even exponentially weaker than $\BPP^{\QNC^0}$, that is, classical computation assisted by noiseless bounded-depth quantum computation (see Section~\ref{sec:boundeddepth} for definitions).

Our findings give evidence that the computational power of $\NISQ$ lies somewhere strictly between $\BPP$ and $\BQP$.

\subsection{\texorpdfstring{$\NISQ$}{NISQ} in three well-studied problems}
\label{sec:three-intro}

Since the modified Simon's problems in Theorems~\ref{thm:superpoly1} and~\ref{thm:superpoly2} are tailored for proving super-polynomial separations, we would also like to study $\NISQ$ for more natural problems.
We explore three well-known problems in quantum computing: unstructured search, the Bernstein-Vazirani problem, and shadow tomography.

For unstructured search, it is well-known that Grover's algorithm \cite{grover1996fast} can achieve a quadratic quantum speedup over any classical algorithm.
To find a single marked element among $N$ elements, Grover's algorithm only requires $\mathcal{O}(\sqrt{N})$ queries, whereas any classical algorithm requires at least $\Omega(N)$ queries in the worst case.
A natural open question is whether a $\NISQ$ algorithm can also achieve such a quadratic speedup.
We resolve this open question in the negative by proving the following theorem.
\begin{theorem}[Unstructured search]\label{thm:grover_informal}
    To find a single marked element among $N$ elements,
    any $\NISQ$ algorithm with access to $\mathrm{poly}(N)$ qubits requires at least $\wt{\Omega}(N)$ queries.
\end{theorem}

\noindent The $\wt{\Omega}(\cdot)$ neglects logarithmic factors.
We stress that the above theorem implies not only that noisy implementations of Grover's algorithm itself fail to achieve a quadratic speedup, but in fact that \emph{any} $\NISQ$ algorithm will fail to do so. 
The intuition behind the proof is that noisy quantum devices can only run for so long before noise overwhelms the system, so it would suffice to prove the lower bound for hybrid quantum-classical algorithms with access to a noiseless quantum device that can run any bounded-depth circuit.
To analyze such algorithms, we draw upon tools developed recently in the context of lower bounds for learning quantum states and processes using unentangled measurements \cite{aharonov2022quantum, chen2021exponential, huang2022foundations, huang2022quantum, chen2021hierarchy, bubeck2020entanglement,chen2022tight}.

For the second task, the Bernstein-Vazirani problem, we find that a large quantum advantage still remains for $\NISQ$ algorithms.
Given an unknown $n$ bit string $s \in \{0, 1\}^n$, the Bernstein-Vazirani problem asks how many queries to a function $f(x) = x \cdot s$ are required to learn $s$.
Any classical algorithm requires at least $\Omega(n)$ queries to learn $s$.
However, the Bernstein-Vazirani quantum algorithm can learn the unknown bit string $s$ with just $1$ query.
We show that the query complexity of this problem in $\NISQ$ remains much smaller than the classical query complexity:

\begin{theorem}[Bernstein-Vazirani]\label{thm:BV_informal}
    There is a $\NISQ$ algorithm that solves the Bernstein-Vazirani problem over $n$ bits in at most $\mathcal{O}(\log n)$ queries.
\end{theorem}

\noindent Perhaps surprisingly, our analysis shows that a simple modification of the original Bernstein-Vazirani algorithm is already sufficiently noise-robust.
This result provides optimism
that there
may be other natural problems for which quantum advantage can be obtained in the NISQ era.

Finally, we consider the problem of predicting many properties in an unknown quantum system, also known as shadow tomography \cite{aaronson2018shadow, aaronson2019gentle, huang2020predicting, buadescu2021improved}.
In particular, we restrict to predicting absolute values of Pauli observables $\{I, X, Y, Z\}^{\otimes n}$:
Given many copies of an unknown $n$-qubit state $\rho$, the goal is to learn $|\Tr(P \rho)|$ for all $P \in \{I, X, Y, Z\}^{\otimes n}$ up to a constant error by processing the quantum state copies.
This task has received considerable attention in recent works \cite{huang2021information,chen2021exponential,huang2022quantum,huang2020predicting} which have demonstrated, both theoretically and experimentally, that very simple $\BQP$ algorithms can solve this task using only $O(n)$ copies, while any classical algorithm that can obtain classical data by performing measurements on individual copies of $\rho$ requires $\Omega(2^n)$ copies.
Here we show that this large quantum advantage is damped by the presence of noise. 
Specifically, we show the following exponential separation between $\NISQ$ and $\BQP$.

\begin{theorem}[Shadow tomography]\label{thm:shadow_informal}
Any $\NISQ$ algorithm with noise rate $\lambda$ per qubit requires at least $\Omega((1-\lambda)^{-n})$ copies of $\rho$ to learn $|\Tr(P \rho)|$, for all $P \in \{I, X, Y, Z\}^{\otimes n}$ up to a constant error.
In contrast, at most $\mathcal{O}(n)$ copies are needed for $\BQP$.
\end{theorem}

\noindent Theorem~\ref{thm:shadow_informal} demonstrates that a fault-tolerant quantum algorithm can be exponentially more powerful than any $\NISQ$ algorithm in learning quantum systems.
On the other hand, for small noise rate $\lambda$, the exponential scaling in the lower bound for $\NISQ$ algorithms has a base which is close to one. In \cite{huang2022foundations} it was shown that $\NISQ$ algorithms can achieve $(1 - \lambda)^{-\Theta(n)}$ even for more general noise models.
This suggests that $\NISQ$ algorithms can still perform well for learning quantum systems with a few hundred qubits \cite{huang2022quantum}.

\section{Related Work}
\label{sec:related}

\paragraph{Faulty oracle models.}

A number of works have studied the effect of \emph{imperfect oracles} on quantum speedups. For instance, \cite{muthukrishnan2019sensitivity} studied whether the exponential speedup achieved by the quantum annealing algorithm for the welded tree problem persists when the oracle is subjected to various kinds of noise. Relevant to our Theorem~\ref{thm:grover_informal}, \cite{shenvi2003effects,long2000dominant} considered the performance of Grover's algorithm when the phase oracle is subject to small phase fluctuations, and \cite{regev2008impossibility} showed that under this faulty oracle model speedups are not possible for any quantum algorithm; see also \cite{ambainis2012grover} for a different oracle noise model. Relevant to our Theorem~\ref{thm:BV_informal}, \cite{cross2015quantum} showed that noise in an oracle for subset parity does not affect the computational complexity of quantum learning algorithms in the same way conjectured for classical learning algorithms.

These results assume that noise occurs inside the oracle, but that the quantum computation leveraging the faulty oracle is noiseless. Moreover, the lower bounds in~\cite{muthukrishnan2019sensitivity,regev2008impossibility,ambainis2012grover} assume a global noise on the oracle as opposed to local qubit-wise noise considered in this work.
Accordingly, we view these results as conceptually orthogonal to the thrust of our work: whereas they focus on the effects of imperfections in implementing the underlying oracle, we study the effects of imperfections in the quantum computation due to local noise.

\paragraph{Noiseless hybrid quantum-classical models.}

A number of works have studied the power and limitations of hybrid quantum-classical models when the quantum computation is assumed to be noiseless. For example, the recent work~\cite{rosmanis2022hybrid} studies unstructured search in a noiseless hybrid setting where the algorithm can make queries to both classical and quantum versions of the search oracle. They show that any algorithm with constant success probability must make either $\Omega(\sqrt{N})$ queries to the quantum oracle or $\Omega(N)$ queries to the classical oracle. Earlier works \cite{chia2020need,coudron2020computations} showed that relative to various oracles, namely the recursive Simon's and welded tree problem, $\BQP$ is strictly more powerful than classical computation which is assisted by a noiseless bounded-depth quantum device. The oracle we use in Theorem~\ref{thm:superpoly2} is essentially a simplified version of the recursive Simon's problem, and in Appendix~\ref{app:shuffle}, we show how recursive Simon's itself can be used to simultaneously separate $\NISQ$ and the complexity classes considered in \cite{chia2020need,coudron2020computations} from $\BQP$.

\paragraph{Noise resilience of specific algorithms.}

A number of works have studied how existing algorithms perform under various forms of noise in their implementation. For unstructured search, \cite{shapira2003effect,pablo1999noise} studied whether Grover's algorithm is robust to various deviations like noisy Hadamard gates and Gaussian noise between iterations, while \cite{magniez2007search} demonstrated that recursive amplitude amplification is robust to noisy reflection operators.
\cite{stilck2021limitations,sharma2020noise,fontana2021evaluating,lavrijsen2020classical,wang2021noise}, among others, studied the resilience of specific quantum optimization algorithms like QAOA and VQE under models similar to our definition of $\lambda$-noisy circuits.
\cite{wang2021noise} found that in various regimes, these algorithms suffer significant slowdown when implemented on noisy quantum devices due to the flattening of the cost landscape.
\cite{bouland2022noise} shows that estimating the output probability of random quantum circuits to exponentially small additive error remains $\#\mathsf{P}$-hard even under the presence of small noise.
Furthermore, \cite{stilck2021limitations} showed that the presence of noise can make these quantum algorithms easy to simulate on classical computer.
In contrast, in the present work we study the capabilities and limitations for $\NISQ$ without necessarily focusing on any particular algorithm.

\paragraph{Complexity of noisy quantum circuits.} To our knowledge, the only other paper to study the relation of noisy quantum circuits and existing complexity classes is that of \cite{aharonov1996limitations}, who considered a generalization of our notion of $\lambda$-noisy circuits in which a random fraction of qubits at every layer are adversarially corrupted. Notably, they showed that polynomial-size noisy quantum circuits are no stronger than the complexity class $\QNC^1$ of logarithmic-depth noiseless circuits, whereas quasipolynomial-size noisy quantum circuits can compute any function in $\QNC^1$. Recall that in the present work, rather than study $\lambda$-noisy circuits in isolation, we consider the power of classical computation augmented by such circuits.

\section{Outlook}
\label{sec:outlook}

By abstracting hybrid quantum-classical computation in the NISQ era into a computational complexity class, our work offers a mathematical framework for reasoning about the potential for noisy quantum advantage.
We used tools from quantum query complexity to characterize how $\NISQ$ lies between $\BPP$ and $\BQP$.  On the one hand, the fact that $\NISQ$ can be more powerful than $\BPP$ provides optimism for the NISQ era and the computational advances it may precipitate.
On the other hard, $\NISQ$ being less powerful than $\BQP$ portends that we will have to wait until the advent of fault-tolerant devices to harness the richest features of quantum computation.
Our results for the $\NISQ$ complexity of three well-known problems in quantum computing punctuate our outlook by demonstrating specific promises and pitfalls of computation in the NISQ era.

There are many future directions to pursue with the $\NISQ$ complexity class.
A main open problem is to understand if we could show a separation between $\BPP, \NISQ,$ and $\BQP$ under a standard complexity-theoretic assumption.
It would be desirable to have an exponential oracle separation between $\BPP$ and $\NISQ$, as opposed to one that is merely super-polynomial.
Moreover, one could ask if a similar oracle separation exists between $\BPP$ and $\NISQ$ with the additional restriction that our quantum gates are spatially local, e.g., the gates are restricted to a two-dimensional geometry. 
Perhaps under this additional restriction of geometric locality on $\NISQ$, it is possible to better understand the computational complexity of $\NISQ$ without relying on oracle separations. 
Additionally, it would be valuable to analyze the $\NISQ$ complexity of other natural quantum algorithms, beyond the ones we studied. 
Natural targets include the original Simon's problem (as opposed to our variations thereof, see Appendix~\ref{app:bkw} for one approach in this direction), Forrelation~\cite{aaronson2015forrelation}, Shor's algorithm~\cite{shor1994algorithms,shor1999polynomial}, linear system solving~\cite{harrow2009quantum}, the recent random oracle result of~\cite{yamakawa2022verifiable}, and topological data analysis \cite{akhalwaya2022exponential}, among many others.

Taking a broader view, our work suggests a roadmap for investigating future quantum devices which may go beyond the NISQ era, but fall short of fault-tolerance.  The approach is to formulate a computational complexity class which encapsulates the salient features of whichever quantum devices are contemporary, and then to study that complexity class to make statements about the capabilities of those devices with great generality.  In such a future, there should be complexity classes beyond $\NISQ$, but still intermediate to $\BQP$.  
Whichever generalizations of $\NISQ$ prove fruitful in the future, they will have much to tell us about what computation is possible in that future world.

\paragraph{Acknowledgments}

The authors would like to thank Isaac Chuang for valuable conversations on complexity classes for learning theory, John Preskill for inspiring discussions on $\NISQ$ and its complexity class formulation, and John Wright for bringing~\cite{regev2008impossibility} and related works to our attention.








\bibliographystyle{unsrt}
\bibliography{ref}

\newpage

\appendix


\section{The \texorpdfstring{$\NISQ$}{NISQ} Complexity Class}
\label{sec:nisq}

In this section, we formally define the complexity class $\NISQ$. Then we recall the notion of classical oracles in classical ($\BPP$) and quantum computation $(\BQP)$ and generalize this to $\NISQ$.

\subsection{Definition of the complexity class}
\label{subsec:nisq}

We begin by recalling the single-qubit depolarizing channel $D_\lambda$.

\begin{definition}[Single-qubit depolarizing channel]
    Given $\lambda \in [0, 1]$.
    We define the \emph{single-qubit depolarizing channel} to be $D_\lambda[\rho] \triangleq (1-\lambda) \rho + \lambda (I / 2)$, where $\rho$ is a single-qubit density matrix.
\end{definition}

\begin{definition}[Depth-$1$ unitary]
    Given $n > 0$. An $n$-qubit unitary $U$ is a \emph{depth-$1$ unitary} if $U$ can be written as a tensor product of two-qubit unitaries.
\end{definition}

\noindent We consider noisy quantum circuits with noise level $\lambda$ to be defined as follows.

\begin{definition}[Output of a noisy quantum circuit]
    Let $\lambda \in [0, 1]$ and $n \in \mathbb{N}$. 
    Given $T \in\mathbb{N}$ and a sequence of $T$ depth-$1$ unitaries $U_1, \ldots, U_T$, the output of the corresponding \emph{$\lambda$-noisy depth-$T$ quantum circuit} is a random $n$-bit string $s \in \{0, 1\}^n$ sampled from the distribution
    \begin{equation} \label{eq:prob-NQC}
        p(s) = \bra{s} D_\lambda^{\otimes n}\bigl[U_T \ldots D_\lambda^{\otimes n}\bigl[U_2 D_\lambda^{\otimes n}\bigl[U_1 D_\lambda^{\otimes n}[\ketbra{0^n}{0^n}] U_1^\dagger\bigr] U_2^\dagger \bigr] \ldots U_T^\dagger \bigr] \ket{s} \, ,
    \end{equation}
    where every quantum operation is followed by a layer of single-qubit depolarizing channel. When $\lambda = 0$, we say that this circuit is \emph{noiseless}.
\end{definition}

\begin{remark}\label{remark:noise}
    We work with the single-qubit depolarizing channel as it is the most standard model for local noise. One could also consider stronger noise models, e.g. every qubit is randomly corrupted with probability $\lambda$ by an adversary rather than randomly decohered. Tautologically, the lower bounds we prove in this work will translate to such stronger models. We also prove our upper bounds, namely Theorem~\ref{thm:superpoly1} and \ref{thm:BV_informal}, under this stronger model (see Remarks~\ref{remark:superpoly} and~\ref{remark:BV}).
\end{remark}

\begin{definition}[Noisy quantum circuit oracle]
    We define $\mathrm{NQC}_{\lambda}$ to be an oracle that takes in an integer $n$ and a sequence of depth-$1$ $n$-qubit unitary $\{U_k\}_{k=1,\ldots,T}$ for any $T\in\mathbb{N}$ and outputs a random $n$-bit string $s$ according to Eq.~\eqref{eq:prob-NQC}.
    

    We define the time to query $\mathrm{NQC}_{\lambda}$ with $T$ depth-$1$ $n$-qubit unitaries to be $\Theta(nT)$, which is linear in the time to write down the input to the query.
\end{definition}

\noindent We now define $\NISQ$ algorithms, which are classical algorithms with access to the noisy quantum circuit oracle.
This provides a formal definition for hybrid noisy quantum-classical computation.

\begin{definition}[$\NISQ$ algorithm]
    A \emph{$\NISQ_\lambda$ algorithm with access to $\lambda$-noisy quantum circuits} is defined as a probabilistic Turing machine $M$ that can query $\mathrm{NQC}_{\lambda}$ to obtain an output bitstring $s$ for any number of times, and is denoted as $A_\lambda \triangleq M^{\mathrm{NQC}_\lambda}$. The runtime of $A_\lambda$ is given by the classical runtime of $M$ plus the sum of the times to query $\mathrm{NQC}_\lambda$.
\end{definition}

\noindent The $\NISQ$ complexity class for decision problems is defined as follows. Observe that the following recovers the definition for $\BPP$ when $M^{\mathrm{NQC}_{\lambda}}$ in the definition of $A_\lambda$ above is replaced by $M$.

\begin{definition}[$\NISQ$ complexity]
    A language $L \subseteq \{0, 1\}^*$ is in $\NISQ$ if there exists a $\NISQ_\lambda$ algorithm $A_\lambda$ for some constant $\lambda > 0$ that decides $L$ in polynomial time, that is, such that
    \begin{itemize}
        \item for all $x \in \{0, 1\}^*$, $A_\lambda$ produces an output in time $\mathrm{poly}(|x|)$, where $|x|$ is the length of $x$;
        \item for all $x \in L$, $A_\lambda$ outputs $1$ with probability at least $2/3$;
        \item for all $x \not\in L$, $A_\lambda$ outputs $0$ with probability at least $2/3$.
    \end{itemize}
\end{definition}

\subsection{Algorithms with oracle access}
\label{subsec:oracle}

In this work we study the complexity of learning a classical oracle or testing a property thereof. For instance, the unstructured search problem considers learning a classical oracle that highlights an element among $N$ elements.
We recall the following definition of a classical oracle $O$, as well as definitions of classical/quantum algorithms with access to the classical oracle $O$.

\begin{definition}[Classical oracle $O$]
    A \emph{classical oracle} $O$ is a function from $\{0, 1\}^n$ to $\{0, 1\}^m$ for some $n, m \in \mathbb{N}$.
    The $(n+m)$-qubit unitary $U_O$ corresponding to the classical oracle $O$ is given by $U_O \ket{x}\ket{y} = \ket{x} \ket{y \oplus O(x)}$ for all $x \in \{0, 1\}^n, y \in \{0, 1\}^m$.
\end{definition}

\begin{definition}[Classical algorithm with access to $O$]
    A \emph{classical algorithm $M^O$ with access to $O$} is a probabilistic Turing machine $M$ that can query $O$ by choosing an $n$-bit input $x$ and obtaining the $m$-bit output $O(x)$.
\end{definition}    

\begin{definition}[Quantum algorithm with access to $O$]
    A \emph{quantum algorithm $Q^O$ with access to $O$} is a uniform family of quantum circuits $\{U_n\}_n$, where $U_n$ is an $n'$-qubit quantum circuit given by
    \begin{equation}
        U_n \triangleq V_{n, k} (U_O \otimes \Id) \cdots (U_O \otimes \Id) V_{n, 2} (U_O \otimes \Id) V_{n, 1} \, ,
    \end{equation}
    for some integer $k \in \mathbb{N}$ and $n'$-qubit unitaries $V_{n, 1}, \ldots, V_{n, k}$ given as the product of many depth-$1$ unitaries. Here, $\Id$ denotes the identity matrix over $n' - n$ qubits.
\end{definition}

\noindent We now present the definition of $\NISQ$ algorithms with access to the classical oracle $O$, which requires first defining noisy quantum circuit oracles with access to $O$.

\begin{definition}[Noisy quantum circuit oracle with access to $O$]
    We define $\mathrm{NQC}^O_{\lambda}$ to be an oracle that takes in an integer $n'$ and a sequence of $n'$-qubit unitaries $\{U_k\}_{k=1,\ldots,T}$ for any $T\in\N$, where $U_k$ can either be a depth-$1$ unitary or $U_O \otimes I$, to a random $n$-bit string $s$ sampled according to the distribution
    \begin{equation} \label{eq:prob-NQC-O}
        p(s) = \bra{s} D_\lambda^{\otimes n'}\bigl[U_T \ldots D_\lambda^{\otimes n'}\bigl[U_2 D_\lambda^{\otimes n'}\bigl[U_1 D_\lambda^{\otimes n'}[\ketbra{0^{n'}}{0^{n'}}] U_1^\dagger\bigr] U_2^\dagger \bigr] \ldots U_T^\dagger \bigr] \ket{s} \, .
    \end{equation}
\end{definition}

\begin{definition}[$\NISQ$ algorithm with access to $O$]
    Let $\lambda\in[0,1]$. A $\NISQ_\lambda$ algorithm $A_\lambda^O = (M^{\mathrm{NQC}_{\lambda}})^O$ with access to $O$ is a probabilistic Turing machine $M$ that has the ability to classically query $O$ by choosing the $n$-bit input $x$ to obtain the $m$-bit output $O(x)$, as well as the ability to query $\mathrm{NQC}^O_{\lambda}$ by choosing $n'$ and $\{U_k\}_{k=1,\ldots,T}$ to obtain a random $n'$-bit string $s$. The runtime of $A^O_\lambda$ is given by the sum of the classical runtime of $M$, the number of classical queries to $O$, and the sum of the times to query $\mathrm{NQC}^O_\lambda$.
\end{definition}

\noindent With this definition in hand, we can extend the usual notions of relativized complexity to $\NISQ$:

\begin{definition}[Relativized $\NISQ$]
    Given a sequence of oracles $O: \brc{0,1}^n \to \brc{0,1}^{m(n)}$ parametrized by $n\in\mathbb{N}$, a language $L\subseteq\brc{0,1}^*$ is in $\NISQ^O$ if there exists a constant $\lambda > 0$ and a $\NISQ_\lambda$ algorithm $A^O_\lambda$ with access to $O$ that decides $L$ in polynomial time.
\end{definition}

\subsection{Algorithms of bounded depth}
\label{sec:boundeddepth}

In parts of this work we leverage the well-known connection \cite{aharonov1996limitations} between noisy quantum circuits and noiseless bounded-depth circuits. Here we briefly recall some standard notions regarding the latter, presented in the language of Sections~\ref{sec:nisq}.

\begin{definition}[Noiseless hybrid quantum-classical computation of bounded depth]
    A \emph{noiseless depth-$T$ algorithm} is a $\NISQ_0$ algorithm $A$ that only queries $\mathrm{NQC}_0$ on sequences of depth-1 $n$-qubit unitaries $\brc{U_k}_{k = 1,\ldots,T'}$ for $1\le T' \le T$.
\end{definition}

\begin{definition}[$\BPP^{\QNC}$]
    Let $f:\mathbb{N}\to\mathbb{N}$ be a nondecreasing function. A language $L\subseteq\brc{0,1}^*$ is in $\BPP^{\QNC[f(n)]}$ if there is a noiseless depth-$f(n)$ algorithm $A$ that decides $L$ in polynomial time. When $f(n) = O(\log^i(n))$, we denote this class by $\BPP^{\QNC^i}$. We also define $\BPP^{\QNC} \triangleq \cup_{i\ge 0} \BPP^{\QNC^i}$.
\end{definition}

\noindent Note that $\BPP^{\QNC}$ is contained in the class $\BQP$, as $\BQP$ can implement arbitrary \emph{polynomial}-depth quantum computation.

We can also define noiseless depth-$T$ algorithms with access to a classical oracle, as well as relativized versions of $\BPP^{\QNC^i}$ which we denote by $(\BPP^{\QNC^i})^O$, completely analogously to what is done in Section~\ref{subsec:oracle}.






\section{Preliminaries}




\subsection{Learning tree formalism and Le Cam's method}

We begin by recalling the learning tree formalism of \cite{chen2021exponential}, adapted here to the setting of $\NISQ$. This formalism will feature heavily in the proofs of our lower bounds against $\NISQ$.

\begin{definition}[Tree representation for $\NISQ$ algorithms] \label{def:tree_nisq}
    Given oracle $O: \brc{0,1}^n\to\brc{0,1}^m$, a $\NISQ_\lambda$ algorithm with access to $O$ can be associated with a pair $(\mathcal{T},\mathcal{A})$ as follows. The \emph{learning tree} $\mathcal{T}$ is a rooted tree, where each node in the tree encodes the transcript of all classical query and noisy quantum circuit results the algorithm has seen so far. The tree satisfies the following properties:
    \begin{itemize}[leftmargin=*,itemsep=0pt]
        \item Each node $u$ is associated with a value $p_O(u)$ corresponding to the probability that the transcript observed so far is given by the path from the root $r$ to $u$. In this way, $\calT$ naturally induces a distribution over its leaves. For the root $r$, $p_O(r) = 1$.
        \item At each non-leaf node $u$, we either classically query the oracle $O$ at an input $x\in\brc{0,1}^n$, or run a $\lambda$-noisy quantum circuit $A$ with access to $O$.
        \begin{enumerate}[itemsep=0pt,label=(\roman*)]
            \item \underline{Classical query}: $u$ has a single child node $v$ connected via an edge $(u,x,O(x))$, and we define 
            \begin{equation}
                p_O(v) = p_O(u) \, .
            \end{equation}
            \item \underline{Noisy circuit query}: The children $v$ of $u$ are indexed by the possible $s\in\brc{0,1}^{n'}$ that could be obtained as a result. We refer to the edge between $u$ and $v$ as $(u,A,s)$. We denote by $\ket{\phi_O(A)}$ the output state of the circuit so that the probability of traversing $(u,A,s)$ from node $u$ to child $v$ is given by $|\braket{s|\phi_O(A)}|^2$. We define 
            \begin{equation}
                p_O(v) = p_O(u) \cdot |\braket{s | \phi_O(A)}|^2 \, .
            \end{equation}
        \end{enumerate}
        \item If the total number of classical/quantum queries to $O$ made along any root-to-leaf path is at most $N$, we say that the \emph{query complexity} of the algorithm is at most $N$.
    \end{itemize}    
    $\mathcal{A}$ is any classical algorithm that takes as input a transcript corresponding to any leaf node $\ell$ and attempts to determine the underlying oracle or predict some property thereof.
\end{definition}


\begin{figure}
    \centering
    \includegraphics[width=0.99\textwidth]{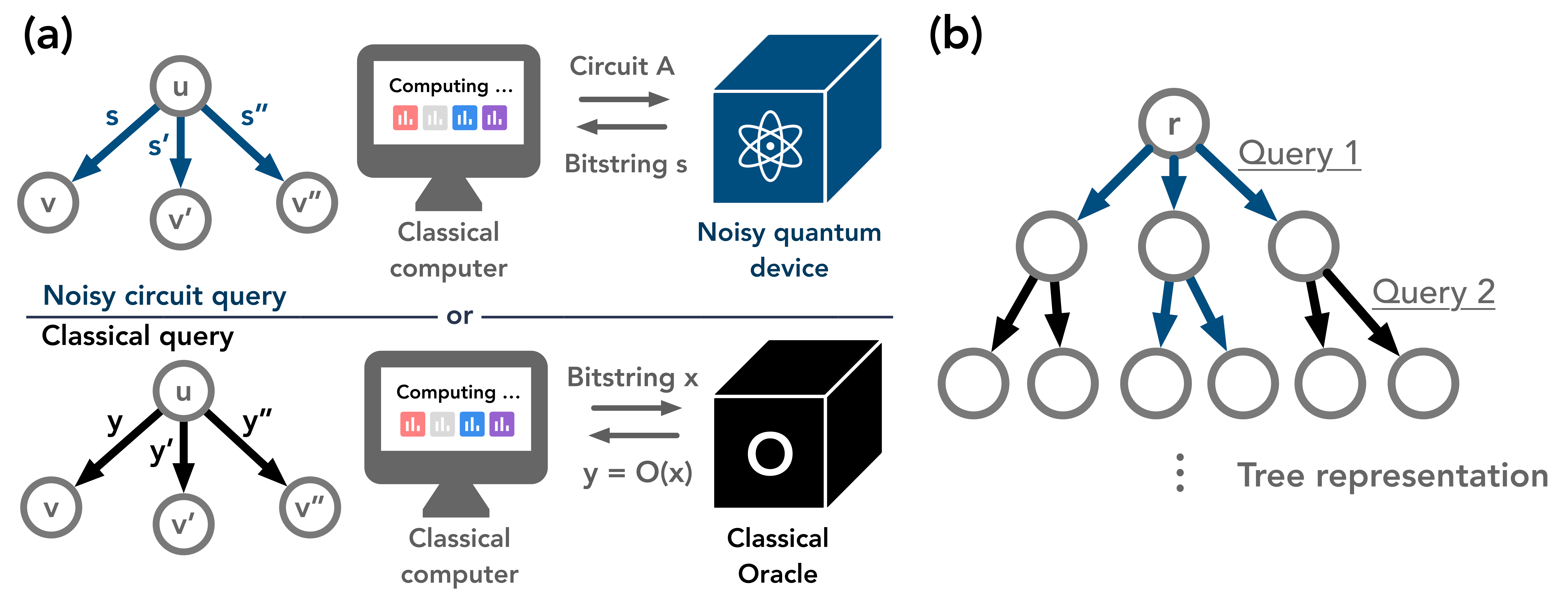}
    \caption{Illustration of the tree representation for $\NISQ$ algorithms. (a) At every memory state $u$ of the classical computer/algorithm, it could either make a noisy circuit query or a classical query. (b) The tree representation with a mix of noisy circuit queries and classical queries.}
    \label{fig:tree}
\end{figure}

\noindent The following lemma shows that slight perturbations to the distributions over children for each node do not change the overall distribution over leaves of $\calT$ by too much.

\begin{lemma}\label{lem:perturbchildren}
    Given learning tree $\calT$ corresponding to a $\NISQ_\lambda$ algorithm with query complexity $N$, suppose $\calT'$ is a learning tree obtained from $\calT$ as follows. For every node $u$ at which a noisy quantum circuit $A$ is run, replace $A$ by another circuit $A'$ such that the new induced distribution over children of $u$ is at most $\epsilon$-far from the original distribution in total variation. Then the distributions over leaves of $\calT$ and $\calT'$ are at most $\epsilon N$-far in total variation.
\end{lemma}

\begin{proof}
    Consider the sequence of trees $\calT^{(i)}$ where $\calT^{(0)} = \calT$ and $\calT^{(i)}$ is given by taking all $u$ in layer $i$ of $\calT^{(i-1)}$ that run some noisy quantum circuit $A$ and replacing them with the corresponding circuit $A'$ from $\calT'$. By design, $\calT^{(N)} = \calT'$. Let $p^{(i)}$ denote the distribution over leaves of $\calT^{(i)}$. It suffices to show that $\tvd(p^{(i)}, p^{(i-1)}) \le \epsilon$. 
    
    Note that $p^{(i-1)}$ specifies some mixture over distributions $p_v$, where $p_v$ is the distribution over leaves conditioned on reaching node $v$ in the $i$-th layer. In particular, in this mixture, $v$ is sampled by sampling parent node $u$ by running the $\NISQ$ algorithm corresponding to $\calT'$ for $i-1$ steps and then running the corresponding quantum circuit $A$ from $\calT$. In contrast, $p^{(i)}$ is a mixture over the same distributions $p_v$, but $v$ is sampled by running the $\NISQ$ algorithm corresponding to $\calT'$ for $i$ steps and then running the corresponding quantum circuit $A'$ from $\calT'$. These two distributions over $v$ are at most $\epsilon$-far in total variation, so the two mixture distributions are also at most $\epsilon$-far in total variation as claimed.
\end{proof}

Our lower bounds will be based on Le Cam's method\--- see Section 4.3 of \cite{chen2021exponential} for an overview in the context of the tree formalism of Definition~\ref{def:tree_nisq}. In every case we will reduce to some \emph{distinguishing task} in which the algorithm must discern whether the oracle it has access to comes from one family of oracles or from another. For example, for unstructured search, the distinguishing task will be whether the oracle corresponds to some element in the search domain or whether the oracle is the identity channel.

More concretely, given two disjoint sets of oracles $S_0, S_1$, we will design distributions $D_0, D_1$ over $S_0, S_1$. Given any algorithm specified by some $(\mathcal{T},\mathcal{A})$, we will upper bound the total variation distance between the following two distributions. We consider the mixture of distributions $p_{O_i}$ over leaves of the learning tree when the underlying oracle $O_i$ is sampled according to $D_0$ at the outset, as well as the mixture when the oracle is sampled according to $D_1$. The following lemma shows that upper bounding $\tvd(\mathbb{E}_{i\sim D_0}[p_{O_i}],\mathbb{E}_{i\sim D_1}[p_{O_i}])$ suffices to show a query complexity lower bound for the distinguishing task:

\begin{lemma}[Le Cam's two-point method, see e.g. Lemma 4.14 from \cite{chen2021exponential}]\label{lem:lecam}
    Let $\brc{O_i}_{i\in S_0}$ and $\brc{O_i}_{i\in S_1}$ be two disjoint sets of oracles. Given a tree $\calT$ as in Definition~\ref{def:tree_nisq} corresponding to a $\NISQ$ algorithm that makes $N$ oracle queries, let $p_i$ denote the induced distribution over leaves when the algorithm has access to $O_i$. If $\tvd(\mathbb{E}_{i\sim D_0}[p_0],\mathbb{E}_{i\sim D_1}[p_1]) < 1/3$, there is no algorithm $\calA$ that maps transcripts $T$ corresponding to leaves of $\calT$ to $\brc{0,1}$ which can distinguish between $S_0$ and $S_1$ with advantage $1/3$.\footnote{By this, we mean that for all $a\in\brc{0,1}$ and $i\in S_a$, $\Pr[T\sim p_i]{\calA(T) = \bone{i \in S_0}} \ge 2/3$.}
\end{lemma}

\subsection{Basic hybrid argument}

Here we describe a standard template for showing quantum query complexity lower bounds via a hybrid argument.

\begin{lemma}\label{lem:hybrid_basic}
    Let $\calE_0, \calE_1$ be quantum channels on $n$ qubits such that for all pure states $\sigma$, we have $\norm{(\calE_0 - \calE_1)[\sigma]}_{\tr} \le \epsilon$. Let $A$ be any depth-$T$ quantum circuit with access to one of the two channels, and let $s \in \brc{0,1}^n$ be the random string output by the circuit. Let $p_0, p_1$ denote the distribution over $s$ when $A$ has access to $\calE_0, \calE_1$ respectively. Then $\tvd(p_0,p_1) \le \epsilon T$.
\end{lemma}

\begin{proof}
    Let $\calE = \calE_s$ for $s\in\brc{0,1}$, and define the channel $\mathcal{U}_i$ which acts by $\mathcal{U}_i(\sigma) = U_i \sigma U_i^\dagger$ where $U_i$ is an associated unitary operator. We proceed via a hybrid argument. The output state of the circuit is given by
    \begin{equation}
        \sigma^s = \mathcal{U}_T \circ \calE \circ \cdots \circ \mathcal{U}_2 \circ \calE \circ \mathcal{U}_1[\ketbra{0^n}{0^n}]
    \end{equation}
    for some unitaries $U_1,\ldots,U_T$.
    For $s' = 1 - s$ and $1 \le i \le T$ define
    \begin{equation}
        \sigma^{(i)} \triangleq \mathcal{U}_T \circ \calE_s \circ \cdots \circ \mathcal{U}_{i+1} \circ \calE_s \circ \mathcal{U}_i \circ \calE \circ \cdots \circ \mathcal{U}_2 \circ \calE \circ \mathcal{U}_1[|0^n\rangle\langle 0^n|]\,.
    \end{equation}
    Then
    \begin{align}
        \norm{\sigma^s - \sigma^{s'}}_{\tr} &= \biggl\|\sum^T_{i=1}\sigma^{(i)} - \sigma^{(i-1)}\biggr\|_\tr \le \sum^T_{i=1} \|\sigma^{(i)} - \sigma^{(i-1)}\|_\tr \\
        &\le \sum^T_{i=1} \norm{(\calE - \calE_s) \circ \mathcal{U}_{i-1} \circ \calE \circ \cdots \circ \mathcal{U}_2\circ \calE \circ \mathcal{U}_1[|0^n\rangle \langle 0^n|]}_\tr \le T\sup_\sigma \norm{(\calE - \calE_s)[\sigma]}_\tr,
    \end{align}
    where the supremum is over all density matrices. By convexity of the trace norm, this bound still holds when the supremum is restricted to pure states $\sigma$. By assumption, the above quantity is $\epsilon T$. The total variation distance between $p_1$ and $p_2$ as defined in lemma statement is simply the $L_1$ distance between the diagonals of $\sigma^s$ and $\sigma^{s'}$, which is upper bounded by $\norm{\sigma^s - \sigma^{s'}}_\tr \le \epsilon T$.
\end{proof}


\section{Super-Polynomial Oracle Separations}
\label{sec:super}

\subsection{\texorpdfstring{$\NISQ$}{NISQ} vs. \texorpdfstring{$\BPP$}{BPP}}
\label{sec:robust_simon}

This section is devoted to proving the following theorem.
\begin{theorem}[Restatement of Theorem~\ref{thm:superpoly1}]
$\BPP^{O_1} \subsetneq \NISQ^{O_1}$ relative to a classical oracle ${O_1}$.
\end{theorem}

\noindent Our basic strategy is to modify the Simon's oracle into a new classical oracle such that the new oracle is robust to noise.
We note that a $\NISQ$ algorithm is unable to implement known fault-tolerant quantum computation schemes that can run for any arbitrary quantum circuit with a polynomial number of gates.
However, we will still take inspiration from a fault-tolerant quantum computation scheme \cite{aharonov1997fault} to define a certain ``robustified Simon's oracle'' relative to which we obtain a super-polynomial separation between $\BPP$ and $\NISQ$. 
As we will show, because the fault-tolerant scheme of \cite{aharonov1997fault} is robust not just to local depolarizing noise but to arbitrary local noise occurring with sufficiently small constant rate, the $\NISQ$ algorithm that we give will ultimately be robust under this stronger noise model as well (see Remark~\ref{remark:superpoly}).

\subsubsection{Recursively-defined concatenated code}
\label{subsec:CSS}

We consider a Calderbank-Shor-Steane (CSS) code built from two classical linear codes $C_1, C_2$, where $C_1 \triangleq C$ is a punctured doubly-even self-dual code and $C_2 \triangleq C^\perp$ (we refer the reader to \cite{aharonov1997fault} for background on these notions).
We consider $C_1, C_2$ to be over $m$ classical bits.
The corresponding CSS code encodes a single logical qubit into $m$ physical qubits.
Let $\vec{1}_m$ denote the all-ones vector of length $m$ (when the subscript is clear from context, we will omit it). The two code words in the CSS code are given by
\begin{equation}
\label{E:S0S1}
    \ket{S_0} = \frac{1}{\sqrt{|C^\perp|}} \sum_{w \in C^\perp} \ket{w}, \quad \ket{S_1} = \frac{1}{\sqrt{|C^\perp|}} \sum_{w \in C^\perp} \ket{w \oplus \vec{1}_m},
\end{equation}
where $\oplus$ denotes addition over $\mathbb{Z}_2^m$ (i.e., it is the bit-wise XOR).
Denote by $d$ the number of errors that can be corrected by the CSS code.
The two parameters $m$ and $d$ are both considered to be constant.
We define
\begin{align}
\label{E:A0}
A_0 &\triangleq \left\{ w \oplus x \, \Big| \, w \in C^\perp, x \in \{0, 1\}^m, |x| \leq d \right\} \\
\label{E:A1}
A_1 &\triangleq \left\{ w \oplus x \, \Big| \, w \in C^\perp \oplus \vec{1}, x \in \{0, 1\}^m, |x| \leq d \right\},
\end{align}
where $C^\perp \oplus \vec{1}$ denotes the set $\{x \oplus \vec{1} \,|\, x \in C^\perp\}$ and $|x|$ is the number of $1$'s in $x$.

\begin{lemma}[Disjointness of $A_0$ and $A_1$] \label{lem:disjoint-A01}
With the above definitions, we have
\begin{equation} \label{eq:disjoint-faults}
    A_0 \cap A_1 = \emptyset.
\end{equation}
\end{lemma}
\begin{proof}
This lemma follows from the definition of $d$.
Assume that the intersection is non-empty, i.e., $A_0 \cap A_1 \neq \emptyset$.
Hence $w_1 \oplus x_1 = w_2 \oplus \vec{1} \oplus x_2$ for some $w_1, w_2 \in C^\perp$, $x_1, x_2 \in \{0, 1\}^m$ with $|x_1|, |x_2| \leq d$.
Let us define $E_1 = \bigotimes_{i=1}^m X^{x_{1 i}}$, $E_2 = \bigotimes_{i=1}^m X^{x_{2 i}}$ for single-qubit Pauli $X$.
Then,
\begin{align}
    E_1 \ket{S_0} &= \frac{1}{\sqrt{|C^\perp|}} \sum_{w \in C^\perp} \ket{w \oplus x_1} = \frac{1}{\sqrt{|C^\perp|}} \sum_{w \in C^\perp} \ket{w \oplus w_1 \oplus x_1}\\
    &= \frac{1}{\sqrt{|C^\perp|}} \sum_{w \in C^\perp} \ket{w \oplus w_2 \oplus x_2 \oplus \vec{1}} = \frac{1}{\sqrt{|C^\perp|}} \sum_{w \in C^\perp} \ket{w \oplus \vec{1} \oplus x_2} = E_2 \ket{S_1}.
\end{align}
Hence $\bra{S_1} E_2 E_1 \ket{S_0} = 1$, where $E_1$ and $E_2$ are Pauli operators with weight at most $d$.
This contradicts the definition of $d$.
\end{proof}

\noindent We now recall the recursive concatentation construction from \cite{aharonov1997fault}. 
Given an integer $r > 0$, we define the following code va $r+1$ levels of recursion.
Each level encodes a single qubit using the above CSS code over $m$ qubits from the previous level.
Hence, a single qubit in the top level is encoded by a total of $m^{r}$ qubits in the bottom level.
More formally, for each $r$, we will define two sets $B^{(r)}_0, B^{(r)}_1$ over $m^r$-bit strings as follows.
These sets contain the computational basis states that are used to span the recursively-defined concatenated code.

\begin{definition}[Basis of the concatenated code]
\label{def:basisconcat1}
For $r = 1$, $B^{(1)}_0 \triangleq C^\perp$ and $B^{(1)}_1 \triangleq C^\perp \oplus \vec{1}$.
For $r > 1$, we define $B^{(r)}_0, B^{(r)}_1$ recursively,
\begin{align}
    B^{(r)}_0 &\triangleq \left\{ (v_1, \ldots, v_m) \in \{0, 1\}^{m^r} \, \Big| \, w \in C^\perp, v_i \in B^{(r-1)}_{w_i}, \forall i = 1, \ldots, m \right\}, \\
    B^{(r)}_1 &\triangleq \left\{ (v_1, \ldots, v_m) \in \{0, 1\}^{m^r} \, \Big| \, w \in C^\perp \oplus \vec{1}, v_i \in B^{(r-1)}_{w_i}, \forall i = 1, \ldots, m \right\}.
\end{align}
\end{definition}
\noindent The two code words in the recursively-defined concatenated code are then given by
\begin{equation}
\label{E:codewords1}
    \ket{R_b} = \frac{1}{\sqrt{|B^{(r)}_b|}} \sum_{x \in B^{(r)}_b} \ket{x}, \quad b \in \{0, 1\}.
\end{equation}
For each $r$, we also define two sets $A^{(r)}_0, A^{(r)}_1$ over $m^r$-bit strings that correspond to the neighborhoods around $B^{(r)}_0, B^{(r)}_1$ induced by errors.

\begin{definition}[Neighborhood of $B^{(r)}_0, B^{(r)}_1$]
For $r = 1$, $A^{(1)}_0 \triangleq A_0$ and $A^{(1)}_1 \triangleq A_1$.
By Eq.~\eqref{eq:disjoint-faults}, $A^{(r)}_0 \cap A^{(r)}_1 = \emptyset$.
For $r > 1$, we define $A^{(r)}_0, A^{(r)}_1$ recursively,
\begin{align}
    A^{(r)}_0 &\triangleq \left\{ (v_1, \ldots, v_m) \in \{0, 1\}^{m^r} \, \Big| \, w_0 \in C^\perp, x_0 \in \{0, 1\}^m, |x_0| \leq d, v_i \in A^{(r-1)}_{w_{0i}} \ \forall i \ \text{s.t.} \  x_{0i} = 0 \right\}, \\
    A^{(r)}_1 &\triangleq \left\{ (v_1, \ldots, v_m) \in \{0, 1\}^{m^r} \, \Big| \, w_1 \in C^\perp \oplus \vec{1}, x_1 \in \{0, 1\}^m, |x_1| \leq d, v_i \in A^{(r-1)}_{w_{1i}} \ \forall i \ \text{s.t.} \  x_{1i} = 0 \right\}.
\end{align}
\end{definition}

\noindent We can prove the following two lemmas.

\begin{lemma}[Structure of $A^{(r)}_0$ and $A^{(r)}_1$] \label{lem:strucAr}
For all $r \geq 1$, we have
\begin{equation}
    A^{(r)}_0 \oplus \vec{1} = A^{(r)}_1.
\end{equation}
\end{lemma}
\begin{proof}
    We consider a proof by induction on $r \geq 1$.
    By definition of $A_0$ and $A_1$, we have $A^{(1)}_0 \oplus \vec{1} = A^{(1)}_1$, which establishes the base case of $r = 1$.
    For $r > 1$, we show that for any $(v_1, \ldots, v_m) \in A^{(r)}_0$, we have $(v_1, \ldots, v_m) \oplus \vec{1} \in A^{(r)}_1$.
    Consider $w_0, x_0$ corresponding to $(v_1, \ldots, v_m)$.
    Using $v_i \in A^{(r-1)}_{w_{0 i}}$ for all $i$ with $x_{0 i} = 0$ and the inductive hypothesis that $A^{(r-1)}_{0} \oplus \vec{1} = A^{(r-1)}_{1}$,
    we have $v_i \oplus \vec{1} \in A^{(r-1)}_{w_{0 i} \oplus 1}$ for all $i$ with $x_{0 i} = 0$.
    Hence, by considering $w_1 = w_0 \oplus \vec{1}$ and $x_1 = x_0$, we have $(v_1, \ldots, v_m) \oplus \vec{1} \in A^{(r)}_1$.
    Similarly, we can show that for any $(v_1, \ldots, v_m) \in A^{(r)}_1$, we have $(v_1, \ldots, v_m) \oplus \vec{1} \in A^{(r)}_0$. Therefore, we have shown that $A^{(r)}_0 \oplus \vec{1} = A^{(r)}_1$.
\end{proof}

\begin{lemma}[Disjointness of $A^{(r)}_0$ and $A^{(r)}_1$] \label{lem:disjointAr} \label{lem:disjoint-Ar01}
For all $r \geq 1$, we have
\begin{equation}
    A^{(r)}_0 \cap A^{(r)}_1 = \emptyset.
\end{equation}
\end{lemma}
\begin{proof}
The base case is given in Lemma~\ref{lem:disjoint-A01}.
Now for the induction step, assume $A^{(r-1)}_0 \cap A^{(r-1)}_1 = \emptyset$.
Eq.~\eqref{eq:disjoint-faults} implies that the minimum Hamming distance between any two bitstrings in $C^\perp$ and $C^\perp \oplus \vec{1}$ is at least $2d + 1$.
Hence, for any $w_1 \in C^\perp$ and $w_2 \in C^\perp \oplus \vec{1}$, after removing at most $2d$ bits (the bits with $x_{1i} = 1$ or $x_{2i} = 1$), there still exists an index $i$ among the rest of the bits (i.e., the bits with $x_{1i} = 0$ and $x_{2i} = 0$) such that $w_{1i} \neq w_{2i}$.
Because $A^{(r-1)}_{w_{1i}} \cap A^{(r-1)}_{w_{2i}} = \emptyset$ by the induction hypothesis, we have $A^{(r)}_0 \cap A^{(r)}_1 = \emptyset$.
\end{proof}

\subsubsection{Robustified Simon's problem}

Given a large enough integer $n$, 
we consider Simon's problem over $n' = 2^{\Theta(\log(n)^{c})}$ bits for a constant $0 < c < 1$.
Here $1/c$ corresponds to the constant $c_2$ from Theorem 10 of \cite{aharonov1997fault}.
We consider $r = \Theta(\log \log (n'))$ and encode each of the $n'$ bits using $m^r$ bits.
Because $m = \mathcal{O}(1)$, we have $m^r n' = 2^{\Theta(\log(n)^{c})} < n$ for large enough $n$.

Given a classical function $f_s: \{0, 1\}^{n'} \rightarrow \{0, 1\}^{n'}$ from Simon's problem with secret string $s \in \{0, 1\}^{n'}$, we define a classical function $\wt{f}_s: \{0, 1\}^n \rightarrow \{0, 1\}^{m^r n'}$ as follows.
Let $x$ be an $n$-bit string. We focus on the first $m^r n'$ bits of $x$ and divide them into $n'$ $m^r$-bit strings as $x_1, \ldots, x_{n'}$.
We first define $\wt{f}^0_s : \{0, 1\}^{n} \rightarrow \{0, 1\}^{n'}$ as follows,
\begin{equation}
\label{E:tildefsdef0}
    \wt{f}^0_s(x) \triangleq \begin{cases}
    f_s(b_1 \ldots b_{n'}), & \mathrm{if} \,\, \exists! \, b_1, \ldots, b_{n'} \in \{0, 1\}, \,\, \mathrm{s.t.} \,\, x_i \in A^{(r)}_{b_i}, \forall i = 1, \ldots, n',\\
    0^{n'}, & \mathrm{otherwise}
    \end{cases}\,.
\end{equation}
We use $\exists!$ to denote ``there exists a unique choice''.
Because $A^{(r)}_{0}$ and $A^{(r)}_{1}$ are disjoint, there either exists a unique choice of $b_1, \ldots, b_{n'}$ or does not exist any choice of $b_1, \ldots, b_{n'}$ that satisfies $x_i \in A^{(r)}_{b_i}, \forall i = 1, \ldots, n'$.  Letting $[\wt{f}_s^0(x)]_k$ denote the $k$th bit of $\wt{f}_s^0(x)$, we define the function $\wt{f}_s : \{0,1\}^{n} \to \{0,1\}^{m^r n'}$
by
\begin{equation}
\label{E:tildefsdef1}
    \wt{f}_s(x) \triangleq \left( \underbrace{[\wt{f}^0_s(x)]_1,\ldots\,,[\wt{f}^0_s(x)]_1}_{m^r\text{ times}}\,, \ldots, \underbrace{[\wt{f}^0_s(x)]_{n'},\ldots\,,[\wt{f}^0_s(x)]_{n'}}_{m^r\text{ times}} \right) \in \{0, 1\}^{m^r n'}.
\end{equation}
The function $\wt{f}_s$ can be considered as the robust version of $f_s$, where the output bitstring is stable over a large number of bitstrings.

Let $U_{\wt{f_s}}$ be the unitary from Eq.~\eqref{eq:oracleaction1}.
\begin{equation}\label{eq:oracleaction1}
    U_{\wt{f_s}} \ket{x}\ket{y} = \ket{x} \ket{y \oplus {\wt{f_s}}(x)}, \quad \forall x \in \{0, 1\}^n, y \in \{0, 1\}^{m^r n'}\,.
\end{equation}
We denote by $O_{\wt{f}_s}$ the oracle which applies this unitary.  Then we have the following theorem, which is the main result of this section and implies Theorem~\ref{thm:superpoly1}:
\begin{theorem}\label{thm:superspeedup}
    For $\lambda$ sufficiently small, there is a $\NISQ_\lambda$ algorithm which, given oracle access to $O_{\wt{f}_s}$, can determine whether $f_s$ is 2-to-1 or 1-to-1 with constant advantage in time at most $\mathcal{O}(\text{\rm poly}(n))$. By contrast, any classical algorithm with access to $O_{\wt{f}_s}$ requires at least $\Omega(\text{\rm superpoly}(n))$ time, to determine whether $f_s$ is 2-to-1 or 1-to-1 with constant advantage.  Thus, relative to oracles $O$ of this form, $\BPP^O \subsetneq \NISQ^O$. 
\end{theorem}
\noindent Our strategy for proving this theorem is to combine the following two ingredients. First, we draw on tools underlying the error correction scheme in Theorem 10 of~\cite{aharonov1997fault}. This scheme requires the algorithm to be able to initialize constant-error noisy zero states in the middle of the circuit, which our computational model does not allow for. Our second proof ingredient is then to leverage the noisy majority vote approach in Theorem 4 of~\cite{aharonov1996limitations} for distilling constant-error noisy zero states in the middle of the circuit from constant-error noisy zero states prepared at the beginning of the circuit. Putting these together, we can utilize a noisy quantum machine with $n' \log(n')^{c_1} \times 2^{\Theta(\log(n')^c)} = \mathrm{poly}(n)$ qubits (and a similar number of gates) and noise rate $\lambda < \lambda_0$ (for a small constant $\lambda_0$) to run an encoded version of Simon's algorithm on the oracle $O_{\wt{f}_s}$.
We will find that our classical oracle, which essentially implements a classical error correction code, interfaces nicely with the quantum error correction scheme of~\cite{aharonov1997fault}, enabling our algorithm to work.

\subsubsection{Simulating fault-tolerant Simon's in \texorpdfstring{$\NISQ$}{NISQ} with exponential overhead}

Recall that Simon's algorithm consists of three steps: (i) prepare the all-plus state on the first half of the qubits and the all-zero state on the second half of the qubits, (ii) query the oracle, and (iii) apply the Hadamards to the first half of the qubits followed by measuring them in the computational basis. In this section, we describe how to implement encoded versions of the first and last steps. The main guarantees for these steps are respectively given by Lemmas~\ref{lem:stateprep}-\ref{lem:stateprep2} and Lemma~\ref{lem:endofprotocol} below.

To set up the proof, we require some definitions from~\cite{aharonov1997fault} so that we can state the needed theorems from~\cite{aharonov1997fault} and~\cite{aharonov1996limitations}.  Accordingly, let us recall the following notions of $(r, k)$-sparse sets and $(r, k)$-deviated states from \cite{aharonov1997fault}.
\begin{definition}[$(r, k)$-sparse set]
An $(r, k)$-sparse set of qubits over many blocks of the $m^r$ qubits is defined recursively as follows.
A set $A$ of qubits over many blocks of $m$ qubits is $(1, k)$-sparse if and only if every block has at most $k$ qubits that are in $A$.
A set $A$ of qubits over many blocks of $m^r$ qubits is $(r, k)$-sparse if and only if for every block, by treating the $m^r$ qubits as $m$ sub-blocks of $m^{r-1}$ qubits, there are at most $k$ sub-blocks that are not $(r-1, k)$-sparse.
\end{definition}
\begin{definition}[$(r, k)$-deviate]
    A state $\rho$ is said to be $(r, k)$-deviated from $\rho'$ if $k$ is the minimum integer such that there exists an $(r, k)$-sparse set of qubits $A$, such that $\rho_{A^c} = \rho'_{A^c}$.
    Here, we denote $\rho_{A^c}$ to be the reduced density matrix of $\rho$ on the qubits not in set $A$.
\end{definition}

\noindent We will also need the definitions of location and quantum computation code from~\cite{aharonov1997fault}.

\begin{definition}[Location]
    A set $(q_1,\ldots,q_a,t)$ is a \emph{location} in a quantum circuit $Q$ if the qubits $q_1,\ldots,q_a$ participated in the same gate in $Q$ at layer $t$, and no other qubit participated in that gate. If a qubit $q$ did not participate in any gate at layer $t$, then $(q,t)$ is a location in $Q$ as well.
\end{definition}

\begin{definition}[Quantum computation code]
A quantum code $C$ encoding a single logical qubit into $m$ physical qubits is called a \emph{quantum computation code} if it is accompanied with a universal set of gates $\mathcal{G}$ with fault tolerant procedures, including fault tolerant encoding, decoding, and error correction procedures. Moreover, we require that \emph{(i)} all procedures use only gates from $\mathcal{G}$, and \emph{(ii)} the error correction procedure takes any $m$-qubit density matrix to a state in the code space.

Apart from state encoding and decoding, the quantum computation code $C$ also encodes any gate $g \in \mathcal{G}$ into a circuit $P(g)$ such that for any pure state $\ket{\psi}$, $P(g)$ maps the state encoding of $\psi$ under $C$ to the state encoding of $g\ket{\psi}$.
If $g$ is a $k$-qubit gate, then the circuit $P(g)$ acts on $k$ blocks of $m$ qubits (each block encoding one logical qubit) possibly with other ancilla qubits.
\end{definition}



\noindent We are now prepared to recall the threshold theorem of \cite{aharonov1997fault}, which we state in the context of probabilistic qubit-wise noise (with probability $\lambda$, an arbitrary single-qubit quantum channel is applied on the qubit). In the rest of the section, we consider any failure probability $\delta > 0$, and a (noiseless) quantum circuit $Q$ with $2 n'$ input qubits, depth $t$, and $v$ locations. 
Let $C$ be a quantum computation code with gates $\mathcal{G}$ that corrects $d$ errors.
Let
\begin{equation}
V: \mathbb{C}^2 \to (\mathbb{C}^2)^{\otimes m^r}    
\end{equation}
be the encoding map for the code given by recursively concatenating $C$  a total of $r = \mathcal{O}(\log\log(v / \delta))$ times.
Using key lemmas for establishing the threshold theorem from \cite{aharonov1997fault} (Theorem 10 therein), we can extract the following two results:

\begin{theorem}[Lemma 8 and 10 from \cite{aharonov1997fault}]
\label{thm:Doritfault1}
    There is an absolute constant $\lambda_c \in[0,1]$ such that for any $\lambda < \lambda_c$, there exists a $\lambda$-noisy quantum circuit $Q'$ which can initialize ancillary qubits at any time (these ancillary qubits are also subject to qubit-wise noise of $\lambda$) during the computation and satisfies the following.
    $Q'$ operates on $m^r n'$ qubits and has depth $\mathcal{O}(t\, \mathrm{polylog}(v/\delta))$,
    and the output state $\rho$ of $Q'$ is $(r, d)$-deviated from
    \begin{equation}
        V^{\otimes n'} Q \ketbra{0^{n'}}{0^{n'}} Q^\dagger (V^{\otimes n'})^\dagger
    \end{equation}
    with probability $1 - \delta$ over the local noise.
\end{theorem}

\begin{theorem}[Lemma 8, 9, and 10 from \cite{aharonov1997fault}]
\label{thm:Doritfault2}
    There is an absolute constant $\lambda_c \in[0,1]$ such that for any $\lambda < \lambda_c$, there exists a $\lambda$-noisy quantum circuit $Q'$ which can initialize ancillary qubits at any time (these ancillary qubits are also subject to qubit-wise noise of $\lambda$) during the computation, and a classical postprocessing algorithm $\mathcal{A}$ based on recursive majority vote, that satisfies the following.
    $Q'$ operates on $m^r n'$ qubits and has depth $\mathcal{O}(t \, \mathrm{polylog}(v/\delta))$.
    Let $\sigma$ be any $n'$-qubit state.
    Let $\mathcal{D}$ be the $n'$-bit string distribution generated by measuring $Q \sigma Q^\dagger$ in the computational basis.
    For any state $\rho$ that is $(r, d)$-deviated from
    \begin{equation}
        V^{\otimes n'} \sigma (V^{\otimes n'})^\dagger,
    \end{equation}
    applying a $\lambda$-noisy computational basis measurements on the output state of $Q'$ given input state $\rho$, followed by the classical algorithm $\mathcal{A}$, produces a distribution $\mathcal{D}'$ equal to $\mathcal{D}$ with probability $1 - \delta$ over the local noise.
\end{theorem}

\noindent Our CSS code defined at the beginning of Subsection~\ref{subsec:CSS} in Eqs.~\eqref{E:S0S1},~\eqref{E:A0},~\eqref{E:A1} and the surrounding text, which corrects $d$ errors 
is a suitable quantum computational code for the purposes of the above Theorem.
Indeed, the construction in the proof of Theorem~\ref{thm:Doritfault1}~and~\ref{thm:Doritfault2} is to build the concatenated code described in Definition~\ref{def:basisconcat1} and Eq.~\eqref{E:codewords1} and the surrounding text, and then to show that it has suitable fault-tolerant quantum error correction properties.

As mentioned above, the key difference between our model for noisy quantum circuits versus the one considered in that work is that the latter allows the circuit to initialize ancillary qubits at any point in the computation, whereas in our setting all qubits are present at time zero. In \cite{aharonov1996limitations} it was shown how to pass from the latter setting to the former with a blowup in total number of qubits that is exponential in the depth of the original circuit.

\begin{theorem}[Theorem 4 from \cite{aharonov1996limitations}]\label{thm:majority}
    There exists an absolute constant $\lambda_c \in [0,1]$ such that for any non-negative $\lambda < \lambda_c$ and any $t \in\mathbb{N}$, there exists a $\lambda$-noisy quantum circuit of depth $t$ operating on $3^t$ qubits initialized to the all-zero state such that the output state's first qubit is in the state $\ket{0}$ with probability at least $1 - \lambda$.
\end{theorem}

\noindent We now use Theorems~\ref{thm:Doritfault1},~\ref{thm:Doritfault2}, and \ref{thm:majority} to argue in Lemma~\ref{lem:stateprep} (resp. Lemma~\ref{lem:stateprep2}) that we can prepare an encoding of the all-plus state (resp. all-zero state) over $n'$ qubits using $\poly(n)$ total qubits including ancillas, thus realizing an encoded version of the first step in Simon's algorithm.

\begin{lemma}\label{lem:stateprep}
    Suppose $n' \le \exp(\log^c n)$ for $0 < c < 1$ a sufficiently small constant. There exists an absolute constant $\lambda_c \in [0,1]$ such that for any non-negative $\lambda < \lambda_c$, there exists a $\lambda$-noisy quantum circuit which operates on $n'$ input qubits and $\poly(n)$ ancillary qubits and has $\mathrm{polylog}(n)$ layers, such that with probability at least $1 - \mathcal{O}(1/n')$ over the local noise, the output state is $(r,d/2)$-deviated from the state
    \begin{equation}
        V^{\otimes n'} H^{\otimes n'} \ketbra{0^{n'}}{0^{n'}} H^{\otimes n'} (V^{\otimes n'})^\dagger \label{eq:state1}
    \end{equation}
    for $r = \log\log(n')$.
\end{lemma}
\begin{proof}
    By Theorem~\ref{thm:Doritfault1}, there is a $\lambda$-noisy circuit with $\mathrm{polylog}(n')$ layers and $n'\mathrm{polylog}(n')$ qubits, which can initialize qubits at any time, such that with probability at least $1 - \mathcal{O}(1/n')$ over the local noise, the output state is $(r,d/2)$-deviated from the state in Eq.~\eqref{eq:state1}. By Theorem~\ref{thm:majority}, for each qubit that is initialized at some arbitrary time $t \le \mathrm{polylog}(n')$, we can append to $Q'$ a $\lambda$-noisy circuit of depth $t$ operating on $3^t$ qubits initialized at time zero, such that the first qubit of the output state of this circuit can play the role of the qubit initialized at time $t$ in $Q'$. Altogether, we obtain a circuit whose depth is the same as $Q'$ but which now operates on at most $3^{\mathrm{polylog}(n')}\cdot n'\mathrm{polylog}(n')$ qubits. This is at most $\poly(n)$ provided that the constant $c$ in the in the Lemma is sufficiently small.
\end{proof}

\noindent Using the same proof as the above lemma, we can establish the following.

\begin{lemma}\label{lem:stateprep2}
    Suppose $n' \le \exp(\log^c n)$ for $0 < c < 1$ a sufficiently small constant. There exists an absolute constant $\lambda_c \in [0,1]$ such that for any non-negative $\lambda < \lambda_c$, there exists a $\lambda$-noisy quantum circuit which operates on $n'$ input qubits and $\poly(n)$ ancillary qubits and has $\mathrm{polylog}(n)$ layers, such that with probability at least $1 - \mathcal{O}(1/n')$ over the local noise, the output state is $(r,d/2)$-deviated from the state
    \begin{equation}
        V^{\otimes n'} \ketbra{0^{n'}}{0^{n'}} (V^{\otimes n'})^\dagger \label{eq:state2}
    \end{equation}
    for $r = \log\log(n')$.
\end{lemma}

\noindent Next, we use Theorems~\ref{thm:Doritfault1} and~\ref{thm:majority} to argue that we can take any state on $m^r n'$ qubits which is not too deviated from a codeword, noisily apply Hadamard transform, and apply fault-tolerant measurement to the result. This realizes an encoded version of the last step of Simon's algorithm. The noisy quantum circuit for this uses $\poly(n)$ total qubits including ancillas.

\begin{lemma}
\label{lem:endofprotocol}
    Suppose $n' \le \exp(\log^c n)$ for $0 < c < 1$ a sufficiently small constant and for $r = \mathrm{polylog}(n')$. There exists an absolute constant $\lambda_c\in[0,1]$ such that for any non-negative $\lambda < \lambda_c$, there exists a $\lambda$-noisy quantum circuit $Q'$ which operates on $m^r n'$ input qubits and $\poly(n)$ ancillary qubits and has $\mathrm{polylog}(n')$ layers, such that the following holds. Let $\mathcal{A}$ be the classical post-processing procedure based on recursive majority vote from Theorem~\ref{thm:Doritfault2}.

    For $k$ satisfying $k \le d$, let input state $\rho$ be $(r,k)$-deviated from the state
    \begin{equation}
        V^{\otimes n'} \sigma (V^{\otimes n'})^\dagger.
    \end{equation}
    Let $\mathcal{D}$ be the classical distribution over $\brc{0,1}^{n'}$ generated by measuring
    \begin{equation}
        H^{\otimes n'} \sigma H^{\otimes n'} \label{eq:state2b}
    \end{equation}
    in the computational basis. Then if one applies $Q'$ to $\rho$, noisily measures the output state under the computational basis, and applies $\mathcal{A}$ to the classical outcome, then the resulting distribution over $\brc{0,1}^{n'}$ is identical to $\mathcal{D}$ with probability at least $1 - \mathcal{O}(1/n')$ over the local noise.
\end{lemma}

\begin{proof}
    Because $\rho$ is $(r,k)$-deviated from $V^{\otimes n'} \sigma (V^{\otimes n'})^\dagger$ for $k\le d$, we can apply Theorem~\ref{thm:Doritfault2} to obtain a $\lambda$-noisy quantum circuit $Q''$ operating on $m^r n'$ qubits with $\mathrm{polylog}(n')$ layers which satisfies the desiderata of the lemma. The caveat is that the circuit has to be able to initialize ancilla qubits at any time.
    
    We can address this similarly to in the proof of Lemma~\ref{lem:stateprep}. By Theorem~\ref{thm:majority}, for each qubit that is initialized at some arbitrary time $t \le \mathrm{polylog}(n')$, we can append to $Q''$ a $\lambda$-noisy circuit of depth $t$ operating on $3^t$ qubits initialized at time zero, such that the first qubit of the output state of this circuit can play the role of the qubit initialized at time $t$ in $Q''$. Altogether, we obtain a circuit $Q'$ whose depth is the same as $Q''$ but which now operates on at most $3^{\mathrm{polylog}(n')} \cdot n' \mathrm{polylog}(n') \le \poly(n)$ qubits.
\end{proof}

\subsubsection{Stability against deviation in the robustified Simon's oracle}

So far, Lemma~\ref{lem:stateprep} implies that we can realize the first step of Simon's in an encoded fashion: noisily prepare a state, call it $\rho_1$, which is only slightly deviated from an encoding of the all-plus state. And Lemma~\ref{lem:endofprotocol} implies that given a state, call it $\rho_2$, which is only slightly deviated from the output of the robustified Simon's oracle, we can simulate the last step of Simon's algorithm in the presence of noise. In order to apply Lemma~\ref{lem:endofprotocol} however, it remains to verify that when we go from $\rho_1$ to $\rho_2$ by invoking the robustified oracle in the second step of Simon's, the sparsity of the deviations in $\rho_1$ is preserved. We show this in Lemma~\ref{lem:stability1} below.

First, recall the unitary $U_{\wt{f_s}}$ from Eq.~\eqref{eq:oracleaction1}.
Note that by construction, $\wt{f}_s$ only depends non-trivially on the first $m^r n' < n$ bits of its input.
By defining $\wt{f}_s^{*}: \{0, 1\}^{m^r n'} \to \{0, 1\}^{m^r n'}$ to be the function $\wt{f}_s$ restricted to the first $m^r n'$ bits, we can rewrite~\eqref{eq:oracleaction1} as
\begin{equation}
\label{E:oracleaction2}
    U_{\wt{f_s}} \ket{x_1} \ket{x_2}\ket{y} = \ket{x_1} \ket{x_2} \ket{y \oplus {\wt{f}_s^{*}}(x_1)}, \quad \forall x_1 \in \{0, 1\}^{m^r n'},\,x_2 \in \{0,1\}^{n - m^r n'},\,y \in \{0, 1\}^{m^r n'}.
\end{equation}
We see from the above equation that $U_{\wt{f_s}}$ acts trivially on the $\ket{x_2}$ part of the input state.  So let us define $U_{\wt{f}_s^*}$ as the restriction of $U_{\wt{f}_s}$ to its $\ket{x_1}$ and $\ket{y}$, subsystems, namely
\begin{equation}
\label{E:oracleaction3}
    U_{\wt{f}_s^*} \ket{x_1} \ket{y} = \ket{x_1} \ket{y \oplus {\wt{f}_s^{*}}(x_1)}, \quad \forall x_1 \in \{0, 1\}^{m^r n'},\,y \in \{0, 1\}^{m^r n'}.
\end{equation}
With the above notations for the oracle, we can now prove the following lemma showing the classical function $\wt{f}_s^*$ preserves the deviation metric.
We note that, on the other hand, the ordinary Simon's function $f_s$ does not have the same property.

\begin{lemma}[Stability of the robustified classical oracle]
\label{lem:stability1}
Consider a Hilbert space $\mathcal{H}$ which decomposes into subsystems as $\mathcal{H} \simeq \mathcal{H}_{\text{\rm main},\,1} \otimes \mathcal{H}_{\text{\rm anc},\,1}$ where $\mathcal{H}_{\text{\rm main},\,1} \simeq (\mathbb{C}^2)^{\otimes(m^r n')}$ and $\mathcal{H}_{\text{\rm anc},\,1} \simeq (\mathbb{C}^2)^{\otimes (m^r n')}$.  Further let $\rho^0 = V^{\otimes n'} H^{\otimes n'} |0^{n'}\rangle\langle 0^{n'}| H^{\otimes n'} (V^{\otimes n'})^\dagger$ and $\sigma^0 = V^{\otimes n'} |0^{n'}\rangle\langle 0^{n'}| (V^{\otimes n'})^\dagger.$

Given any $k \leq d$, if $\rho \otimes \sigma$ is $(r, k)$-deviated from $\rho^0 \otimes \sigma^0$, then
$\text{\rm tr}_{\mathcal{H}_{\text{\rm anc},\,1}}\!\!\left\{U_{\wt{f}_s^*} (\rho \otimes \sigma) U_{\wt{f}_s^*}^\dagger\right\}$ is $(r, k)$-deviated from $\text{\rm tr}_{\mathcal{H}_{\text{\rm anc},\,1}}\!\!\left\{U_{\wt{f}_s^*} (\rho^0 \otimes \sigma^0) U_{\wt{f}_s^*}^\dagger\right\}$.
\end{lemma}

\begin{proof}
Consider the set $S \subset \{0,1\}^{m^r n'}$ defined by
\begin{equation}
S \triangleq \{t \in \{0,1\}^{m^r n'} \, \big| \, t = q_1 \cdots q_{n'}\,,\, q_i \in B_{0}^{(r)} \cup B_{1}^{(r)}\}\,.
\end{equation}
Moreover, if $A = \{a_1, a_2, ..., a_{|A|}\}$ is a subset of $\{1,2,...,m^r n'\}$ where $a_1 < a_2 < \cdots < a_{|A|}$, then further define
\begin{equation}
S_A \triangleq \{u \in \{0,1\}^{|A|} \, \big| \, u = t|_{A}\,\text{ for some } t \in S\}\,.
\end{equation}
Here $t|_A$ means the restriction of $t$ to its bits in locations $a_1,a_2,...,a_{|A|}$.

Because $\rho$ is $(r,k)$-deviated from $\rho^0$, there exists an $(r, k)$-sparse subset $A$ of qubits, such that $\rho_{A^c} = \rho^0_{A^c}$. Decomposing the Hilbert space $\mathcal{H}_{\text{main},\,1}$ as $\mathcal{H}_{\text{main},\,1} \simeq \mathcal{H}_A \otimes \mathcal{H}_{A^c}$, then using our notations above we can write
\begin{equation}
\rho_{A^c}^0 = \rho_{A^c} = \sum_{t, t' \in S_{A^c}} c_{t, t'}|t\rangle \langle t'|
\end{equation}
for some constants $c_{t, t'} \in \mathbb{C}$.  We further know, by hypothesis, that we can write
\begin{equation}
\rho^0 = \sum_{t, t'\in S_{A^c}} O_{t, t'}^{A} \otimes |t\rangle \langle t'|
\end{equation}
for some operators $O_{t, t'}^{A}$ such that $\text{tr}(O_{t, t'}^{A^c}) = c_{t, t'}$.  If all we know about $\rho$ is that it is $(r,k)$-deviated from $\rho^0$, then \emph{a priori} we only have the more general decomposition
\begin{equation}
\label{E:rhodecomp1}
\rho = \sum_{t, t'\in S_{A^c}} \widetilde{O}_{t, t'}^{A} \otimes |t\rangle \langle t'| + \sum_{\substack{t, t'\in \{0,1\}^{|A^c|} \\ \text{one of }t, t'\text{ is not in }S_{A^c}}} Q_{t, t'}^{A} \otimes |t\rangle \langle t'|\,,
\end{equation}
where $\text{tr}(\widetilde{O}_{t, t'}^{A}) = c_{t, t'}$ such that $t, t' \in S_{A^c}$, and $\text{tr}(Q_{t, t'}^A) = 0$ for all $t, t' \in \{0,1\}^{|A^c|}$ such that one of $t, t'$ is not in $S_{A^c}$.
But observe that because
\begin{equation}
    \sum_{\substack{t, t'\in \{0,1\}^{|A^c|} \\ t, t'\not\in S_{A^c}}} Q_{t, t'}^{A} \otimes |t\rangle \langle t'|
\end{equation}
has trace zero and is positive semi-definite (as it is obtained by left- and right-multiplying $\rho$ by the projector $\sum_{t\not\in S_{A^c}} |t\rangle \langle t|$\,), we must have $Q^A_{t, t'} = 0$ for all $t, t'\not\in S_{A^c}$. On the other hand, if $\ket{1},\ldots,\ket{|A|}$ is an orthonormal basis for $\mathcal{H}_A$, then the fact that the $2\times 2$ principal minors of $\rho$ must be nonnegative implies that for any $1\le i, j \le |A|$ and any $t\in S_{A^c}$, $t'\not\in S_{A^c}$ we have
\begin{equation}
    (\bra{i}\otimes\bra{t})\rho(\ket{i}\otimes\ket{t}) \cdot (\bra{j}\otimes\bra{t'})\rho(\ket{j}\otimes\ket{t'}) \ge (\bra{i}\otimes\bra{t})\rho(\ket{j}\otimes\ket{t'}) \cdot (\bra{j}\otimes\bra{t'})\rho(\ket{i}\otimes\ket{t})\,.
\end{equation}
Because we have already shown that $Q^A_{t',t'} = 0$, the left-hand side is zero. On the other hand, because $\rho$ is Hermitian, the right-hand side is nonnegative. The right-hand side is the squared magnitude of $(Q^A_{t, t'})_{ij}$, so we conclude that $Q^A_{t, t'} = 0$ for all $t, t'$ for which at least one of $t, t'$ is not in $S_{A^c}$. In other words, $\rho$ actually has the decomposition
\begin{equation}
    \rho = \sum_{t, t'\in S_{A^c}} \widetilde{O}_{t, t'}^{A} \otimes |t\rangle \langle t'|\,.
\end{equation}
From the definition of $\sigma^0$, we have
\begin{equation}
    \sigma^0 = \sum_{y, y' \in (B^{(r)}_0)^{n'}} \sigma^0_{y, y'} |y\rangle \langle y'|\,.
\end{equation}
We can then use the fact that $\sigma$ is $(r, k)$-deviated from $\sigma^0$ to show that
\begin{equation}
    \sigma = \sum_{y, y' \in (A^{(r)}_0)^{n'}} \sigma_{y, y'} |y\rangle \langle y'|\,.
\end{equation}
This is because the definition of $A^{(r)}_0$ ensures that it contains all bitstrings that can be generated by taking any bitstring in $B^{(r)}_0$ and flipping the bits in an $(r, d)$-sparse set of bits.
We are now ready to study the effect under the oracle $U_{\tilde{f}_s^*}$.

Writing $\wt{O}_{t, t'}^A = \sum_{u,u' \in \{0, 1\}^{|A|}} \wt{b}_{t,t',u,u'} |u\rangle \langle u'|$, we then have
\begin{equation}
\label{E:rdecomp1}
U_{\tilde{f}_s^*}(\rho \otimes \sigma)U_{\tilde{f}_s^*}^\dagger  = \sum_{\substack{u,u' \in \{0, 1\}^{|A|} \\ t,t'\in S_{A^c} \\ y, y' \in (A^{(r)}_0)^{n'}}} \wt{b}_{t,t',u,u'} \sigma_{y, y'} |u\rangle \langle u'| \otimes |t\rangle \langle t'| \otimes |y \oplus \wt{f}_s^*(ut)\rangle \langle y' \oplus \wt{f}_s^*(u't')|\,.
\end{equation}
Recall that $A$ is $(r,k)$-sparse. Since $k \leq d$, on account of~\eqref{E:tildefsdef0}, \eqref{E:tildefsdef1}, and the definition of $(r, k)$-sparseness, we have that $\wt{f}_s^*(ut)$ only depends on $t$ (and likewise $\wt{f}_s^*(u't')$ only depends on $t'$), and so in a slight abuse of notation we write $\wt{f}_s^*(ut) = \wt{f}_s^*(t)$ so that~\eqref{E:rdecomp1} becomes
\begin{equation}
\label{E:rdecomp2}
U_{\tilde{f}_s^*}(\rho \otimes \sigma)U_{\tilde{f}_s^*}^\dagger  = \sum_{\substack{u,u' \in \{0, 1\}^{|A|} \\ t,t'\in S_{A^c} \\ y, y' \in (A^{(r)}_0)^{n'}}} \wt{b}_{t,t',u,u'} \sigma_{y, y'} |u\rangle \langle u'| \otimes |t\rangle \langle t'| \otimes |y \oplus \wt{f}_s^*(t)\rangle \langle y' \oplus \wt{f}_s^*(t')|\,.
\end{equation}
We similarly write $\wt{f}_s^0(ut) = \wt{f}_s^0(t)$, and define $a_{t,t'} \triangleq \langle \wt{f}_s^0(t')| \wt{f}_s^0(t)\rangle$.
Consider the following two cases that cover all possibilities of $a_{t,t'}$.
\begin{enumerate}
    \item $a_{t, t'} = 0$: In this case,  $\wt{f}_s^0(t')$ and $\wt{f}_s^0(t)$ have at least one bit which differs. Say their $j$-th bits differ, i.e.~$[\wt{f}_s^0(t)]_j \neq [\wt{f}_s^0(t')]_j$. Let $[y]_{i,...,j}$ for $i \leq j$ denote the bitstring $(y_i, y_{i+1},...,y_{j})$, with similar notation for $y'$. For the two bitstrings $[y]_{j m^r+1,...,(j+1) m^r}, [y']_{j m^r+1,...,(j+1) m^r} \in A^{(r)}_0$, Lemma~\ref{lem:strucAr} gives
    \begin{align}
        [y]_{j m^r+1,...,(j+1) m^r} \oplus (\underbrace{[\wt{f}_s^0(t)]_j,...,[\wt{f}_s^0(t)]_j}_{m^r\text{ times}}) &\in A^{(r)}_{\wt{f}_s^0(t)_j} \\ [y']_{j m^r+1,...,(j+1) m^r} \oplus (\underbrace{[\wt{f}_s^0(t')]_j,...,[\wt{f}_s^0(t')]_j}_{m^r\text{ times}}) &\in A^{(r)}_{\wt{f}_s^0(t')_j}.
    \end{align}
    Using Lemma~\ref{lem:disjointAr} on the disjointness of $A^{(r)}_0$ and $A^{(r)}_1$, we have
    \begin{equation}
        [y \oplus \wt{f}_s^*(t)]_{jm^r+1 ,...,(j+1) m^r} \neq [y' \oplus \wt{f}_s^*(t')]_{j m^r+1,..., (j+1)m^r}\,.
    \end{equation}
    Thus, we have $\langle y' \oplus \wt{f}_s^*(t')| y \oplus \wt{f}_s^*(t)\rangle = 0$.
    \item $a_{t, t'} = 1$: All bits in $\wt{f}_s^0(t'), \wt{f}_s^0(t)$ are the same. Hence, $\langle y' \oplus \wt{f}_s^*(t')| y \oplus \wt{f}_s^*(t)\rangle = \langle y' | y \rangle$.
\end{enumerate}
Tracing out the $\mathcal{H}_{\text{anc},\,1}$ subsystem above gives
\begin{equation}
\label{E:toagree1}
\text{tr}_{\mathcal{H}_{\text{anc},\,1}}\!\left\{U_{\tilde{f}_s^*}(\rho \otimes \sigma)U_{\tilde{f}_s^*}^\dagger \right\} = \sum_{\substack{t, t'\in S_{A^c}}} a_{t, t'}\,\wt{O}_{t, t'}^A \otimes |t\rangle \langle t'| \sum_{y \in (A^{(r)}_0)^{\otimes n'}} \sigma_{y, y} = \sum_{\substack{t,t'\in S_{A^c}}} a_{t,t'}\,\wt{O}_{t,t'}^A \otimes |t\rangle \langle t'| \,.
\end{equation}
Since we likewise have
\begin{equation}
\label{E:toagree2}
\text{tr}_{\mathcal{H}_{\text{anc},\,1}}\!\left\{U_{\tilde{f}_s^*}(\rho^0 \otimes \sigma^0)U_{\tilde{f}_s^*}^\dagger \right\} = \sum_{\substack{t,t'\in S_{A^c}}} a_{t,t'}\,O_{t,t'}^A \otimes |t\rangle \langle t'| \,,
\end{equation}
we find that~\eqref{E:toagree1} and~\eqref{E:toagree2} agree upon taking the partial trace of each over $\mathcal{H}_A$.  Thus we find that $\text{\rm tr}_{\mathcal{H}_{\text{\rm anc},\,1}}\!\!\left\{U_{\wt{f}_s^*} (\rho \otimes \sigma) U_{\wt{f}_s^*}^\dagger\right\}$ is $(r, k)$-deviated from $\text{\rm tr}_{\mathcal{H}_{\text{\rm anc},\,1}}\!\!\left\{U_{\wt{f}_s^*} (\rho^0 \otimes \sigma^0) U_{\wt{f}_s^*}^\dagger\right\}$, as claimed.
\end{proof}

\subsubsection{Proof of super-polynomial separation between \texorpdfstring{$\NISQ$}{NISQ} and \texorpdfstring{$\BPP$}{BPP}}

We are now ready to complete the proof of the oracle separation between $\NISQ$ and $\BPP$.

\begin{proof}[Proof of Theorem~\ref{thm:superpoly1}]
Let us decompose our total Hilbert space $\mathcal{H}$ as $\mathcal{H} \simeq \mathcal{H}_{\text{main},\,1} \otimes \mathcal{H}_{\text{main},\,2} \otimes \mathcal{H}_{\text{anc},\,1} \otimes \mathcal{H}_{\text{anc},\,2}$ where
\begin{equation}
\mathcal{H}_{\text{main},\,1} \simeq (\mathbb{C}^2)^{\otimes (m^r n')}, \,\, \mathcal{H}_{\text{main},\,2} \simeq (\mathbb{C}^2)^{\otimes (n - m^r n')},\,\, \mathcal{H}_{\text{anc},\,1} \simeq (\mathbb{C}^2)^{\otimes (m^r n')},\,\, \mathcal{H}_{\text{anc},\,2} \simeq (\mathbb{C}^2)^{\otimes \mathcal{O}(\text{poly}(n))}.
\end{equation}
We begin with a state on $\mathcal{H}$ initialized in the all-zero state.  By Lemma~\ref{lem:stateprep}, we can prepare a state $\rho$ on $\mathcal{H}_{\text{main},\,1}$, using the ancillas on $\mathcal{H}_{\text{anc},\,2}$, such that $\rho$ is $(r,d/2)$-deviated from $\rho^0 = V^{\otimes n'} H^{\otimes n'}\ketbra{0^{n'}}{0^{n'}} H^{\otimes n'} (V^{\otimes n'})^\dagger$.
By Lemma~\ref{lem:stateprep2}, we can prepare a state $\sigma$ on $\mathcal{H}_{\text{anc},\,1}$, using the ancillas on $\mathcal{H}_{\text{anc},\,2}$, such that $\sigma$ is $(r,d/2)$-deviated from $\sigma^0 = V^{\otimes n'}\ketbra{0^{n'}}{0^{n'}} (V^{\otimes n'})^\dagger$.

At this point in the algorithm, our qubits on $\mathcal{H}_{\text{main},\,2}$ are no longer in the all-zero state due to the local noise.
We do not care what the state is and suppose that the state is given by $\rho^{(2)}$.
We proceed by applying our oracle unitary $U_{\wt{f}_s}$ to $(\rho \otimes \rho^{(2)}) \otimes \sigma$ on $\mathcal{H}_{\text{main},\,1} \otimes \mathcal{H}_{\text{main},\,2} \otimes \mathcal{H}_{\text{anc},\,1}$.
Since the oracle unitary acts as the identity on $\mathcal{H}_{\text{main},\,2}$ by construction, we can equivalently just apply $U_{\wt{f}_s^*}$ to $\rho \otimes \sigma$ on $\mathcal{H}_{\text{main},\,1} \otimes \mathcal{H}_{\text{anc},\,1}$. Doing so and subsequently neglecting the $\mathcal{H}_{\text{anc},\,1}$ register (corresponding to tracing out the qubits), we obtain
\begin{equation}
\rho' = \text{tr}_{\mathcal{H}_{\text{anc},\,1}}\!\left\{U_{\tilde{f}_s^*}(\rho \otimes \sigma)U_{\tilde{f}_s^*}^\dagger \right\}\,.
\end{equation}
But by Lemma~\ref{lem:stability1}, this state is only $(r,d/2)$-deviated from
\begin{equation}
\rho^1 = \text{tr}_{\mathcal{H}_{\text{anc},\,1}}\!\left\{U_{\tilde{f}_s^*}(\rho^0 \otimes \sigma^0)U_{\tilde{f}_s^*}^\dagger \right\}\,.
\end{equation}

If $f_s$ is a 1-to-1 function, then
\begin{equation}
\rho^1 = V^{\otimes n'}\left(\frac{1}{2^{n'}}\sum_{z \in \{0,1\}^{n'}} |z\rangle \langle z|\right) (V^{\otimes n'})^\dagger\,,
\end{equation}
whereas if $f_s$ is a 2-to-1 function we have
\begin{equation}
\rho^1 = V^{\otimes n'}\left(\frac{1}{2^{n'}}\sum_{z \in \{0,1\}^{n'}} \frac{1}{\sqrt{2}}\big(|z\rangle + |z \oplus s\rangle\big) \cdot \frac{1}{\sqrt{2}} \big(\langle z| + \langle z \oplus s|\big)\right) (V^{\otimes n'})^\dagger
\end{equation}
where $s$ is the hidden string. Applying Hadamards to the encoded qubits of $\rho'$, measuring in the computational basis, and applying classical post-processing via recursive majority vote as per Lemma~\ref{lem:endofprotocol}, we will obtain an $n'$ bit string $z_0$ which with probability $1 - \mathcal{O}(1/n')$ is sampled from the distribution $\mathcal{D}$ defined as follows.  If $f_s$ is 1-to-1 function then $\mathcal{D}$ will be the uniform distribution over $n'$ bit strings, whereas if $f_s$ is a 2-to-1 function then $\mathcal{D}$ will be the uniform distribution over $n'$ bit strings subject to the constraint $z_0 \cdot s = 0\,\,(\text{mod }2)$.  


If we repeat the entire procedure $n'$ times, then with probability $(1 - \mathcal{O}(1/n'))^{n'} = \Omega(1)$ we obtain $n'$ such bit strings $z_0, z_1,...,z_{n'-1}$.  If this event, call it $\mathcal{E}$, happens, then by solving the $n'$ linear equations $z_i \cdot s = 0\,\,(\text{mod }2)$ for $i=0,1,...,n'-1$, we can determine whether $s$ is the all-zero string meaning $f_s$ is 1-to-1, or some non-trivial string in which case $f_s$ is 2-to-1. In general, if $\mathcal{E}$ does not happen and we have obtained some arbitrary string $s$, we can check that this situation is the case by querying the classical oracle at $f_s(0)$ and $f_s(s)$. So by repeating the entire procedure $\mathcal{O}(\log(1/\delta))$ times, with probability at least $1 - \delta$ the event $\calE$ will happen at least once, and we will be able to determine if $f_s$ is 1-to-1 or 2-to-1.
\end{proof}

\begin{remark}\label{remark:superpoly}
    As alluded to at the beginning of this section, the proof above applies verbatim to the stronger noise model where at every layer, every qubit is independently corrupted with probability $\lambda$ by an adversary. As a result, although our definition of $\NISQ$ pertains to local depolarizing noise, the oracle separation between $\BPP$ and $\NISQ$ holds even when the local noise could be adversarially chosen.
\end{remark}

\subsection{\texorpdfstring{$\NISQ$}{NISQ} vs. \texorpdfstring{$\BQP$}{BQP}}
\label{subsec:lift_lower}


In this section we show an oracle separation between $\NISQ$ and $\BQP$ via a simple ``lifting'' of Simon's problem. In fact, we will actually be able to separate $\NISQ$ and $\BPP^{\QNC^0}$ relative to this oracle.

We begin by describing the modification of Simon's problem we will consider. For $n\in\mathbb{N}$, given a function $f: \brc{0,1}^n\to\brc{0,1}^n$, we define the \emph{lift} of $f$ to be the function $\wt{f}: \brc{0,1}^{2n}\to\brc{0,1}^n$ given by
\begin{equation}
    \wt{f}(x) \triangleq \begin{cases}
        f(x_1,\ldots,x_n) & \ \text{if} \ x_{n+1},\ldots,x_{2n} = 0 \\
        0 & \text{otherwise}
    \end{cases}.
\end{equation}
Given lifted function $\wt{f}$, we will abuse notation and let $O_{\wt{f}}$ denote both the classical oracle given by evaluating $\wt{f}$ as well as the quantum oracle
\begin{equation}
    O_{\wt{f}}: \ket{x}\ket{y} \mapsto \ket{x} \ket{y\oplus\wt{f}(x)}.
\end{equation}

It is not hard to see that in the absence of depolarizing noise, a minor modification of Simon's algorithm, which can be implemented in $\BPP^{\QNC^0}$, still works under this lifting. In contrast, for $\NISQ$ algorithms, we show the following:

\begin{theorem}\label{thm:generic_slowdown}
    Let $\lambda\in[0,1]$. Any $\NISQ_\lambda$ algorithm which, given oracle access to $O_{\wt{f}}$ for any lift of a function $f:\brc{0,1}^n\to\brc{0,1}^n$ which is either 2-to-1 or 1-to-1, can determine whether $f$ is 2-to-1 or 1-to-1 with constant advantage must have query complexity at least $\exp(\Omega(\lambda n))$. Thus, relative to oracles $O$ of this form, $\NISQ^O \subsetneq \BQP^O$. 
\end{theorem}

\noindent Our lifting operation is reminiscent of the shuffling Simon's problem introduced in \cite{chia2020need} to give an oracle separation between $\BPP^{\QNC^0_d}$ and $\BPP^{\QNC^0_{2d+1}}$. As we show in Appendix~\ref{app:shuffle}, the shuffling Simon's problem can also be used to separate $\NISQ$ from bounded-depth noiseless quantum computation. For instance, this implies the existence of an oracle relative to which $\NISQ\cup \BPP^{\QNC} \subsetneq \BQP$.


We now proceed to the proof of Theorem~\ref{thm:generic_slowdown}. We begin by recording the following basic fact about local depolarizing noise, whose proof we defer to Appendix~\ref{app:defer_lemproj}.

\begin{lemma}\label{lem:proj}
    Given $n'\in\N$, let $\Omega$ denote some subset of $\brc{0,1}^{n'}$, and let $\Pi$ denote the projection to the span of $\brc{\ket{x}}_{x\in\Omega}$. Then for any $\lambda\in[0,1]$ and any $n'$-qubit state $\ket{\psi}$,
    \begin{equation}
        \Tr(\Pi D^{\otimes n'}_\lambda[\ket{\psi}\bra{\psi}]) \le \sup_D \Pr[a\sim D, \tilde{a}]{\tilde{a} \in\Omega}, \label{eq:probflip}
    \end{equation}
    where the supremum is over probability distributions over $\brc{0,1}^{n'}$, and $\tilde{a}$ is the random string obtained by flipping each of the bits of $a$ independently with probability $\lambda/2$.
\end{lemma}

\noindent Note the probability on the right-hand side of \eqref{eq:probflip} is exponentially small when $\Omega\subset\brc{0,1}^{2n}$ is the set of strings $x$ for which $x_{n+1},\ldots,x_{2n} = 0$. We will now use this to show that the distribution over measurement outcomes from running a noisy quantum circuit that has query access to either $O_{\wt{f}}$ or the identity oracle $\mathrm{Id}$ gives very little information about which oracle the circuit has access to.

\begin{lemma}\label{lem:test_slowdown}
    Let $A$ be any $\lambda$-noisy quantum circuit which makes $N$ oracle queries. If $p_{\wt{f}}$ (respectively $p_\id$) is the distribution over the random string $s$ output by the circuit when the oracle is $O_{\wt{f}}$ (respectively the identity oracle $\mathrm{Id}$), then $\tvd(p_{\wt{f}}, p_\id) \le N\exp(-\Omega(\lambda n))$.
\end{lemma}

\begin{proof}
    Let $n'$ denote the number of qubits on which $A$ operates. For convenience, we denote by $\wh{O}_{\wt{f}}$ the channel given by pre-composing $O_{\wt{f}}$ with $D^{\otimes 3n}_\lambda$. We will show that for all $n'$-qubit pure states $\sigma$, $\norm{(\wh{O}_{\wt{f}}\otimes D^{\otimes n'-3n}_\lambda - D^{\otimes n'}_\lambda)[\sigma]}_\tr$ is small so that we can apply Lemma~\ref{lem:hybrid_basic}.
    
    When $\Omega\subset\{0,1\}^{2n}$ is given by all strings whose last $n$ bits are $0$, then for any $a\in\Omega$, if $\tilde{a}$ is obtained by flipping each of the bits of $a$ independently with probability $\lambda/2$, then $\Pr{\tilde{a}\in\Omega} \le (1 - \lambda/2)^n \le \exp(-\lambda n/2)$. So by Lemma~\ref{lem:proj}, if $D^{\otimes n'}_\lambda [\sigma] = \sum_i \lambda_i \ket{v_i}\bra{v_i}$, then $\sum_i \lambda_i \norm{\Pi' v_i}^2 \le \exp(-\lambda n/2)$, where $\Pi'$ is the projection to the span of $\brc{\ket{x}\ket{y}\ket{w}}_{x\in\Omega, y\in\{0,1\}^n, w\in\{0,1\}^{n'-3n}}$. If we write every $v_i$ as $\sum_{x\in\{0,1\}^{2n},y\in\{0,1\}^n,w\in\{0,1\}^{n'-3n}} v_{i,x,y,w} \ket{x}\ket{y}\ket{w}$, then 
    \begin{equation}
        \sum_i \lambda_i \sum_{x\in\Omega,y,w} v_{i,x,y,w}^2 \le  \exp(-\lambda n/2).
    \end{equation}
    and we see that $O_{\wt{f}}\otimes \mathrm{Id}$ maps $\ket{v_i}\bra{v_i}$ to 
    \begin{multline}
        \sum_{x,x'\in\Omega,y,y',w,w'} v_{i,x,y,w} v_{i,x',y',w'}\ket{x,y\oplus f(g(x)),w}\bra{x',y'\oplus f(g(x')),w'} \\ 
        + \sum_{x\in\Omega,x'\not\in\Omega,y,y',w,w'} v_{i,x,y,w} v_{i,x',y',w'}\ket{x,y\oplus f(g(x)),w}\bra{x',y',w'} \\
        + \sum_{x\not\in\Omega,x'\in\Omega,y,y',w,w'} v_{i,x,y,w} v_{i,x',y',w'}\ket{x,y,w}\bra{x',y'\oplus f(g(x')),w'} \\
        + \sum_{x,x'\not\in\Omega,y,y',w,w'} v_{i,x,y,w} v_{i,x',y',w'}\ket{x,y,w}\ket{x',y',w'}.
    \end{multline}
    In particular, 
    \begin{align}
        \MoveEqLeft
        \norm{(O_{\wt{f}}\otimes \mathrm{Id} - \mathrm{Id})[\ket{v_i}\bra{v_i}]}_\tr \\
        &\le \sqrt{2} \norm{(O_{\wt{f}}\otimes \mathrm{Id} - \mathrm{Id})[\ket{v_i}\bra{v_i}]}_F \\
        &\le 2\sqrt{2}\Bigl(\sum_{x\in\Omega \ \text{or} \ x'\in\Omega, y,y',w,w'} v^2_{i,x,y,w}v^2_{i,x',y',w'}\Bigr)^{1/2} \\
        &\le 2\sqrt{2}\Bigl(1 - \Bigl(1 - \sum_{x\in\Omega,y,w} v^2_{i,x,y,w}\Bigr)^2\Bigr)^{1/2} \le 4\sqrt{\sum_{x\in\Omega,y,w} v^2_{i,x,y,w}}.
    \end{align}
    By Jensen's inequality, we can bound $\norm{(\wh{O}_{\wt{f}}\otimes D^{\otimes n'-3n}_\lambda - D^{\otimes n'}_\lambda)[\sigma]}_\tr$ by 
    \begin{equation}
        4\sum_i \lambda_i \sqrt{\sum_{x\in\Omega,y,w} v^2_{i,x,y,w}} \le 4\sqrt{\sum_i \lambda_i \sum_{x\in\Omega,y,w} v^2_{i,x,y,w}} \le 4\exp(-\lambda n/4).
    \end{equation}
    By taking the channels $\calE_1$ and $\calE_2$ in Lemma~\ref{lem:hybrid_basic} to be $\wh{O}_{\wt{f}}\otimes D^{\otimes n'-3n}_\lambda$ and $D^{\otimes n'}_\lambda$, we obtain the desired bound on $\tvd(p_{\wt{f}}, p_\id)$.
\end{proof}

\noindent We are ready to conclude the proof of Theorem~\ref{thm:generic_slowdown}. Roughly, because Lemma~\ref{lem:test_slowdown} tells us that running a noisy quantum circuit with oracle access gives negligible information about the underlying oracle, a $\NISQ$ algorithm with access to $O_{\wt{f}}$ is no more powerful than a classical algorithm with access to the corresponding classical oracle. The lower bound of Theorem~\ref{thm:generic_slowdown} then follows from the classical lower bound for Simon's problem.

\begin{proof}[Proof of Theorem~\ref{thm:generic_slowdown}]
    Let $\calT$ be the learning tree corresponding to a $\NISQ_\lambda$ algorithm that makes at most $N$ classical or quantum oracle queries to $O_{\wt{f}}$, as in Definition~\ref{def:tree_nisq}. By Lemma~\ref{lem:perturbchildren} and Lemma~\ref{lem:test_slowdown}, if we replace every noisy quantum circuit $A$ in the tree with a noisy quantum circuit $A'$ that makes queries to the identity oracle instead of to $O_{\wt{f}}$, then the new distribution over the leaves of $\calT$ is at most $N^2\exp(-\Omega(\lambda n))$-far in total variation from the original distribution $p_{O_{\wt{f}}}$; for $N = \exp(o(\lambda n))$, this quantity is $o(1)$. For convenience, denote this new distribution by $p'_f$.

    To apply Lemma~\ref{lem:lecam}, we wish to bound $\tvd(\E[f \ \text{1-to-1}]{p'_f},\E[f \ \text{2-to-1}]{p'_f})$. But note that because the quantum circuits $A'$ in the new learning tree are independent of the underlying function $f$, the learning tree is simply implementing a randomized classical query algorithm. We can thus think of $p'_f$ as a mixture over distributions $p'^r_f$ each corresponding to some fixing of the internal randomness $r$ of the algorithm (here the coefficients of the mixture are independent of $f$). It thus suffices to bound $\sup_r \tvd(\E[f \ \text{1-to-1}]{p'^r_f}, \E[f \ \text{2-to-1}]{p'^r_f})$. 
    
    Henceforth fix any $r$. The rest of the argument follows the standard proof of the classical lower bound for Simons' algorithm. The algorithm queries the classical oracle at some deterministic sequence of inputs $x_1,\ldots,x_a$, which we may assume without loss of generality are distinct and lie in $\Omega$. For any $y_1,\ldots,y_a$ which are all distinct, $(x_1,y_1),\ldots,(x_a,y_a)$ and $r$ determine some leaf node $\ell$ of the tree. The probability of this leaf node under $\E[f \ \text{1-to-1}]{p'^r_f}$ is $\frac{(2^n - a)!}{(2^n)!}$ and under $\E[f \ \text{2-to-1}]{p'^r_f})$ is $\frac{M}{2^n - 1}\frac{(2^n - a)!}{(2^n)!}$ where $M \triangleq 2^n - 1 - |\brc{x_i\oplus x_j \mid 1 \le i < j \le a}|$. For any $y_1,\ldots,y_a$ for which there is a collision, the probability of the corresponding leaf node under $\E[f \ \text{1-to-1}]{p'^r_f}$ is clearly $0$. We conclude that the total variation between these two mixtures is upper bounded by the probability that there is a collision among $f(x_1),\ldots,f(x_a)$ for a random 2-to-1 function $f$. The latter is at most
    \begin{equation}
        \sum^{a-1}_{i=0} \frac{i}{2^n - 1 - \binom{i}{2}} \le \frac{a^2}{2^{n+1} - 2 - a^2},
    \end{equation}
    so for $a \ll 2^{n/2}$, this quantity is $o(1)$. As $\min(\exp(\Omega(\lambda n)), 2^{n/2}) = \exp(\Omega(\lambda n))$, the theorem thus follows by Lemma~\ref{lem:lecam}. 
\end{proof}

\begin{remark}
    The reader may observe that apart from the classical lower bound for Simon's problem, our proof of the lower bound in Theorem~\ref{thm:generic_slowdown} makes very little use of the fact that $f$ is either a 2-to-1 or 1-to-1 function. In fact, the above argument shows more generally that for any search problem over a family of Boolean functions, the query complexity of any $\NISQ$ algorithm is essentially given by the classical query complexity for that problem.
\end{remark}

\section{Unstructured Search}
\label{sec:grover}

In this section we show that there is no quadratic speedup for unstructured search in $\NISQ$. Given $d\in\mathbb{N}$ and $i\in[d]$, we abuse notation and let $O_i$ denote both the classical oracle $O_i: [d]\to\brc{0,1}$ given by $O_i(x) = \bone{x = i}$ as well as the quantum oracle
\begin{equation}
    O_i: \ket{x}\ket{w} \mapsto (-1)^{\bone{x = i}} \ket{x}\ket{w} \ \ \forall \ x \in [d].
\end{equation} 

We remark that we have opted to work with the phase version of the classical oracle for convenience, since the two standard formulations of the oracle can simulate one another and so our lower bound will be unaffected.
We will also consider the identity oracle, which is classically given by $O_0(x) \triangleq x$ and quantumly given by $O_0: \ket{x}\ket{w}\mapsto \ket{x}\ket{w}$.  

Formally, we show the following:

\begin{theorem}\label{thm:main_grover}
    Let $\lambda\in[0,1]$. Any $\NISQ_\lambda$ algorithm which, given oracle access to $O_i$ for any $i\in[d]$, can determine $i$ with probability $2/3$ must have query complexity at least $\wt{\Omega}(d\lambda)$.
\end{theorem}


\noindent Our proof is composed of two parts. The first and main part is to establish a stronger result, namely a query complexity lower bound even for \emph{noiseless} bounded-depth quantum computation. The second part is to verify, essentially via an argument of \cite{aharonov1996limitations}, that this implies a lower bound for $\NISQ$.

\subsection{Proof preliminaries}

We will still work with the tree formalism from Definition~\ref{def:tree_nisq}, but because our focus now is on noiseless bounded-depth quantum computation, the definition simplifies somewhat:

\begin{definition}[Tree representation for bounded-depth algorithms] \label{def:tree_bounded_depth}
    Given oracle $O$, a noiseless depth-$T$ algorithm with access to $O$ can be associated with a pair $(\mathcal{T},\mathcal{A})$ as follows. The \emph{learning tree} $\mathcal{T}$ is a rooted tree, where each node in the tree encodes the transcript of all classical query and noisy quantum circuit results the algorithm has seen so far. The tree satisfies the following properties:
    \begin{itemize}[leftmargin=*,itemsep=0pt]
        \item Each node $u$ is associated with a value $p_O(u)$ corresponding to the probability that the transcript observed so far is given by the path from the root $r$ to $u$. In this way, $\calT$ naturally induces a distribution over its leaves; we abuse notation and denote this distribution by $p_O$. For the root $r$, $p_O(r) = 1$.
        \item At each non-leaf node $u$, we run a noiseless depth-$T$ quantum circuit $A$ with access to $O$. The children $v$ of $u$ are indexed by the possible $s\in\brc{0,1}^{n'}$ that could be obtained as a result. We refer to the edge between $u$ and $v$ as $(u,A,s)$. We denote by $\ket{\phi_O(A)}$ the output state of the circuit so that the probability of traversing $(u,A,s)$ from node $u$ to child $v$ is given by $|\braket{s|\phi_O(A)}|^2$. We define 
        \begin{equation}
            p_O(v) = p_O(u) \cdot |\braket{s | \phi_O(A)}|^2 \, .
        \end{equation}
        \item If the total number of queries to $O$ made along any root-to-leaf path is at most $N$, we say that the \emph{query complexity} of the algorithm is at most $N$.
    \end{itemize}    
    $\mathcal{A}$ is any classical algorithm that takes as input a transcript corresponding to any leaf node $\ell$ and attempts to determine the underlying oracle or predict some property thereof.
\end{definition}

\noindent We note that one minor distinction between the tree formalism for bounded-depth computation versus the one for $\NISQ$ is that we do not consider classical queries to $O$. The reason is that because the quantum circuits $A$ used at every node in Definition~\ref{def:tree_bounded_depth} are noiseless, we can simulate noiseless classical query access to $O$ using quantum query access.


In the sequel, for $O = O_i$ for $i \in \brc{0,\ldots,d}$, we will refer to $p_O$ and $\phi_O$ in Definition~\ref{def:tree_bounded_depth} as $p_i$ and $\phi_i$. 

We will make use of the following well-known bound:

\begin{lemma}[Eq. (7) in \cite{zalka1999grover}]\label{lem:zalka}
    For any quantum circuit $A$ for unstructured search that makes $T$ oracle queries, if $\ket{\phi_i(A)}$ denotes the output state when the underlying oracle is $O_i$, then
    \begin{equation}
        \sum^d_{i=1} \norm{\ket{\phi_i(A)} - \ket{\phi_0(A)}}^2 \le 4T^2.
    \end{equation}
\end{lemma}

\subsection{Lower bound against bounded-depth computation}

We now use these tools to prove the following query complexity lower bound. This will be the main component in our proof of Theorem~\ref{thm:main_grover}.

\begin{theorem}\label{thm:depth}
    There is an absolute constant $c > 0$ for which the following holds. Let $d,T\in\mathbb{N}$ with $T \le d$. Then no noiseless quantum algorithm of depth $T$ with query complexity at most $cd/T$ can, given oracle access to $O_i$ for any $i\in[d]$, output $i$ with probability $2/3$.
\end{theorem}

\subsubsection{Likelihood ratio calculations}

\noindent To prove Theorem~\ref{thm:depth}, our goal is to bound $\tvd(p_0,\mathbb{E}_{i\sim \brk{d}}p_i)$. To that end, we will analyze the likelihood ratio between these two distributions. Given any path $\vec{z} = ((u_1,A_1,s_1),(u_2,A_2,s_2),\ldots,(u_T,A_T,s_T))$ in the tree, the likelihood ratio $L_i(\vec{z})$ between traversing that path when the underlying oracle is $O_i$ versus when it is $O_0$ is given by
\begin{align}
    L_i(\vec{z}) &= \prod^n_{t=1} \frac{|\iprod{s_t | \phi_i(A_t)}|^2}{|\iprod{s_t | \phi_0(A_t)}|^2} \\
    &\ge \prod^n_{t=1} \brk*{1 + \frac{\iprod{s_t | \phi_i(A_t)} - \iprod{s_t | \phi_0(A_t)}}{\iprod{s_t | \phi_0(A_t)}} + \frac{\iprod{\phi_i(A_t) | s_t} - \iprod{\phi_0(A_t) | s_t}}{\iprod{\phi_0(A_t) | s_t}} }
\end{align}
For convenience, define
\begin{align}
    Y_i(u_t,A_t,s_t) &= \frac{\iprod{s_t | \phi_i(A_t)} - \iprod{s_t | \phi_0(A_t)}}{\iprod{s_t | \phi_0(A_t)}} + \frac{\iprod{\phi_i(A_t) | s_t} - \iprod{\phi_0(A_t) | s_t}}{\iprod{\phi_0(A_t) | s_t}}\\
    &= 2 \,\mathrm{Re}\left[ \frac{\iprod{s_t | \phi_i(A_t)} - \iprod{s_t | \phi_0(A_t)}}{\iprod{s_t | \phi_0(A_t)}} \right]
\end{align}
so that we have
\begin{equation}
    L_i(\vec{z}) \ge \prod^n_{t=1} \left(1 + Y_i(u_t,A_t,s_t) \right). \label{eq:prodY}
\end{equation}
Our goal is to show that, with respect to the distribution over paths $\vec{z}$ when the underlying oracle is $O_0$, $\E[i\sim\brk{d}]{L_i(\vec{z})}$ is with high probability not too small. This readily implies the desired upper bound on $\tvd(p_0,\mathbb{E}_{i\sim \brk{d}}p_i)$, as the latter satisfies 
\begin{equation}
    \tvd(p_0,\mathbb{E}_{i\sim[d]} p_i) = \E[\vec{z}\sim p_0]{\max(0,1 - \E[i]{L_i(\vec{z})})} \le \E[\vec{z}\sim p_0, i]{\max(0,1 - L_i(\vec{z}))}. \label{eq:LRtvbound}
\end{equation}
The idea for showing this will be to argue that the successive partial products in \eqref{eq:prodY} give rise to a \emph{multiplicative sub-martingale}\footnote{Equivalently, the partial sums $\sum^m_{t=1} \log(1 +Y_i(u_t,A_t,s_t))$ for $m = 1,\ldots,n$ give rise to a sub-martingale.} with suitably bounded increments $1 + Y_i(u_t,A_t,s_t)$, so that we can apply off-the-shelf martingale concentration inequalities.

Key to bounding these increments are the following moment bounds on $Y_i(u,A,s)$, as a random variable in $s$.

\begin{lemma}\label{lem:momentY}
    For any edge $(u,A,s)$ in the tree and any $i\in[d]$, we have that
    \begin{itemize}
        \item $\E[s]*{Y_i(u,A,s)} = -\norm{\ket{\phi_i(A)} - \ket{\phi_0(A)}}_2^2$,
        \item $\E[s]*{Y_i(u,A,s)^2} \leq 4 \norm{\ket{\phi_i(A)} - \ket{\phi_0(A)}}_2^2$,
    \end{itemize}
    Here the expectations are with respect to the distribution over measurement outcomes when the underlying oracle is $O_0$.
\end{lemma}

\begin{proof}
    Recalling that we observe outcome $s$ under this distribution with probability $|\braket{\phi_0(A)|s}|^2$, we have
    \begin{align}
        \E[s]{Y_i(u,A,s)} &= \sum_{s}   \iprod{\phi_0(A) | s} \left( \iprod{s | \phi_i(A)} - \iprod{s | \phi_0(A)} \right) \\
        &\qquad + \sum_{s}  \left( \iprod{\phi_i(A) | s} - \iprod{\phi_0(A) | s} \right) \iprod{s | \phi_0(A)} \\
        &= \bra{\phi_0(A)}\ket{\phi_i(A)} - 1 + \bra{\phi_i(A)}\ket{\phi_0(A)} - 1 \\
        &= -\norm{\ket{\phi_i(A)} - \ket{\phi_0(A)}}^2,
    \end{align}
    which establishes the first statement. For the second statement, we have
    \begin{align}
        \E[s]{Y_i(u,A,s)^2} &\le 4\E[s]*{\abs*{\frac{\bra{s}\ket{\phi_i(A)} - \bra{s}\ket{\phi_0(A)}}{\bra{s}\ket{\phi_0(A)}}}^2} \\
        &= 4 \sum_{s} \left(\iprod{\phi_i(A) | s } - \iprod{\phi_0(A) | s}\right) \left(\iprod{s | \phi_i(A)} - \iprod{s | \phi_0(A)}\right)\\
        &= 4 \norm{\ket{\phi_i(A)} - \ket{\phi_0(A)}}_2^2.\qedhere
    \end{align}
\end{proof}

\subsubsection{Good and balanced paths}

In this section we define two conditions on paths $\vec{z}$ of the tree (Definitions~\ref{def:igood} and~\ref{def:ibalanced}) under which we can show that Eq.~\eqref{eq:prodY} is not too small with high probability over paths for which these two conditions hold. We then prove that a random path in $\calT$ under $p_0$ will satisfy these conditions with high probability.

First, let $0 < \epsilon < 1/2$ be a small enough constant that we will set later. Given a path in the learning tree given by $\vec{z} = ((u_1,A_1,s_1),\ldots,(u_n,A_n,s_n))$, if each $A_t$ queries the oracle $T_t$ times, define the potential
\begin{equation}
    \tau(\vec{z}) \triangleq \sum_t T^2_t.
\end{equation}
Note that if every algorithm in the tree queries the oracle $T$ times, then $\tau(\vec{z}) = nT^2$; indeed, it may be helpful for the reader to focus on this case and think of $\tau(\vec{z})$ as $nT^2$ in the sequel. In general, H\"{o}lder's inequality implies the following basic fact:

\begin{lemma}\label{lem:holders}
    If $(\calT,\calA)$ specifies a noiseless quantum algorithm of depth $T$, then its query complexity is at least $\frac{1}{T}\max_{\vec{z}}\tau(\vec{z})$.
\end{lemma}

\noindent We are now ready to define our two conditions on paths:

\begin{definition}\label{def:igood}
    We say that an edge $(u,A,s)$ is \emph{$i$-good} if
    \begin{equation}
        Y_i(u,A,s) \ge -\epsilon.
    \end{equation}
    We say that a path $\vec{z}$ is \emph{$i$-good} if all of its constituent edges are $i$-good.
    
    Let $\Igood(\vec{z})$ (respectively $\Ibad(\vec{z})$) denote the set of indices $i\in[d]$ for which $\vec{z}$ is $i$-good (respectively not $i$-good).
\end{definition}

\begin{definition}\label{def:ibalanced}
    We say a path $\vec{z}$ is \emph{$i$-balanced} if its constituent edges $((u_1,A_1,s_1),\ldots,(u_n,A_n,s_n))$ satisfy
    \begin{equation}
        \sum^n_{t=1}\norm{\ket{\phi_i(A_t)} - \ket{\phi_0(A_t)}}^2 \le \max_{\vec{z}'} \frac{40000}{d}\tau(\vec{z}').
    \end{equation}
    Let $\Ibal(\vec{z})$ (respectively $\Iimbal(\vec{z})$) denote the set of indices $i\in[d]$ for which $\vec{z}$ is $i$-balanced (respectively not $i$-balanced).
\end{definition}

\noindent Intuitively, goodness of a path ensures that as one goes from one partial product of \eqref{eq:prodY} to the next, we never experience any significant multiplicative decreases. On the other hand, as we will see in the proof of Lemma~\ref{lem:applymartingale-gen} below, balancedness of a path ensures that the ``variance'' of these multiplicative changes is also not too large. These two conditions are important for applying off-the-shelf martingale concentration bounds like Freedman's inequality, which is governed by how large the changes can be in the worst case and how large they can be on average. 

We now argue that most paths are $i$-good and $i$-balanced for most $i\in[d]$ (Lemmas~\ref{lem:Igood} and \ref{lem:Ibal} below). To do this, we will need the following consequence of Lemma~\ref{lem:momentY}.

\begin{lemma}\label{lem:chebyshev-gen}
    For any edge $(u,A,s)$ in the tree and any $i\in[d]$, we have that $\Pr[s]{(u,A,s) \ \mathrm{not} \ i\mathrm{-good}} \le 12 \norm{\ket{\phi_i(A)} - \ket{\phi_0(A)}}_2^2 / \epsilon^2$. Here the probability is with respect to the distribution over measurement outcomes when the underlying oracle is $O_0$.
\end{lemma}

\begin{proof}
    Suppose that $\norm{\ket{\phi_i(A)} - \ket{\phi_0(A)}}^2 < \epsilon / 2$ (otherwise the claim vacuously holds). Then by Chebyshev's inequality,
    \begin{align}
        \Pr[s]{(u,A,s) \ \mathrm{not} \ i\mathrm{-good}} &= \Pr[s]{Y_i(u,A,s) < -\epsilon}\\
        &\leq \frac{\Var[s]{Y_i(u,A,s)}}{(\E[s]{Y_i(u,A,s)} + \epsilon)^2}\\
        &\leq \frac{3 \norm{\ket{\phi_i(A)} - \ket{\phi_0(A)}}_2^2}{(\epsilon - \norm{\ket{\phi_i(A)} - \ket{\phi_0(A)}}_2^2)^2}\\
        &\leq 12 \norm{\ket{\phi_i(A)} - \ket{\phi_0(A)}}_2^2 / \epsilon^2. \qedhere
    \end{align}
\end{proof}

\noindent 
We now combine Lemma~\ref{lem:zalka}, Lemma~\ref{lem:momentY}, and Lemma~\ref{lem:chebyshev-gen} to conclude that for a root-to-leaf path $\vec{z}$ sampled according to $p_0$, the probability that $\Igood(\vec{z})$ or $\Ibal(\vec{z})$ is small is low.

\begin{lemma}\label{lem:Igood}
    $\Pr[\vec{z} \sim p_0]*{|\Igood(\vec{z})| \le d - 4800\E[\vec{z}' \sim p_0]{\tau(\vec{z})} / \epsilon^2} \le 1/100$.
\end{lemma}

\begin{proof}
    Note that
    \begin{align}
        \E[\vec{z} \sim p_0]{|\Ibad(\vec{z})|} &= \sum^d_{i=1} \Pr[z]{\vec{z} \ \mathrm{not} \ i\mathrm{-good}} \\
        &\le \sum^d_{i=1} \sum^n_{t=1} \Pr{(u_t,A_t,s_t) \ \mathrm{not} \ i\mathrm{-good}} \\
        &= \sum^d_{i=1} \sum^n_{t=1} \Pr[s_t]{(u_t,A_t,s_t) \ \mathrm{not} \ i\mathrm{-good}} \\
        &\le \sum^d_{i=1} \sum^n_{t=1} \E[\vec{z} \sim p_0]*{12\norm{\ket{\phi_i(A_t)} - \ket{\phi_0(A_t)}}^2 / \epsilon^2} \\
        &\le \E[z\sim p_0]{48\tau(\vec{z})/\epsilon^2},
    \end{align}
    where the penultimate step follows by Lemma~\ref{lem:chebyshev-gen}, and the last step follows by Lemma~\ref{lem:zalka}. The lemma follows by the fact that $|\Ibad(\vec{z})| + |\Igood(\vec{z})| = d$ and by Markov's inequality.
\end{proof}

\begin{lemma}\label{lem:Ibal}
    $\Pr[\vec{z}\sim p_0]*{|\Ibal(\vec{z})| < 99d/100} \le 1/100$.
\end{lemma}

\begin{proof}
    Note that 
    \begin{align}
        \E[\vec{z}\sim p_0]*{|\Ibal(\vec{z})|} &= \sum^d_{i=1} \left(1 - \Pr[\vec{z}]*{\vec{z} \ \text{not} \ i \text{-balanced}}\right) \\
        &\ge \sum^d_{i=1} \left(1 - \E[\vec{z}\sim p_0]*{\frac{\sum^n_{t=1} \norm{\ket{\phi_i(A_t)} - \ket{\phi_0(A_t)}}_2^2}{\max_{\vec{z}'} 40000\tau(\vec{z}')/d}}\right) \\
        &\ge d - d\E[\vec{z}\sim p_0]*{\frac{\sum^d_{i=1}\sum^n_{t=1} \norm{\ket{\phi_i(A_t)} - \ket{\phi_0(A_t)}}^2}{\max_{\vec{z}'}40000\tau(\vec{z}')}} \\
        &\geq 9999d/10000.
    \end{align}
    where the second step follows by Markov's inequality and the last step follows by Lemma~\ref{lem:zalka}. The lemma follows by the fact that $|\Ibal(\vec{z})| + |\Iimbal(\vec{z})| = d$ and by Markov's inequality.
\end{proof}

\subsubsection{Martingale concentration}

\noindent The paths that are both $i$-balanced and $i$-good are the ones over which the log-likelihood $\log L_i(\vec{z})$ will concentrate as a random variable in $z\sim p_0$. As alluded to in the discussion above, being $i$-balanced ensures bounded variance, while being $i$-good ensures bounded differences. Together these yield the following Bernstein-type concentration which is the main technical ingredient in the proof of Theorem~\ref{thm:depth}:

\begin{lemma}\label{lem:applymartingale-gen}
    For any $i\in[d]$, consider the following sequence of random variables
    \begin{equation}
        \brc*{X_t \triangleq \log\left(1 + Y_i(u_t,A_t,s_t)\right) \cdot \bone{(u_t,A_t,s_t) \ i\mathrm{-good}}}^n_{t=1}, 
    \end{equation}
    where the randomness is with respect to $p_0$. For any $\eta,\nu > 0$, we have
    \begin{multline}
        \Pr[\vec{z}]*{\sum^n_{t=1} X_t \le -13\sum^n_{t=1} \norm{\ket{\phi_i(A_t)} - \ket{\phi_0(A_t)}}_2^2 - \eta \ \text{and} \ \sum^n_{t=1} \norm{\ket{\phi_i(A_t)} - \ket{\phi_0(A_t)}}_2^2 \le \nu}\\
        \le \exp\left(-\frac{\eta^2}{16\nu + 4 \epsilon \eta/3}\right). \label{eq:martingale-gen}
    \end{multline}
\end{lemma}

\begin{proof}
        Note that for any $(u,A,s)$,
    \begin{align}
        \MoveEqLeft\E[s]*{\log(1 + Y_i(u,A,s))\cdot \bone{(u,A,s) \ i\mathrm{-good}}} \\
        &\ge \E[s]*{\left(Y_i(u,A,s) - Y_i(u,A,s)^2\right)\cdot \bone{(u,A,s) \ i\mathrm{-good}}} \\
        &\ge \E[s]*{Y_i(u,A,s)-Y_i(u,A,s)^2} - \E[s]*{Y_i(u,A,s)\cdot \bone{(u,A,s) \ \mathrm{not} \ i\mathrm{-good}}} \\
        &\ge -5\norm{\ket{\phi_i(A)} - \ket{\phi_0(A)}}_2^2 - \E[s]*{Y_i(u,A,s)\cdot \bone{(u,A,s) \ \mathrm{not} \ i\mathrm{-good}}} \\
        &\ge -5\norm{\ket{\phi_i(A)} - \ket{\phi_0(A)}}_2^2 - \E[s]*{Y_i(u,A,s)^2}^{1/2}\cdot \Pr[s]{(u,A,s) \ \mathrm{not} \ i\mathrm{-good}}^{1/2} \\
        &\ge -13 \norm{\ket{\phi_i(A)} - \ket{\phi_0(A)}}_2^2. \label{eq:bound-gen}
    \end{align}
    where in the first step we used the fact that $\log(1 + z) \ge z - z^2$ for $z \ge -1/2$, in the third step we used Lemma~\ref{lem:momentY}, in the fourth step we used Cauchy-Schwarz, and in the last step we used the last part of Lemma~\ref{lem:momentY} and Lemma~\ref{lem:chebyshev-gen}.
    
    Additionally,
    \begin{align}
        \E[s]*{\log(1 + Y_i(u,A,s))^2\cdot \bone{(u,A,s) \ i\mathrm{-good}}} &\le \E[s]*{2Y_i(u,A,s)^2\cdot \bone{(u,A,s) \ i\mathrm{-good}}} \\
        &\le \E[s]*{2Y_i(u,A,s)^2}  \\
        &= 8 \norm{\ket{\phi_i(A)} - \ket{\phi_0(A)}}_2^2, \label{eq:varbound-gen}
    \end{align}
    where in the first step we used the fact that $\log(1+z)^2 \le 2z^2$ for $z \ge -1/2$, and in the last step we used the second part of Lemma~\ref{lem:momentY}.
    
    For every $t$, define the random variable
    \begin{equation}
        Z_t \triangleq \log\left(1 + Y_i(u_t,A_t,s_t)\right)\cdot \bone{(u_t,A_t,s_t) \ i\mathrm{-good}} + 13 \norm{\ket{\phi_i(A_t)} - \ket{\phi_0(A_t)}}_2^2,
    \end{equation}
    where the randomness is with respect to $p_0$.
    By Eq.~\eqref{eq:bound-gen}, $\brc{Z_t}_t$ is a submartingale difference sequence satisfying $Z_t \ge \log(1 - \epsilon) \ge -2\epsilon$ given $0 < \epsilon < 1/2$, so the lemma follows by Freedman's inequality and Eq.~\eqref{eq:varbound-gen}.
\end{proof}

\noindent We will take $\eta$ to be a small constant to be tuned later, and 
\begin{equation}
    \nu = (40000/d)\max_{\vec{z}}\tau(\vec{z}). \label{eq:nudef}
\end{equation}
In light of Lemma~\ref{lem:applymartingale-gen}, we introduce one more property of paths $\vec{z}$ which, together with goodness (Definition~\ref{def:igood}) and balancedness (Definition~\ref{def:ibalanced}), ensures that $\E[i\sim{[d]}]{L_i(\vec{z})}$ is not too small.

\begin{definition}
    We say that a root-to-leaf path $\vec{z}$ is $i$-concentrated if the sum as considered in Lemma~\ref{lem:applymartingale-gen} is not too negative, that is,
    \begin{equation}
        \sum^n_{t=1} X_t > -13\sum^n_{t=1}\norm{\ket{\phi_i(A_t)} - \ket{\phi_0(A_t)}}^2 - \eta,
    \end{equation}
    and/or the path is not $i$-balanced. Let $\Iconc(\vec{z})$ (respectively $\Idiv(\vec{z})$) denote the set of indices $i\in[d]$ for which $\vec{z}$ is $i$-concentrated (respectively not $i$-concentrated).
\end{definition}

\subsubsection{Completing the argument}

\noindent We assume the query complexity of the algorithm is at most $cd/T$ for a small constant $0 < c < 1$ to be tuned later. Note that by Lemma~\ref{lem:holders}, this implies that 
\begin{equation}
    \max_{\vec{z}} \le cd. \label{eq:taubound}
\end{equation}

\begin{lemma}\label{lem:Iconc}
    $\Pr[\vec{z}\sim p_0]{|\Iconc(\vec{z})| < 99d/100} \le 100 \exp\left(-\frac{\eta^2}{640000c + 4\epsilon\eta/3}\right)$.
\end{lemma}

\begin{proof}
    By Lemma~\ref{lem:applymartingale-gen}, our choice of $\nu$ in Eq. \eqref{eq:nudef}, and \eqref{eq:taubound}, for any $i\in[d]$ we have
    \begin{equation}
        \Pr{\vec{z} \ \text{is} \ i\text{-concentrated}} \ge 1 - \exp\left(-\frac{\eta^2}{640000c + 4\epsilon\eta / 3}\right).
    \end{equation}
    By linearity of expectation, $\E{|\Idiv(\vec{z})|} \le d\cdot \exp\left(-\frac{\eta^2}{640000c + 4\epsilon\eta / 3}\right)$, so the proof follows by the fact that $|\Iconc(\vec{z})| + |\Idiv(\vec{z})| = d$ and Markov's inequality.
\end{proof}

\noindent We are now ready to finish the argument.

\begin{proof}[Proof of Theorem~\ref{thm:depth}]
    By taking
    \begin{equation}
        \eta = 1/10 \qquad c = 10^{-10} \qquad \epsilon = 1/135,
    \end{equation}
    we can combine Lemma~\ref{lem:Igood},~Lemma \ref{lem:Ibal}, and Lemma \ref{lem:Iconc} to obtain
    \begin{align}
        \Pr[\vec{z}\sim p_0]*{|\Igood(\vec{z})| < (99/100) (N/G)} &\le 1/100,\\    \Pr[\vec{z}\sim p_0]*{|\Ibal(\vec{z})| < (99/100)(N/G)} &\le 1/100,\\
        \Pr[\vec{z}\sim p_0]*{|\Iconc(\vec{z})| < (99/100)(N/G)} &\le 1/100.
    \end{align}
    By union bound, with probability at least $0.97$ over $\vec{z}\sim p_0$, there are at least $0.97 (N/G)$ indices $i$ such that $\vec{z}$ is $i$-good, $i$-balanced, and $i$-concentrated.
    For an index $i$ that satisfies all three conditions for a path $\vec{z}$, we have
    \begin{align}
        \log(L_i(\vec{z}) &\ge \sum_t \log(1 +Y_i(u_t,A_t,s_t)) = \sum_t X_t \\
        &> -13\sum^n_{t=1}\norm{\ket{\phi_i(A_t)} - \ket{\phi_0(A_t)}}^2 - \eta \\
        &\ge -\frac{520000}{d}\max_{\vec{z'}}\tau(\vec{z}') - \eta \ge -1/10,
    \end{align}
    where in the second, third, and fourth steps we used that $\vec{z}$ is $i$-good, $i$-concentrated, and $i$-balanced respectively. Hence, $L_i(\vec{z}) \ge 9/10$ with probability at least $0.97^2$ given a random $\vec{z}\sim p_0$ and random $i\in[d]$. Therefore, by Eq.~\ref{eq:LRtvbound} above, we can bound the total variation distance by
    \begin{equation}
        0.97^2(1 - 0.905) + (1 - 0.97^2) \le 1/6.
    \end{equation}
    By Lemma~\ref{lem:lecam}, we conclude that provided the query complexity of the algorithm is at most $cd / T$, it cannot distinguish between the oracle $O_0$ and the oracle $O_i$ for a random choice of $i$ with constant advantage.
\end{proof}

\subsection{From bounded-depth to noisy computation}

We now show how to extract from Theorem~\ref{thm:depth} a lower bound against $\NISQ$. We begin with the following basic lemma, a proof of which we include in Appendix~\ref{app:defer_leminfo} for completeness, that quantifies the amount of information that is lost from running many layers of noisy computation:

\begin{lemma}[Lemma 8 from \cite{aharonov1996limitations}]\label{lem:info}
    Let $A$ be a $\lambda$-noisy depth-$T$ quantum circuit on $n$ qubits with output state $\rho$. Then $\mathcal{I}(\rho) \triangleq n - S(\rho) \le (1 - \lambda)^T\cdot n$, where $S(\cdot)$ denotes von Neumann entropy.
\end{lemma}

\noindent We will also use the following standard operational characterization of $I(\rho)$:

\begin{lemma}[See e.g. Lemma 2 from \cite{aharonov1996limitations}]\label{lem:info_POVM}
    Given any $n$-qubit state $\rho$ and any POVM, the distributions $p,q$ induced by respectively measuring $\rho$ and $\Id/2^n$ with the POVM satisfy $\KL{p}{q} \le \mathcal{I}(\rho)$.
\end{lemma}


\begin{proof}[Proof of Theorem~\ref{thm:main_grover}]
    Let $\calT$ be the learning tree corresponding to a $\NISQ$ algorithm which has access to $O_i$ for some $0\le i \le d$ and has query complexity $N$, as in Definition~\ref{def:tree_nisq}. Let $\overline{T}$ be some choice of depth that we will tune later. We will convert $\calT$ to a learning tree $\wh{\calT}$ corresponding to a bounded-depth noiseless $\NISQ$ algorithm, as in Definition~\ref{def:tree_bounded_depth}.
    
    Define $\wh{\calT}$ as follows. For every non-leaf node $u$, if the algorithm makes a single classical query at input $j$, then replace the edge $(u,x,O_i(x))$ to its child $v$ by an edge $(u,A,s)$ where $A$ is a depth-1 quantum algorithm simulating the classical query. On the other hand, suppose that at $u$, the algorithm runs some $\lambda$-noisy quantum circuit $A$ on $\poly(d)$ qubits. If $A$ makes fewer than $\overline{T}$ oracle queries in total, then consider the noiseless quantum circuit $A'$ which simulates $A$ by applying depolarizing noise at each layer. If $A$ makes more than $\overline{T}$ queries, then replace $A$ with the quantum circuit $A'$ that simply measures the maximally mixed state in the computational basis, rather than the output state $\ket{\phi_i(A)}$. By Lemma~\ref{lem:info_POVM} and Pinsker's inequality, the total variation distance between the induced conditional distributions on children when $\ket{\phi_i(A)}$ gets measured versus when the maximally mixed state gets measured is at most $\sqrt{\frac{1}{2} \mathcal{I}(\ket{\phi_i(A)})}$, and by Lemma~\ref{lem:info} this is at most $(1 - \lambda)^{\overline{T}/2} \cdot \mathcal{O}(\sqrt{\log d})$.

    Let $p_i$ (respectively $\wh{p}_i$) denote the distribution over leaves when running the $\NISQ$ algorithm given by $\calT$ (respectively the noiseless algorithm given by $\wh{\calT}$). As the number of oracle queries is at most $N$, the depth of both trees is at most $N$, so we conclude that the total variation distance between $p_i$ and $\wh{p}_i$ is at most $N(1 - \lambda)^{\overline{T}/2}\cdot \mathcal{O}(\sqrt{\log d})$ by Lemma~\ref{lem:perturbchildren}. We will take $\overline{T} = C\lambda^{-1}(\log\log d + \log N)$ for constant $C > 0$ so that this quantity is an arbitrarily small constant.
    
    We conclude by Theorem~\ref{thm:depth} that if $N \le cd/\overline{T} = \Theta(d\lambda / (\log\log d + \log N))$, then the $\NISQ$ algorithm given by $\calT$ cannot solve unstructured search with probability $2/3$. This concludes the proof of our $\wt{\Omega}(d\lambda)$ lower bound.
\end{proof}

\section{Bernstein-Vazirani Problem}

In this section we show that a $\NISQ$ algorithm can solve the Bernstein-Vazirani problem~\cite{bernstein1997quantum} with $\mathcal{O}(\log n)$ queries, whereas it is known that any classical algorithm requires $\Theta(n)$ queries. As with the upper bound in Section~\ref{sec:robust_simon}, we will show that our algorithm is robust not just to local depolarizing noise, but also to arbitrary local noise that occurs with sufficiently small constant rate (see Remark~\ref{remark:BV}).

We begin by recalling the Bernstein-Vazirani problem on $n$ bits.  There is an unknown function $f : \{0,1\}^n \to \{0,1\}$ of the form $f(x) = s \cdot x\,\,(\text{mod }2)$, where $s \in \{0,1\}^n$ is often called the \emph{hidden string}.  The goal is to determine the hidden string.  In the quantum context, the classical oracle is rendered into a unitary $O_f$ which acts as 
\begin{equation}
    O_f: |x\rangle \otimes |y\rangle \mapsto |x\rangle \otimes |y \oplus f(x)\rangle
\end{equation} for $|x\rangle$ a state on $n$ qubits and $|y\rangle$ a state on one qubit.  
In the noiseless quantum setting, the best quantum algorithm can find the hidden string $s$ in $\Theta(1)$ queries
~\cite{bernstein1997quantum}.

We prove the following result on the $\NISQ$ complexity of the Bernstein-Vazirani problem: 
\begin{theorem}
\label{thm:BV1}
Let $0 \le \lambda < 1/24$.  Then there is a $\NISQ_\lambda$ algorithm which can solve the Bernstein-Vazirani problem with probability at least $1 - \delta$ using $\mathcal{O}(\frac{1}{1 - 24\lambda}\, \log(n/\delta))$ queries.
\end{theorem}

\subsection{Proof preliminaries}

\noindent Let us establish some notation to be used in the proof.
\begin{definition}[Permutation operators] For $n > 0$, let $S_n$ be the permutation group on $m$ objects.  To each $\pi$ in $S_n$ we associate an operator acting on $(\mathbb{C}^2)^{\otimes n}$ defined by
\begin{equation}
\pi \big(|\psi_{1}\rangle \otimes |\psi_2\rangle \otimes \cdots \otimes |\psi_n\rangle\big) = |\psi_{\pi^{-1}(1)}\rangle \otimes |\psi_{\pi^{-1}(2)}\rangle \otimes \cdots \otimes |\psi_{\pi^{-1}(n)}\rangle\,,\quad \forall \, |\psi_1\rangle,\,|\psi_2\rangle,...,\,|\psi_n\rangle \in \mathbb{C}^2
\end{equation}
which extends by multilinearity to all of $(\mathbb{C}^2)^{\otimes m}$.
\end{definition}
\noindent We have a similar definition for permutations acting on bit strings.
\begin{definition}[Permutations acting on bit strings] Again letting $S_n$ be the permutation group on $n$ objects, to each $\pi$ in $S_n$ we associate a function $\pi : \{0,1\}^n \to \{0,1\}^n$ defined by
\begin{equation}
\pi( s_1 s_2 \cdots s_m) = s_{\pi^{-1}(1)} s_{\pi^{-1}(2)} \cdots s_{\pi^{-1}(m)}\,, \quad \forall s_1 s_2 \cdots s_m \in \{0,1\}^n\,. 
\end{equation}
\end{definition}
\noindent Moreover, if $f : \{0,1\}^n \to \{0,1\}$ is the unknown function in the Bernstein-Vazirani problem, then we define $f_\pi := f \circ \pi$.

\begin{figure}
    \centering
    \includegraphics[width=0.83\textwidth]{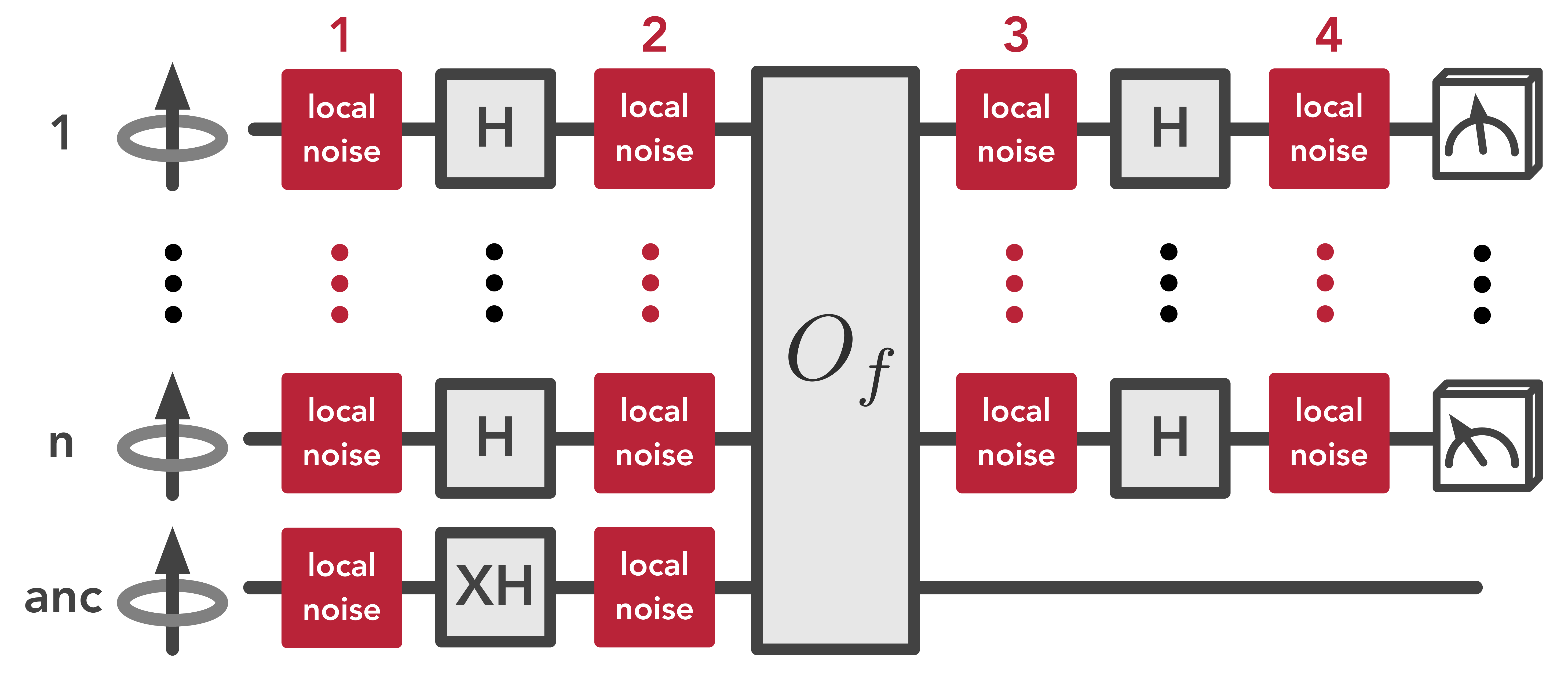}
    \caption{Bernstein-Vazirani algorithm in the presence of arbitrary noise (each box labeled by ``local noise'' denotes that with probability $\lambda$, an arbitrary, adversarially chosen single-qubit operation is applied). We have labeled the layers of noise for ease of reference in the proof.}
    \label{fig:BV}
\end{figure}

\paragraph{Bernstein-Vazirani algorithm.} We conclude this subsection by reviewing how the original Bernstein-Vazirani algorithm~\cite{bernstein1997quantum} works, see Figure~\ref{fig:BV}. One begins by preparing the initial state $|+\rangle^{\otimes n} \otimes |-\rangle$, and then acting on it with the oracle.  In so doing, we obtain the state
\begin{equation}
\frac{1}{2^{n/2}}\sum_{x \in \{0,1\}^n} (-1)^{s \cdot x} |x\rangle \otimes |-\rangle\,.
\end{equation}
Tracing out the $|-\rangle$ qubit, we can then apply the Hadamards $H^{\otimes n}$ to the first $n$ qubits to obtain
\begin{equation}
\frac{1}{2^n}\sum_{x,y \in \{0,1\}^n} (-1)^{(s + y)\cdot x} |y\rangle = |s\rangle
\end{equation}
which gives us the hidden string $s$.

\subsection{Proof of Theorem~\ref{thm:BV1}}

With these notations at hand, we are ready to prove Theorem~\ref{thm:BV1}. Suppose we (noisily) initialize in the state $|+\rangle^{\otimes n} \otimes |-\rangle$.  We call the last qubit the \emph{ancilla qubit}.  
There are two cases depending on whether or not the ancilla qubit is corrupted prior to the oracle application in the algorithm. We handle these two cases separately in the following two lemmas.

\begin{lemma}\label{lem:not_decohere}
    If the ancilla qubit is not corrupted prior to application of the oracle in the Bernstein-Vazirani algorithm, then for every $i\in[n]$, with probability $(1- \lambda)^4$ the $i$th output bit is given by $s_i$  and otherwise is given by a possibly incorrect bit.
\end{lemma}

\begin{proof}
    When we reach the step of the Bernstein-Vazirani algorithm where we apply the oracle, suppose that $k$ qubits out of the first $n$ have been corrupted in the first two layers of noise (see Figure~\ref{fig:BV}). Further suppose that the qubits are located at $a_1,...,a_k$, where $\{a_1,...,a_k\} \subset \{1,...,n\}$.  Picking some permutation $\pi \in S_n$ such that $\pi(i) = a_i$ for $i = 1,...,k$, we can write the state of the system as
    \begin{equation}
    \frac{1}{2^n} \sum_{\substack{z,z' \in \{0,1\}^k \\ x,x' \in \{0,1\}^{n-k}}} \beta_{z,z'} \, \pi(|z\rangle \langle z'| \otimes |x\rangle \langle x'|)\pi^\dagger \otimes |-\rangle \langle -|\,
    \end{equation}
    for some coefficients $\brc{\beta_{z,z'}}$ satisfying $\sum_z \beta_{z,z} = 1$.
    The rest of the protocol proceeds as follows.  We apply $O_f$ to get
    \begin{equation}
    \frac{1}{2^n} \sum_{\substack{z,z' \in \{0,1\}^k \\ x,x' \in \{0,1\}^{n-k}}} \beta_{z,z'} \, \pi(|z\rangle \langle z'| \otimes |x\rangle \langle x'|)\pi^\dagger \otimes |-\rangle \langle -|\,(-1)^{f_\pi(zx) + f_\pi(z'x')}\,.
    \end{equation}
    Next we trace out the ancilla to find
    \begin{equation}
    \label{E:traced1}
        \frac{1}{2^n} \sum_{\substack{z,z' \in \{0,1\}^k \\ x,x' \in \{0,1\}^{n-k}}} \beta_{z,z'} (-1)^{f_\pi(zx) + f_\pi(z'x')} \, \pi(|z\rangle \langle z'| \otimes |x\rangle \langle x'|)\pi^\dagger \,.
    \end{equation}
    Following this we apply a third layer of local noise, apply $H^{\otimes n}$, and then apply a fourth layer of local noise again.  Suppose that this procedure corrupts any number of the already corrupted qubits at positions $a_1,...,a_k$, as well as $\ell$ qubits at positions different from $a_1,...,a_k$.  Suppose that the positions of these $\ell$ qubits are $a_{k+1},...,a_{k+\ell}$ where $\{a_1,...,a_k, a_{k+1},...,a_{k + \ell}\} \subset \{1,...,n\}$.  Since the local noise and $H^{\otimes n}$ act qubit-wise, if we only want to track the uncorrupted qubits we can do as follows: at the outset we trace out the $k + \ell$ qubits which are to be corrupted, and then we apply $H^{\otimes (n-k-\ell)}$ to the residual qubits.
    
    We implement this procedure presently. Defining $\sigma \in S_n$ by $\sigma(i) = a_i$ for $i = 1,...,k + \ell$, we can rewrite~\eqref{E:traced1} as
    \begin{equation}
    \label{E:traced2}
    \frac{1}{2^n} \sum_{\substack{z,z' \in \{0,1\}^k \\ w,w' \in \{0,1\}^{\ell} \\ y,y' \in \{0,1\}^{n-k-\ell}}} \beta_{z,z',w,w'} (-1)^{f_\sigma(zwy) + f_\sigma(z'w'y')} \, \sigma(|z\rangle \langle z'| \otimes |w\rangle \langle w'|\otimes |y\rangle \langle y'|) \sigma^\dagger\,
    \end{equation}
    for some coefficients $\brc{\beta_{z,z',w,w'}}$ satisfying $\sum_{z,w} \beta_{z,z,w,w} = 1$.
    Letting $\{b_1,...,b_{n-k-\ell}\} = \{1,...,n\}\setminus \{a_1,...,a_{k + \ell}\}$ be the uncorrupted registers, where we choose $b_1 < b_2 < \cdots < b_{n-k-\ell}$\,, tracing out everything but the $b_1,...,b_{n-k-\ell}$ registers we find the residual pure state
    \begin{equation}
    \frac{1}{\sqrt{2^{n-k-\ell}}} \sum_{y \in \{0,1\}^{n-k-\ell}} \, (-1)^{y \cdot [s]_{b_1,...,b_{n-k-\ell}}} \, |y\rangle
    \end{equation}
    where $[s]_{b_1,...,b_{n-k-\ell}}$ denotes the $b_1,...,b_{n-k-\ell}$ bits of the hidden string $s$.
    Applying $H^{\otimes (n-k-\ell)}$ we find
    \begin{equation}
    \frac{1}{2^{n-k-\ell}} \sum_{y,y'\in \{0,1\}^{n-k-\ell}} \, (-1)^{y \cdot ([s]_{b_1,...,b_{n-k-\ell}} + y')} \, |y'\rangle\,.
    \end{equation}
    Finally, measuring in the computational basis, the probability of measuring $|[s]_{b_1,...,b_{n-k-\ell}}\rangle$ is equal to one.
    
    All in all, we have seen that if we perform the usual Bernstein-Vazirani algorithm, we obtain the hidden bit string $s$ but with a fraction of its bits corrupted, corresponding precisely to the qubits that were corrupted by one of the four layers of local noise. 
\end{proof}



\noindent We can now 
conclude the proof of Theorem~\ref{thm:BV1}.

\begin{proof}[Proof of Theorem~\ref{thm:BV1}]
By Lemma~\ref{lem:not_decohere}, we see that for each bit $i$ of the $n$ output bits, the probability of being $\ket{s_i}$ is at least $(1 - \lambda)^6$, which happens if the ancilla qubit is never corrupted in either of the two layers of noise prior to the application of the oracle (with probability $(1 - \lambda)^2$), and if additionally the former of the two possible events in Lemma~\ref{lem:not_decohere} happens (with probability $(1 - \lambda)^4$). Let $f(\lambda) \triangleq (1 - \lambda)^6$ and note that for $\lambda \le 1/10$, $f(\lambda) > 1/2$.

Let $X_i$ be a random variable which equals zero if $s_i$ is obtained correctly with our procedure, and equal to one otherwise.  Letting $Y_i$ be the average of $M$ i.i.d.~copies of $X_i$, then the Chernoff-Hoeffding bound tells us that
\begin{equation}
\text{Prob}\left(Y_i^{(j)} \geq \frac{1}{2}\right) \leq \exp\left(- 2 M \left(\frac{1}{2} - f(\lambda)\right)^2\right) \,.
\end{equation}
This is an upper bound on the probability that if we repeat the Bernstein-Vazirani algorithm $M$ times and employ the majority votes strategy on the $i$th site to determine $s_i$, then we will fail.  The probability that we fail for at least one of the $n$ sites is upper bounded by
\begin{equation}
\text{Prob}\left(\max_i \,Y_i \geq \frac{1}{2}\right) \leq \sum_{i=1}^n \exp\left(- 2 M \left(\frac{1}{2} - f(\lambda)\right)^2\right) \leq n \, \exp\left(- 2 M \left(\frac{1}{2} - f(\lambda)\right)^2\right)\,.
\end{equation}
So if we want the right-hand side to be at most $\delta$, then we can pick some $M$ such that
\begin{equation}
\label{E:Mbound1}
M = O\left(\frac{1}{(1 - 2\,f(\lambda))^2} \, \log(n/\delta)\right)\,.
\end{equation}
Since $(1 - 2f(\lambda))^2$ is an alternating series in $\lambda$, we can lower bound it by its first two terms as
\begin{equation}
(1 - 2 f(\lambda))^2 \geq 1 - 24 \lambda\,.
\end{equation}
Then~\eqref{E:Mbound1} can be written in a slightly simplified form as
\begin{equation}
\label{E:Mbound2}
M = \left(\frac{1}{1 - 24 \lambda} \, \log(n/\delta)\right)
\end{equation}
for $\lambda < 1/24$, as claimed.\footnote{We have not attempted to optimize this threshold for $\lambda$ for general local noise channels or particular choices of local noise channels.  However, we note that with more careful bookkeeping one can show, e.g.~when the local noise is depolarizing noise, that for any $\lambda$ bounded away from $1$ the above algorithm can solve the Bernstein-Vazirani problem with $\mathcal{O}(\log(n/\delta))$ queries.}
\end{proof}

\begin{remark}\label{remark:BV}
As with the proof of Theorem~\ref{thm:superpoly1}, our proof considers a stronger noise model where at every layer, every qubit is independently corrupted with probability $\lambda$ by a (potentially adversarially chosen) single-qubit channel.
While our definition of $\NISQ$ focuses on local depolarizing noise, the quantum advantage in solving the Bernstein-Vazirani problem holds even when the local noise could be adversarially chosen.
\end{remark}

\section{Shadow Tomography}

In this section we show that relative to a natural \emph{quantum oracle}, there is also an exponential separation even between $\NISQ$ and $\BPP^{\QNC^0}$.
The task we consider that witnesses this separation has been studied previously in the context of separations between algorithms with and without quantum memory \cite{huang2021information,chen2021exponential,huang2022quantum} and is based on \emph{shadow tomography}, i.e. predicting properties of an unknown state to which one has access via a state oracle. 
Informally, for an unknown state $\rho$, one would like to predict $|\Tr(Q\rho)|$ for all Pauli operators $Q\in\brc{\Id,X,Y,Z}^{\otimes n}$ given copies of $\rho$. 

We will show that $\NISQ_\lambda$ algorithms with access to such an oracle require $1/(1 - \lambda)^{\Omega(n)}$ copies of $\rho$ to estimate all of these observables to within constant error. On the other hand, existing upper bounds \cite{huang2021information} imply that $\BPP^{\QNC^0}$ algorithms only require $\mathcal{O}(n)$ copies of $\rho$.

\subsection{A quantum oracle preparing an unknown state}

Before stating this separation formally in Theorem~\ref{thm:quantum_oracle} below, we first formalize how to extend the oracle model from Section~\ref{subsec:nisq} to quantum oracles. We consider noisy quantum circuits with two registers: an $n$-qubit $\mathsf{state}$ register for loading in new copies of $\rho$, and a separate workspace register on at most $\poly(n)$ qubits.

\begin{definition}\label{def:quantum_oracle}
    Let $\rho$ be an $n$-qubit state.
    We consider an oracle $O_\rho$ given as a CPTP map that traces out $n$ qubits in the $\mathsf{state}$ register and prepares the state $\rho$ in the $\mathsf{state}$ register:
    \begin{equation}
        O_\rho(\sigma) \triangleq \rho \otimes \Tr_{\mathsf{state}}(\sigma),
    \end{equation}
    for any integer $n' \geq n$ and any $n'$-qubit state $\sigma$.
\end{definition}

\noindent For the purposes of showing a lower bound in this oracle model, we will assume that in between oracle queries, the algorithm can perform arbitrary noiseless quantum computation, and at the end it can perform a noiseless measurement in the computational basis. The only noise that gets applied is local depolarizing noise after any call to $O_{\rho}$. Note that this is a stronger model of computation than a $\lambda$-noisy quantum algorithm or a $\NISQ_\lambda$ algorithm, which merely makes the lower bound we show even stronger. Furthermore, as there is no notion of a classical oracle in this setting, it is not necessary to work with the tree formalism of Definition~\ref{def:tree_nisq}.

\subsection{Exponential lower bound}

We are now ready to state the main oracle separation of this section:

\begin{theorem}\label{thm:quantum_oracle}
    Let $\rho = \frac{1}{2^n}(\Id + s\cdot P)$ for some $s\in\brc{0, 1}$ and $n$-qubit Pauli operator $P$.
    Given access to the oracle $O_{\rho}$ in Definition~\ref{def:quantum_oracle}, no $\NISQ_\lambda$ algorithm can determine either $s$ or $P$ with constant advantage unless it makes $\Omega((1 - \lambda)^{-n})$ oracle queries. 
    
    On the other hand, there is an algorithm in $\BPP^{\QNC^0}$ that, given $\mathcal{O}(n)$ oracle queries, can determine both $s$ and $P$ with high probability.
    In fact, even if $\rho$ is an arbitrary state, there is an algorithm in $\BPP^{\QNC^0}$ that, given $\mathcal{O}(n)$ oracle queries, can estimate $|\Tr(P\rho)|$ for all $n$-qubit Pauli operators $P$ to within small constant error with high probability.
\end{theorem}

\noindent We remark that \cite{huang2022foundations} gives an upper bound of $(1 - \lambda)^{-\Theta(n)}$ for $\NISQ_\lambda$ algorithms, so our exponential lower bound is qualitatively best possible.
For the proof of Theorem~\ref{thm:quantum_oracle} we consider the output state of any noisy quantum circuit and argue that for any Pauli operator $P$, the distance between the output state when the unknown state $\rho$ is given by $\frac{1}{2^n}(\Id + P)$ (likewise when it is given by $\frac{1}{2^n}(\Id - P)$) versus when it is given by $\frac{\Id}{2^n}$ is exponentially small unless if the circuit makes exponentially many oracle queries:

\begin{lemma}\label{lem:hybrid_pauli}
    For any $\lambda > 0$.
    Let $P\in\brc{\Id,X,Y,Z}^{\otimes n}$. Any $\NISQ_\lambda$ algorithm that has oracle access to the state oracle $O_\rho$ for either $\rho = \frac{1}{2^n}(\Id + P)$ or $\rho = \frac{\Id}{2^n}$ and can distinguish which oracle it has access to with at least $2/3$ probability must make $\Omega((1 - \lambda)^{-|P|})$ queries, where $|P|$ denotes the number of non-identity components in $P$.
\end{lemma}

\begin{proof}
    Suppose the circuit operates on $n'$ qubits.
    For convenience, for $s\in\brc{0,1}$ denote $O_{\frac{1}{2^n}(\Id + s\cdot P)}$ by $O_s$.
    We would like to show that for all $n'$-qubit states $\sigma$, $\norm{D_\lambda^{\otimes n'}[(O_1 - O_0)[\sigma]]}_\tr$ is small so that we can apply Lemma~\ref{lem:hybrid_basic}.
    But note that for any $Q\in\brc{X,Y,Z}$, $D_\lambda[Q] = (1 - \lambda)Q$, whereas $D_\lambda[\Id] = \Id$. We conclude that
    \begin{equation}
        \norm{D^{\otimes n'}_\lambda[(O_1 - O_0)[\sigma]]}_\tr  = \norm*{D^{\otimes n'}_{\lambda}\!\!\left[\frac{P}{2^{n}}\otimes \Tr_{\mathsf{state}}(\sigma)\right]}_\tr \le \norm*{D^{\otimes n}_\lambda\!\!\left[\frac{P}{2^{n}}\right]}_\tr = (1 - \lambda)^{|P|}.
    \end{equation}
    By taking the channels $\calE$ in Lemma~\ref{lem:hybrid_basic} to be $D^{\otimes n'}_\lambda \circ O_1$ and $D^{\otimes n'}_\lambda \circ O_0$, we conclude that no algorithm given by alternately querying the oracle followed by depolarizing noise, and running arbitrary noiseless quantum computation, and finally measuring in the computational basis can distinguish whether the underlying oracle is $O_+$ or $O_-$ with at least $2/3$ probability unless it makes $\Omega((1 - \lambda)^{-|P|})$ queries.
    
    As this model is a stronger model of computation than $\NISQ_\lambda$ (note that there is no notion of a classical oracle in this setting), this implies the claimed lower bound for $\NISQ_\lambda$.
\end{proof}

\noindent Theorem~\ref{thm:quantum_oracle} follows immediately from Lemma~\ref{lem:hybrid_pauli}:

\begin{proof}[Proof of Theorem~\ref{thm:quantum_oracle}]
    The second part of the theorem was shown in \cite[Theorem 2]{huang2021information}. The $\NISQ_\lambda$ part of the theorem follows by applying Lemma~\ref{lem:hybrid_pauli} with $P$ ranging over all $n$-qubit Pauli operators which act nontrivially on all qubits. The Lemma implies that regardless of which $P$ is the one defining the oracle, the algorithm is unable to distinguish whether it has access to $O_\rho$ or $O_{I/2^n}$ without making $\Omega((1-\lambda)^{-n})$ queries.
\end{proof}

\section{Separating \texorpdfstring{$\NISQ$}{NISQ} and \texorpdfstring{$\BPP^{\QNC}$}{BPPQNC} from \texorpdfstring{$\BQP$}{BQP}}
\label{app:shuffle}

While the relatively simple oracle from Section~\ref{subsec:lift_lower} allowed us to separate $\NISQ$ and $\BQP$, it is insufficient even to separate $\BPP^{\QNC^0}$ from $\BQP$, as Simon's algorithm can also be implemented in the former. Here we show how to simultaneously separate $\NISQ$ and $\BPP^{\QNC}$ from $\BQP$, at the cost of relying on a more involved oracle construction, namely the shuffling Simon's problem from \cite{chia2020need}.

We first recall the setup of the shuffling Simon's problem.

\begin{definition}[Shuffling]\label{def:shuffling}
    Given a function $f: \brc{0,1}^n\to\brc{0,1}^n$, consider any sequence of functions $f_0,\ldots,f_d: \brc{0,1}^{(d+2)n}\to \brc{0,1}^{(d+2)n}$ where $f_0,\ldots,f_{d-1}$ are 1-to-1 functions, and where $f_d$ is chosen as follows. Define $S_d\subset\brc{0,1}^{(d+2)n}$ to be the set of all $f_{d-1}\circ\cdots\circ f_0(x)$ for $x$ ranging over the lexicographically first $2^n$ bitstrings in $\brc{0,1}^{(d+2)n}$. Then define
    \begin{equation}
        f_d(x) = \begin{cases} 
            f((f_{d-1}\circ\cdots \circ f_0)^{-1}(x)) & \text{if} \ x\in S_d \\
            \vec{0} & \text{otherwise}
        \end{cases},
    \end{equation}
    where we identify the lexicographically first $2^n$ bitstrings in $\brc{0,1}^{(d+2)n}$ with $\brc{0,1}^n$ in the natural way.\footnote{In \cite{chia2020need}, they define $f_d(x) = \perp$ for all $x\not\in S_d$ for a special output symbol $\perp$. The shuffling Simon's problem under our definition is strictly more difficult, so the lower bound they prove immediately translates to a lower bound in our setting.} 
    Let $\mathbf{SHUF}(f,d)$ denote the set of all sequences $(f_0,\ldots,f_d)$ that can be obtained in this way, and let $\calD(f,d)$ denote the distribution over $\mathbf{SHUF}(f,d)$ induced by sampling $f_0,\ldots,f_{d-1}$ uniformly at random from the set of 1-to-1 functions.
\end{definition}

\noindent We now describe the quantum oracle associated to $\mathbf{SHUF}(f,d)$.

\begin{definition}[Quantum shuffling oracle]
    Given $(f_0,\ldots,f_d)\in\mathbf{SHUF}(f,d)$, we define the associated quantum channel $\calF$ as follows. Given input state
    \begin{equation}
        \ket{\phi} = \bigotimes^d_{i=0} \ket{i,x_i}\ket{y_i},
    \end{equation}
    where $x_0,\ldots,x_d\in\brc{0,1}^{(d+2)n}$,
    we define
    \begin{equation}
        \calF\ket{\phi} \triangleq \bigotimes^d_{i=0} \ket{i,x_i}\ket{y_i \oplus f_i(x_i)}.
    \end{equation}
    The \emph{shuffling oracle} $O_{f,d}$ then sends $\ket{\phi}$ to the mixed state
    \begin{equation}
        O_{f,d}(\ket{\phi}\bra{\phi}) \triangleq \E[\calF\sim\calD(f,d)]*{\calF|\phi\rangle \langle \phi|\calF}.
    \end{equation}
\end{definition}

\noindent We are now ready to describe the shuffling Simon's problem.

\begin{definition}[$d$-Shuffling Simon's Problem: $\SSP{d}$]
    Let $f: \brc{0,1}^n\to\brc{0,1}^n$ be a random 2-to-1 function with probability 1/2, and a random 1-to-1 function otherwise. Given quantum access to the oracle $O_{f,d}$ as well as classical oracle access to $f$, the problem is to decide which of these two cases we are in.
\end{definition}

\noindent The lower bound against $\BPP^{\QNC}$ is immediate from \cite{chia2020need}:

\begin{theorem}[Corollary 1.3 from \cite{chia2020need}]
    For $d = \log^{\omega(1)}(n)$, $\SSP{d}\in\BQP^O$ but $\SSP{d}\not\in(\BPP^{\QNC})^{O}$, where $O$ denotes the shuffling oracle $O_{f,d}$ of the underlying function $f$.
\end{theorem}

\noindent The main result of this section is to show that $\SSP{d}$ also yields an oracle separation between $\BQP$ and $\NISQ_\lambda$.

\begin{theorem}\label{thm:nisq_ssp}
    Any $\NISQ_\lambda$ algorithm that solves $\SSP{d}$ with constant advantage must make $\exp(\Omega(n\cdot \min(1/2,\lambda d)))$ oracle queries.
\end{theorem}

\noindent The lower bound against $\NISQ$ follows along similar lines to the argument in Section~\ref{subsec:lift_lower}. The primary difference is that the subset $\Omega$ from Section~\ref{subsec:lift_lower} was chosen to consist of all strings where the last $n$ bits are $0$, whereas the subset $S_d$ of the domain on which $f_d$ acts nontrivially in Definition~\ref{def:shuffling} is a random subset. We will argue that with high probability over the randomness of this subset, depolarizing noise will kill off most of a states' component within the subspace corresponding to this subset.

\noindent We begin with a simple lemma showing that random small subsets of the hypercube are well-separated. This lemma is merely a re-interpretation of the Gilbert-Varshamov bound.

\begin{lemma}\label{lem:sepwhp}
    Let $\Omega$ be a uniformly random subset of $\brc{0,1}^M$ of size $S$, where $|S| \le 2^{M-1}$. Then the probability that there exist distinct $x,y\in\Omega$ that are $\frac{M}{2}(1 - \sqrt{2\log_2(S^2/\delta)/M})$-close in Hamming distance is at most $\delta$.
\end{lemma}

\begin{proof}
    Given $x\in\brc{0,1}^M$ and $0 \le r\le M$, let $B(x,r)$ denote the Hamming ball around $x$ of radius $r$. Note that 
    \begin{equation}
        |B(x,r)| = \sum^r_{i=0} \binom{M}{i} \le 2^{M\cdot H(r/M)},
    \end{equation}
    where $H(\cdot)$ denotes the binary entropy function. So given a subset $\Omega'\subset\brc{0,1}^M$ of size $S'$, the probability that a random point from $\brc{0,1}^M \backslash \Omega'$ is at least $r$-far in Hamming distance from every point of $\Omega'$ is 
    \begin{equation}
        1 - \frac{2^{M\cdot H(r/M)}\cdot S'}{2^M - S'}.
    \end{equation}
    As we can sample a random subset $\Omega$ of $\{0,1\}^M$ of size $S$ by randomly sampling $S$ points from $\{0,1\}^M$ without replacement, we conclude that the probability that these points are all $r$-far in Hamming distance is
    \begin{align}
        \prod^{S-1}_{t=0}\left(1 - \frac{2^{M\cdot H(r/M)}\cdot t}{2^M - t}\right) &\ge 1 - 2^{M\cdot H(r/M)}\sum^{S-1}_{t=0} \frac{t}{2^M - t} \ge 1 - S^2 \cdot 2^{- M(1 - H(r/M))} \\
        &\ge 1 - S^2 \cdot 2^{-M(1 - 2\sqrt{r/M - r^2/M^2})}.
    \end{align}
    By taking $r = \frac{M}{2}(1 - \sqrt{2\log_2(S^2/\delta)/M})$, we find that the probability that all points in $\Omega$ are at least $r$ far in Hamming distance is at least $1 - \delta$ as claimed.
\end{proof}

\noindent We would like to apply Lemma~\ref{lem:proj} to $\Omega$ consisting of strings which are separated in Hamming distance. To that end, we prove the following expansion result for such subsets of the hypercube.

\begin{lemma}\label{lem:classical_anti}
    Let $3/10 \le q < 1/2$. Let $\Omega\subset\brc{0,1}^n$ be a subset such that all strings are at least $qn$ apart in Hamming distance. Then for any distribution $D$ over $\brc{0,1}^n$ and any $0 < \lambda \le 1$, if $\tilde{a}$ is obtained by sampling $a\sim D$ and independently flipping each bit of $a$ with probability $\lambda/2$, then $\Pr{\tilde{a} \in \Omega} \le \exp(-\Omega(\lambda n))$.
\end{lemma}

\begin{proof}
    Consider any fixed string $a\in\brc{0,1}^n$. If the Hamming distance from $a$ to any element of $\Omega$ is at least $qn/2$, then conditioned on $a$, the probability that $\tilde{a}\in\Omega$ is at most $(\lambda/2)^{qn/2} (1 - \lambda/2)^{n(1-q/2)}\cdot |\Omega| \le 2^{-qn/2}\cdot |\Omega|$. On the other hand, if the Hamming distance from $a$ to some element $x\in\Omega$ is some $d\le qn/2$, then by triangle inequality its Hamming distance to any other element of $\Omega$ is at least $qn - d \ge qn/2$. Conditioned on such an $a$, the probability that $\tilde{a}\in\Omega$ can be bounded by
    \begin{equation}
        \Pr{\tilde{a} = x \mid a} + (\lambda/2)^{qn/2} (1 - \lambda/2)^{n(1-q/2)} \cdot (|\Omega - 1|) \le (1 - \lambda/2)^n + 2^{-qn/2}\cdot |\Omega|.
    \end{equation}
    Putting everything together, we conclude that
    \begin{equation}
        \Pr{\tilde{a}\in\Omega} \le \exp\left(n\cdot \brk*{2(1/2 - q)^2 - (qn/2)\ln 2}\right) + \exp(-\lambda n / 2).
    \end{equation}
    Note that if $q \ge 0.28$, this expression is at most $\exp(-Cn)$ for an absolute constant $C > 0$.
\end{proof}

\noindent Combining Lemma~\ref{lem:proj} and Lemma~\ref{lem:classical_anti} immediately yields the following estimate:

\begin{corollary}\label{lem:offspan}
    Let $3/10 \le q < 1/2$. Given $\Omega\subset\brc{0,1}^n$ such that all strings in $\Omega$ are at least $qn$ apart in Hamming distance, let $\Pi$ denote the projection to the span of $\brc{\ket{x}}_{x\in\Omega}$. Then for any $0 < \lambda \le 1$ and any $n$-qubit pure state $\ket{\psi}$,
    \begin{equation}
        \Tr\left(\Pi D_{\lambda}[|\psi\rangle \langle \psi|]\right) \le \exp(-\Omega(\lambda n)).\label{eq:projection}
    \end{equation}
\end{corollary}

\noindent We now show the analogue of Lemma~\ref{lem:test_slowdown} for $\SSP{d}$:

\begin{lemma}\label{lem:test_ssp_slowdown}
    Let $A$ be any $\lambda$-noisy quantum circuit which makes $N$ oracle queries to $O_{f,d}$ for some $f: \brc{0,1}^n\to\brc{0,1}^n$. If $p_{f,d}$ is the distribution over the random string $s$ output by the circuit when the oracle is $O_{f,d}$, then $\tvd(p_{f,d}, p_{f',d}) \le \exp(-\Omega(\lambda dn))$ for any $f,f': \brc{0,1}^n\to\brc{0,1}^n$.
\end{lemma}

\begin{proof}
    For convenience, define $n' \triangleq (d+2)n$, and suppose that $A$ operates on $n''$ qubits. Let $f: \brc{0,1}^n\to\brc{0,1}^n$ be an arbitrary Boolean function. We would like to show that for all $n''$-qubit pure states $\sigma$, $\norm{(O_{f,d}\otimes \mathrm{Id} - \mathrm{Id})[D^{\otimes n''}_\lambda[\sigma]]}_\tr$ is small so that we can apply Lemma~\ref{lem:hybrid_basic}.

    Given $\calF\in\mathbf{SHUF}(f,d)$, let $\Pi_\calF$ denote the projection to the span of $\brc{\ket{x}}_{x\in S_d}$ for the set $S_d$ given by $\calF$; we will occasionally denote $S_d$ by $S_d(\calF)$ to emphasize the dependence on $\calF$. For $\calF\sim \calD(f,d)$, observe that $S_d$ is a uniformly random subset of $\{0,1\}^{n'}$ of size $2^n$, so by Lemma~\ref{lem:sepwhp}, we have with probability at least $1 - 2^{-cdn}$ that all strings in $S_d$ are at least Hamming distance $\Delta \triangleq \frac{n'}{2}\big(1 - \frac{(4 + 2cd)n}{n'}\big)^2$ apart. By taking $c$ a sufficiently small absolute constant, we note that $\Delta \ge 3n'/10$ provided $d$ is larger than some absolute constant.
    
    Let $E\subseteq \mathbf{SHUF}(f,d)$ denote the set of $\calF$ for which this is the case, so that $|E|/|\mathbf{SHUF}(f,d)| \ge 1 - 2^{-c}$. By Lemma~\ref{lem:offspan} (applied to $n'$-qubit states instead of $n$-qubit states), $\Tr(\Pi_\calF  D^{\otimes n'}_\lambda[\rho]) \le \exp(-\lambda n'/2) + \exp(-Cn')$ for an absolute constant $C > 0$ for all $\calF\in E$ and $n'$-qubit pure states $\rho$.
    
    If $D_\lambda^{\otimes n''}[\sigma] = \sum_j \lambda_j \ket{\phi_j}\ket{\phi_j}$ for
    \begin{equation}
        \ket{\phi_j} = \sum_{\vec{x},\vec{y}} v_{j,\vec{x},\vec{y}} \left(\bigotimes^d_{i=0} \ket{i,x_i}\ket{y_i}\right) \otimes \ket{w_{j,\vec{x},\vec{y}}},
    \end{equation}
    then for any $\calF\in\mathbf{SHUF}(f,d)$,
    \begin{equation}
        \sum_j \lambda_j \sum_{\vec{x}: x_d\in S_d(\calF), \vec{y}} v^2_{j,\vec{x},\vec{y}} \le \exp(-\lambda n'/2) + \exp(-Cn').
    \end{equation}
    We see that $O_{f,d}\otimes\mathrm{Id}$ maps $\ket{\phi_j}\bra{\phi_j}$ to
    \begin{multline}
        \mathbb{E}_{\calF\sim\calD(f,d)}\Bigg[\sum_{\vec{x},\vec{x}',\vec{y},\vec{y}'} v_{j,\vec{x},\vec{y}} v_{j,\vec{x}',\vec{y}'}\bigotimes^{d-1}_{i=0} \ket{x_i,y_i\oplus f_i(x_i)}\bra{x'_i,y'_i\oplus f_i(x_i)} \\
        \otimes \left(\ket{x_d,y_d\oplus f_d(x_d)} \cdot \bone{x_d\in S_d(\calF)} + \ket{x_d,y_d}\cdot \bone{x_d \not\in S_d(\calF)}\right) \\ 
        \left(\bra{x'_d,y'_d\oplus f_d(x'_d)} \cdot \bone{x'_d\in S_d(\calF)} + \bra{x'_d,y'_d}\cdot \bone{x'_d \not\in S_d(\calF)}\right) \otimes \ket{w_{j,\vec{x},\vec{y}}}\bra{w_{j,\vec{x}',\vec{y}'}} \label{eq:Ofd}
        \Bigg]
    \end{multline}
    Noting that the only dependence on $f$ in the above expression is in the definition of $f_d$, we see that for any $f,f': \brc{0,1}^n\to\brc{0,1}^n$,
    \begin{multline}
        (O_{f,d}\otimes\mathrm{Id} - O_{f',d}\otimes\mathrm{Id})[\ket{\phi_j}\bra{\phi_j}] = \mathbb{E}_{\calF\sim\calD(f,d)}\Bigg[\sum_{\vec{x},\vec{x}',\vec{y},\vec{y}'} v_{j,\vec{x},\vec{y}} v_{j,\vec{x}',\vec{y}'}\bigotimes^{d-1}_{i=0} \ket{x_i,y_i\oplus f_i(x_i)}\bra{x'_i,y'_i\oplus f_i(x_i)} \\
        \otimes \Big\{\bra{x_d,y_d\oplus f_d(x_d)}\ket{x'_d,y'_d\oplus f'_d(x'_d)}\cdot \bone{x_d,x'_d\in S_d(\calF)} \\
        + \bra{x_d,y_d\oplus f_d(x_d)}\ket{x'_d,y'_d}\cdot \bone{x_d\in S_d(\calF), x'_d\not\in S_d(\calF)} \\
        + \bra{x_d,y_d}\ket{x'_d,y'_d\oplus f'_d(x'_d)}\cdot \bone{x_d\not\in S_d(\calF), x'_d\in S_d(\calF)}\Big\} \otimes \ket{w_{j,\vec{x},\vec{y}}}\bra{w_{j,\vec{x}',\vec{y}'}} \Bigg]
    \end{multline}
    where we define $f'_d(x) = f'((f_{d-1}\circ\cdots\circ f_0)^{-1}(x))$ for $x\in S_d(\calF)$ and $f'_d(x) = \vec{0}$ otherwise.
    This is a mixed state of rank at most 2, so its trace norm is at most $\sqrt{2}$ times its Frobenius norm, which we can bound via triangle inequality by
    \begin{equation}
        \mathbb{E}_{\calF\sim\calD(f,d)}\Bigg[2\sqrt{\sum_{\substack{\vec{x},\vec{x}',\vec{y},\vec{y}': \\ x_d\in\Omega \ \text{or} \ x'_d\in\Omega}} v^2_{j,\vec{x},\vec{y}} v^2_{j,\vec{x}',\vec{y}'}} \le 2\sqrt{1 - \Bigg(1 - \sum_{\vec{x}: x_d\in S_d, \vec{y}} v^2_{j,\vec{x},\vec{y}}\Bigg)^2} \le 2\sqrt{2\sum_{\vec{x}: x_d\in S_d, \vec{y}} v^2_{j,\vec{x},\vec{y}}}\Bigg].
    \end{equation}
    Letting $j\sim\lambda$ denote sampling from the distribution which places mass $\lambda_j$ on index $j$, we conclude that 
    \begin{align}
        \norm{(O_{f,d}\otimes\mathrm{Id} - O_{f',d}\otimes\mathrm{Id})[D^{\otimes n''}_\lambda[\sigma]]}_\tr &\le 4\mathbb{E}_{j\sim\lambda, \calF\sim\calD(f,d)}\Bigg[\sqrt{\sum_{\vec{x}: x_d\in S_d(\calF),\vec{y}} v^2_{j,\vec{x},\vec{y}}}\Bigg] \\
        &\le 4\mathbb{E}_{\calF\sim\calD(f,d)}\Bigg[\sqrt{\mathbb{E}_{j\sim\lambda}\Bigg[\sum_{\vec{x}: x_d\in S_d(\calF),\vec{y}} v^2_{j,\vec{x},\vec{y}}}\Bigg]\Bigg]  \\
        &\le 4(\exp(-\lambda n'/4) + \exp(-Cn'/2) + 2^{-cdn}) = \exp(-\Omega(\lambda dn)).
    \end{align}
    The lemma follows by taking the channels $\calE$ in Lemma~\ref{lem:hybrid_basic} to be $\E[f \ \text{2-to-1}]{(O_{f,d}\otimes\mathrm{Id})\circ D^{\otimes n''}_\lambda}$ and $\E[f \ \text{1-to-1}]{(O_{f,d}\otimes\mathrm{Id})\circ D^{\otimes n''}_\lambda}$.
\end{proof}

\begin{proof}[Proof of Theorem~\ref{thm:nisq_ssp}]
    Let $\calT$ be the learning tree corresponding to a $\NISQ_\lambda$ algorithm that makes at most $N$ classical queries to $f$ or quantum oracle queries to $O_{f,d}$. By Lemma~\ref{lem:perturbchildren} and Lemma~\ref{lem:test_ssp_slowdown}, if we replace every noisy quantum circuit $A$ in the tree with a noisy quantum circuit $A'$ that makes queries to $O_{g,d}$ for $g: \brc{0,1}^n\to\brc{0,1}^n$ the identity function, then the new distribution over the leaves of $\calT$ is at most $N^2\exp(-\Omega(\lambda dn))$-far in total variation from the original distribution $p_{O_{f,d}}$; for $N = \exp(o(\lambda d n))$, this quantity is $o(1)$. For convenience, denote this new distribution $p'_f$. 

    To apply Lemma~\ref{lem:lecam}, we wish to bound $\tvd(\E[f \ \text{1-to-1}]{p'_f}, \E[f \ \text{2-to-1}]{p'_f})$. But note that because the noisy quantum circuits $A'$ in the new learning tree are independent of the underlying function $f$, the learning tree is simply implementing a randomized classical query algorithm. As in the proof of Theorem~\ref{thm:generic_slowdown}, if $p'^r_f$ denotes the distribution $p'_f$ conditioned on some internal randomness $r$ of the algorithm, it suffices to bound $\sup_r \tvd(\E[f \ \text{1-to-1}]{p'_f}, \E[f \ \text{2-to-1}]{p'_f})$. The bound for this can be proven in an identical fashion as in Theorem~\ref{thm:generic_slowdown}, so we conclude that this quantity is $o(1)$ if there are $o(2^{n/2})$ classical queries along any root-to-leaf path of $\calT$.
\end{proof}

\section{Deferred Proofs}
\label{app:defer}

\subsection{Proof of Lemma~\ref{lem:info}}
\label{app:defer_leminfo}

Given $S\subseteq[n]$ and mixed $n$-qubit state $\sigma$, let $\sigma|_S$ denote the restriction of $\sigma$ to the subsystem indexed by $S$. That is, 
\begin{equation}
    \sigma_{!S} = \left(\Tr(\sigma|_{S^c})\cdot \Id/2^{|S^c|}\right) \otimes \sigma|_S. \label{eq:decohere}
\end{equation}

\begin{proof}
    Von Neumann entropy is invariant under unitary transformation, so it suffices to show that for any mixed state $\sigma$, $I(D_\lambda[\sigma]) \le (1 - \lambda)\cdot I(\sigma)$. Because $D_\lambda[\sigma] = \E[S\sim\mu]{\sigma_{!S}}$ for $\mu$ the distribution over $S\subseteq[n]$ which includes each element of $[n]$ independently with probability $\lambda$. So by concavity of entropy, additivity of entropy for tensor products, and \eqref{eq:decohere},
    \begin{equation}
        I(D_\lambda[\sigma]) \le \sum^n_{k=0} \lambda^{n-k} (1 - \lambda)^k \sum_{S\subset[n]: |S| = k} I(\sigma|_S) \le \sum^n_{k=0} \binom{n}{k}\frac{k}{n} \lambda^{n-k}(1-\lambda)^k = (1 - \lambda)I(\sigma),
    \end{equation}
    where in the last step we used Lemma~\ref{lem:avg} below.
\end{proof}

\noindent The proof above uses the following fact:

\begin{lemma}[Lemma 7 from \cite{aharonov1996limitations}]\label{lem:avg}
    For any density matrix $\sigma$ on $n$ qubits and any $0\le k < n$,
    \begin{equation}
        \binom{n}{k}^{-1}\sum_{S\subset[n]: |S| = k} I(\sigma|_S) \le \frac{k}{n} I(\sigma).
    \end{equation}
\end{lemma}

\subsection{Proof of Lemma~\ref{lem:proj}}
\label{app:defer_lemproj}

\begin{proof}
    Given $S\subseteq[n]$ and $z\in\brc{0,1}^{|S|}, y\in\brc{0,1}^{n-|S|}$, let $z\circ_S y\in\brc{0,1}^n$ denote the string whose $i$-th entry is $z_i$ if $i\in S$ and $y_i$ otherwise.

    Let $\ket{\psi} = \sum_{x\in\brc{0,1}^{n'}} c_x\ket{x}$, and for convenience let $\rho_{\psi} \triangleq D^{\otimes n'}_\lambda[|\psi\rangle \langle \psi|]$. For any $S\subseteq[n']$, we have
    \begin{equation}
        \Tr_S(\ketbra{\psi}{\psi}) = \sum_{y,y'\in\brc{0,1}^{2n-|S|}} \Bigg(\sum_{w\in\brc{0,1}^{|S|}} c_{w\circ_S y} c_{w\circ_S y'}\Bigg) \ketbra{y}{y'}.
    \end{equation}
    For convenience, denote the coefficient $\sum_w c_{w\circ_S y} c_{w\circ_S y'}$ by $C^S_{y,y'}$. If $S$ is sampled by including every $i\in[n']$ independently with probability $\lambda/2$, then
    \begin{equation}
        \rho_\psi = \mathbb{E}_S\Bigl[\sum_{z\in\brc{0,1}^{|S|}, y,y'\in\brc{0,1}^{n'-|S|}} \frac{1}{2^{|S|}} \ketbra{z}{z} \otimes_S C^S_{y,y'} \ketbra{y}{y'}\Bigr].
    \end{equation}
    Note that for any $x\in\Omega$,
    \begin{align}
         \bra{x}\rho_\psi\ket{x} &= \mathbb{E}_{S\sim\mu}\Bigl[\frac{1}{2^{|S|}}\sum_{z,y} C^S_{y,y}\cdot \bone{z\circ y = x}\Bigr] \\
         &= \mathbb{E}_{S\sim\mu}\Bigl[\frac{1}{2^{|S|}}C^S_{x_{[n']\backslash S}, x_{[n']\backslash S}}\Bigr] \\
         &= \mathbb{E}_{S\sim\mu}\Bigl[\frac{1}{2^{|S|}}\sum_{w\in\brc{0,1}^{|S|}} c^2_{w\circ_S x_{[n]\backslash S}}\Bigr].
    \end{align}
    Note that this expression is precisely $\Pr[a\sim D,\tilde{a}]{\tilde{a} = x}$, so $\Tr(\Pi\rho_\psi)$ is the probability that $\tilde{a}\in\Omega$.
\end{proof}








\section{Simon's Problem and Noisy Parity}
\label{app:bkw}

In this section we describe a path towards separating $\NISQ$ and $\BPP$ using the original Simon's problem, i.e. without modifications based on quantum error-correcting codes like in Section~\ref{sec:robust_simon}. We begin by recalling the problem of learning parity with noise and a classical result on the computational complexity of this problem.

\begin{definition}
    For $\eta \in [0,1/2)$ and $n\in\mathbb{N}$, an instance of \emph{noisy parity with noise rate $\eta$} is given by an unknown string $s\in\brc{0,1}^n$ and samples $(x_1,y_1),\ldots,(x_N,y_N)$ where each $x_i$ is sampled uniformly from $\brc{0,1}^n$ and each $y_i$ is independently either $\iprod{x_i,s} \ \mathrm{mod} \ 2$ with probability $1 - \eta$ or $1 - \iprod{x_i,s} \ \mathrm{mod} \ 2$ otherwise. We say that an algorithm solves noisy parity with noise rate $\eta$ with $n$ samples if, given such an instance, it recovers $S$ with probability at least $2/3$ over its internal randomness and the randomness of the samples.
\end{definition}

\noindent While is widely conjectured (see e.g. \cite{blum1993cryptographic}) that this problem is computationally hard, the following seminal result implies that it is possible to do somewhat better than brute force:

\begin{theorem}[Theorem 2 and Section 3.3 from~\cite{blum2003noise}]\label{thm:bkw}
    If the unknown string $s$ in noisy parity is promised to be a subset of the first $\mathcal{O}(\log n \cdot \log \log n)$ bits, then there is an algorithm for noisy parity with noise rate $\eta$ which has runtime and sample complexity $\poly(n)\cdot (1 - 2\eta)^{-\sqrt{\log n}}$.
\end{theorem}

\noindent Now consider Simon's problem on $n$ qubits, with the promise that the unknown period $s\in\brc{0,1}^n$ is supported on the first $k \triangleq \mathcal{O}(\log n \cdot \log \log n)$ bits and has Hamming weight at most $\log n$. Note that because there are $\binom{\log n \cdot \log\log n}{\log n}$ such periods, the complexity of this special case of Simon's problem for classical algorithms is super-polynomial in $n$.

We now sketch an argument showing that under a plausible strengthening of Theorem~\ref{thm:bkw}, there is a $\NISQ$ algorithm that solves this special case of Simon's efficiently. Note that if one runs Simon's algorithm restricted to these qubits using a $\lambda$-noisy quantum circuit with $\lambda \in (0,1)$ an absolute constant, then with probability $(1 - \lambda)^{\mathcal{O}(\log n)} \ge 1/\poly(n)$ none of the qubits $i\in[n]$ for which $s_i \neq 0$ will be decohered over the course of Simon's algorithm, in which case the classical string $x\in\brc{0,1}^n$ satisfies $\iprod{x,s}\equiv 0 \ \mathrm{mod} \ 2$. Otherwise, if one of the first $k$ qubits gets decohered at any point, the classical string $x$ satisfies $\iprod{x,s}\equiv 0 \ \mathrm{mod} \ 2$ with probability $1/2$. In particular, if one runs this noisy quantum circuit $N$ times, the result is an instance of noisy parity with noise rate $\eta$ satisfying $1 - 2\eta \ge 1 / \poly(n)$. In particular, if one could improve the dependence on $\eta$ in Theorem~\ref{thm:bkw} from $(1 - 2\eta)^{-\sqrt{\log n}}$ to $\poly(1/(1 - 2\eta))$, then we would find that $\NISQ$ algorithms can efficiently solve this special case of Simon's problem. Formally, we have concluded the following:

\begin{theorem}
    Suppose there is an algorithm for noisy parity with noise rate $\eta$, where the unknown string $s$ is promised to be a subset of the first $\mathcal{O}(\log n \cdot \log \log n)$ bits, with runtime $\poly(n, (1 - 2\eta)^{-1})$. Then for $O$ the oracle for Simon's problem under the promise that the unknown period is supported on the first $\mathcal{O}(\log n \cdot \log\log n)$ bits and has Hamming weight at most $\log n$, $\BPP^O \subsetneq \NISQ^O$.
\end{theorem}

\end{document}